%% file: bernoulli_final.tex
\documentclass[preprint]{imsart2}

\usepackage{amsmath,amsfonts,amssymb,amsthm,floatrow,caption,subcaption}
\usepackage[colorlinks,citecolor=blue,urlcolor=blue]{hyperref}
\usepackage{tikz}
\usetikzlibrary{positioning}
\input xy
\xyoption{all}

\startlocaldefs
\newtheorem{thm}{Theorem}[section]
\newtheorem{lem}[thm]{Lemma}
\newtheorem{prop}[thm]{Proposition}

\theoremstyle{definition}
\newtheorem{defn}[thm]{Definition}

\newtheorem{exmp}[thm]{Example}

\theoremstyle{remark}
\newtheorem{rem}[thm]{Remark}

\newcommand{\QED}{\ifhmode\unskip\nobreak\fi\quad {\rm Q.E.D.}} % QED

\newcommand{\cP}{\mathcal{P}}

\newcommand{\R}{\mathbb{R}}

\newcommand{\cV}{\mathcal{V}}

\newcommand{\RLCT}{{\rm RLCT}}

\DeclareMathOperator{\codim}{codim}

\numberwithin{equation}{section}

\endlocaldefs

\begin{document}

\begin{frontmatter}

% "Title of the Paper"
\title{Marginal likelihood and model selection for Gaussian latent tree and forest models}

\runtitle{Gaussian latent tree and forest models}

\begin{aug}

\author{\fnms{Mathias} \snm{Drton}\thanksref{a}\ead[label=e1]{md5@uw.edu}}
\author{\fnms{Shaowei} \snm{Lin}\thanksref{b}\ead[label=e2]{lins@i2r.a-star.edu.sg}}
\author{\fnms{Luca} \snm{Weihs}\thanksref{a}\ead[label=e3]{lucaw@uw.edu}}
\and
\author{\fnms{Piotr} \snm{Zwiernik}\corref{}\thanksref{c}\ead[label=e4]{piotr.zwiernik@upf.edu}}
\address[a]{Department of Statistics, University of Washington, Seattle, WA, U.S.A. \printead{e1,e3}}
\address[b]{Institute for Infocomm Research, Singapore. \printead{e2}}
\address[c]{Department of Economics and Business, Pompeu Fabra University, Barcelona, Spain. \printead{e4}}

\runauthor{M. Drton et al.}

\affiliation{University
  of Washington, Seattle and Institute for Infocomm Research and Universit\`{a} di Genova}

\end{aug}

\begin{abstract}
  Gaussian latent tree models, or more generally, Gaussian latent
  forest models have Fisher-information matrices that become singular
  along interesting submodels, namely, models that correspond to
  subforests.  For these singularities, we compute the real
  log-canonical thresholds (also known as stochastic complexities or
  learning coefficients) that quantify the large-sample behavior of
  the marginal likelihood in Bayesian inference.  This provides the
  information needed for a recently introduced generalization of the
  Bayesian information criterion.  Our mathematical developments treat
  the general setting of Laplace integrals whose phase functions are
  sums of squared differences between monomials and constants.  We
  clarify how in this case real log-canonical thresholds can be
  computed using polyhedral geometry, and we show how to apply the
  general theory to the Laplace integrals associated with Gaussian
  latent tree and forest models.  In simulations and a data example,
  we demonstrate how the mathematical knowledge can be applied in
  model selection.

\end{abstract}


% history:
% \received{\smonth{1} \syear{0000}}

%\tableofcontents

\end{frontmatter}

\section{Introduction}

%% While limited in their expressibility, 
Graphical models based on trees are particularly tractable, which
makes them useful tools for exploring and exploiting multivariate
stochastic dependencies, as first demonstrated by \cite{chow:1968}.
More recent work develops statistical methodology for extensions that
allow for inclusion of latent variables and in which the graph may be
a forest, that is, a union of trees over disjoint vertex sets
\cite{choi:2011,tan:2011,mossel:2013}.  These extensions lead to a new
difficulty in that the Fisher-information matrix of a latent tree
model is typically singular along submodels given by subforests.  As
explained in \cite{watanabe_book}, such singularity invalidates
%% , in particular, 
the mathematical arguments that lead to the Bayesian
information criterion (BIC) of \cite{schwarz1978edm}, which is widely
used to guide model selection algorithms that infer trees or forests
\cite{edwards:2010}.  Indeed, the BIC will generally no longer share
the asymptotic behavior of Bayesian methods; see also
\cite[Sect.~5.1]{oberwolfach2009}.  Similarly, Akaike's information
criterion may no longer be an asymptotically unbiased estimator of the
expected Kullback-Leibler divergence that it is designed to
approximate
\cite{watanabe_book,watanabe:2010:jmlr,watanabe:2010:neural}.

In this paper, we study the large-sample behavior of the marginal
likelihood in Bayesian inference for Gaussian tree/forest models with
latent variables, with the goal of obtaining the mathematical
information needed to evaluate a generalization of BIC proposed in
\cite{drton:2013:sbic}.  As we review below, this information comes in
the form of so-called real log-canonical thresholds (also known as
stochastic complexities or learning coefficients) that appear in the
leading term of an asymptotic expansion of the marginal likelihood.
We begin by more formally introducing the models that are the object
of study.

Let $Z=(Z_u)_{u\in U}$ be a random vector whose components are indexed
by the vertices of an undirected tree $T=(U,E)$ with edge set
$E$. %%\subset U\times U$.
Via the paradigm of graphical modeling \cite{lauritzen:96}, the tree
$T$
induces a \textit{Gaussian tree model} $\mathbf{N}(T)$
for the joint distribution of $Z$.
The model $\mathbf{N}(T)$
is the collection of all multivariate normal distributions on
$\mathbb{R}^U$
under which $Z_u$
and $Z_v$
are conditionally independent given $Z_C=(Z_w:w\in
C)$ for any choice of two nodes $u,v$ and a set $C\subset
U\setminus\{u,v\}$ such that
$C$
contains a node on the (unique) path between $u$
and $v$.
For two nodes $u,v\in
U$, let $\overline{uv}$ be the set of edges on the path between
$u$
and $v$.
It can be shown that a normal distribution with correlation matrix
$R=(\rho_{uv})$ belongs to $\mathbf{N}(T)$ if and only if
\begin{equation}
  \label{eq:rhoijpath}
  \rho_{uv}=\prod_{e\in \overline{uv}} \rho_{e},
\end{equation}
where $\rho_{e}:=\rho_{xy}$
when $e$
is the edge incident to $x$
and $y$.
Indeed, for three nodes $v,w,u\in
U$ the conditional independence of $Z_v$ and $Z_w$ given
$Z_u$
is equivalent to $\rho_{vw}=\rho_{uv}\rho_{uw}$;
compare also \cite[p.~4359]{mossel:2013}.

In this paper, we are concerned with latent tree models in which only
the tree's leaves correspond to observed random variables.  So let
$V\subset
U$ be the set of leaves of tree
$T=(U,E)$.
Then the \emph{Gaussian latent tree model} $\mathbf{M}(T)$
for the distribution of the subvector $X:=(Z_v:{v\in
  V})$ is the set of all
$V$-marginals
of the distributions in $\mathbf{N}(T)$.
The object of study in our work is the parametrization of the model
$\mathbf{M}(T)$.
Without loss of generality, we may assume that the latent variables
$Z_a$
at the inner nodes $a\in
U\setminus
V$ have mean zero and variance one.  Moreover, we assume that the
observed vector $X$
has mean zero.  Then, based on~(\ref{eq:rhoijpath}), the distributions
in $\mathbf{M}(T)$
can be parametrized by the variances $\omega_{v}$
for each variable $X_v$,
$v\in V$, and the edge correlations $\omega_e$, $e\in E$.

Our interest is in the marginal likelihood of model $\mathbf{M}(T)$
when the variance and correlation parameters are given a prior
distribution with smooth and everywhere positive density.  Following
the theory developed by \cite{watanabe_book}, we will derive
large-sample properties of the marginal likelihood by studying the
geometry of the fibers (or preimages) of the parametrization map.

\begin{exmp}
  \label{ex:intro:3star}
  Suppose $T$ is a star tree with one inner node $a$ that
  is connected to each one of three leaves, labelled 1, 2, and 3.
  A positive definite correlation matrix
  $R=(\rho_{vw})\in \mathbb{R}^{V\times V}$ is the correlation matrix
  of a distribution in model $\mathbf{M}(T)$ if
  \begin{equation}
    \label{eq:intro:3star}
    R =  
     \begin{pmatrix}
      1 &  \rho_{12} & \rho_{13}\\
      \rho_{12} & 1 & \rho_{23}\\
      \rho_{13} & \rho_{23} & 1
    \end{pmatrix}=
     \begin{pmatrix}
      1 &  \omega_{a1}\omega_{a2} & \omega_{a1}\omega_{a3}\\
      \omega_{a1}\omega_{a2} & 1 & \omega_{a2}\omega_{a3}\\
      \omega_{a1}\omega_{a3} & \omega_{a2}\omega_{a3} & 1
    \end{pmatrix}
  \end{equation}
  for a choice of the three correlation parameters
  $\omega_{a1},\omega_{a2},\omega_{a3}\in[-1,1]$ that are associated
  with the three edges of the tree.  
  
  Now suppose that $R=(\rho_{vw})$ is indeed the correlation matrix of
  a distribution in $\mathbf{M}(T)$ and that $\rho_{vw}\neq 0$ for all
  $v\not=w$.  Then, modulo a sign change that corresponds to negating the
  latent variable at the inner node $a$, the parameters can be identified
  uniquely using the identities
  \begin{align*}
    \omega_{a1}^2&=\frac{\rho_{12}\rho_{13}}{\rho_{23}}, &
    \omega_{a2}&=\frac{\rho_{12}}{\omega_{a1}}, &
    \omega_{a3}&=\frac{\rho_{13}}{\omega_{a1}}.
  \end{align*}
  Hence, the fiber of the parametrization is finite, containing two
  points.

  If instead the correlations between the leaves are zero then this
  identifiability breaks down.  If $R$ is the identity matrix with
  $\rho_{12}=\rho_{13}=\rho_{23}=0$, then every vector
  $(\omega_{a1},\omega_{a2},\omega_{a3})\in[-1,1]^3$ that lies in the
  set
  \begin{equation*}
      \{\omega_{a1}=\omega_{a2}=0\}\cup
     \{\omega_{a1}=\omega_{a3}=0\}  
    \cup
      \{\omega_{a2}=\omega_{a3}=0\}
  \end{equation*}
  satisfies~(\ref{eq:intro:3star}).  The fiber of the identity
  matrix  is thus the union of three line
  segments that form a one-dimensional semi-algebraic set with a
  singularity at the origin where the lines intersect.  
  % \hfill
  % $\Diamond$
\end{exmp}

\begin{rem}
  \label{rem:directed-trees}
  Some readers may be more familiar with rooted trees with directed
  edges and model specifications based on the Markov properties for
  directed graphs or structural equations.  However, these are
  equivalent to the setup considered here, as can be seen by applying
  the so-called trek rule \cite{spirtes:2000}.  Our later results also
  apply to Bayesian inference in graphical models associated with
  directed trees.
\end{rem}

Suppose $\varphi$ is a smooth and positive density that defines a
prior distribution on the parameter space
$\Omega=(0,\infty)^V\times [-1,1]^E$ of the Gaussian latent tree model
$\mathbf{M}(T)$.  Let $\mathbf X_n=(X^{(1)},\dots,X^{(n)})$ be a
sample consisting of $n$ independent and identically distributed
random vectors in $\mathbb{R}^V$, and write
$L(\mathbf{M}(T)|\mathbf{X}_n)$ for the marginal likelihood of
$\mathbf{M}(T)$.  If $\mathbf{X}_n$ is generated from a distribution
$q\in\mathbf{M}(T)$ and $n\to\infty$, then it holds that
\begin{equation}
  \label{eq:th:watanabe:intro}
  \log L(\mathbf M(T)|\mathbf X_n) - \sum_{i=1}^n \log q(X^{(i)}) = -\frac{\lambda_{q}^{T}}{2} \log n + (\mathfrak{m}_{q}^{T}-1)\log \log
  n + O_p(1),
\end{equation}
where $\lambda_{q}^{T}\ge 0$ is a rational number smaller than or equal to
the dimension of the model $\mathbf{M}(T)$.  The number
$\mathfrak{m}_{q}^{T}$ is an integer greater than or equal to 1.  More
detail on how~(\ref{eq:th:watanabe:intro}) follows from results in
\cite{watanabe_book} is given in Section~\ref{sec:background}.  In
this paper, we derive formulas for the pair
$(\lambda_{q}^{T},\mathfrak{m}_{q}^{T})$ from~(\ref{eq:th:watanabe:intro}),
which will be seen to depend on the pattern of zeros in the
correlation matrix of the distribution $q$.

Let $\sigma_{vv}^*$ and $\rho_{vw}^*$ be the variances and the
correlations of the data-generating distribution $q$.  The point of
departure for our work is Proposition~\ref{prop:simple}, which
clarifies that the pair $(\lambda_{q}^{T},\mathfrak{m}_{q}^{T})$ is also
determined by the behavior of the deterministic Laplace integral
\begin{equation}
 \label{eq:intro-laplace-int} 
 \int_\Omega e^{-nH_q(\omega)} \varphi(\omega) \;d\omega,
\end{equation}
where the phase function in the exponent is
\[
H_q(\omega)\;=\;\sum_{v\in V} (\omega_{v}-\sigma_{vv}^*)^2+\sum_{v,w\in
  V\atop v\not= w} \bigg(\prod_{e\in \overline{vw}}
  \omega_{e}-\rho_{vw}^*\bigg)^2. 
\]
In the formulation of our results, we adopt the notation
\[
\RLCT_\Omega(H_q)\; :=\;(\lambda_{q}^{T},\mathfrak{m}_{q}^{T}),
\]
as $\lambda_q^{T}$ is sometimes referred to as real log-canonical
threshold and $\mathfrak{m}_q^{T}$ is the threshold's multiplicity.  Our
formulas for $\RLCT_\Omega(H_q)$ are stated in Theorem~\ref{th:main2}.
The proof of the theorem relies on facts presented in
Section~\ref{sec:monred}, which concern models with monomial
parametrizations in general.  As our formulas show, the marginal
likelihood admits non-standard large-sample asymptotics, with
$\lambda_q^{T}$ differing from the model dimension if $q$ exhibits zero
correlations (recall Example~\ref{ex:intro:3star}).  We describe 
the zero patterns of $q$ in terms of a subforest $F^*$ with edge set $E^*$.

Our result for trees generalizes directly to models based on forests.
If $F=(U,E)$ is a forest with the set $V\subset U$ comprising the
leaves of the subtrees, then we may define a Gaussian latent forest
model $\mathbf{M}(F)$ in the same way as for trees.  Again we assign a
variance parameter $\omega_v$ to each node $v\in V$ and a correlation
parameter $\omega_e$ to each edge $e\in E$.  Forming products of
correlations along paths, exactly as in~(\ref{eq:rhoijpath}), we
obtain again a parametrization of the correlation matrix of a
multivariate normal distribution on $\mathbb{R}^V$.  In contrast to
the case of a tree, there may be pairs of nodes with necessarily zero correlation, namely, when two
leaves $v$ and $w$ are in distinct connected components of $F$.
Theorem~\ref{th:main3} extends Theorem~\ref{th:main2} to the case of
forests.  The non-standard cases arise when the data-generating
distribution lies in the submodel defined by a proper subforest $F^*$ of
the given forest $F$.

The remainder of the paper begins with a review of the connection
between the asymptotics of the marginal likelihood and that of the
Laplace integral in~(\ref{eq:intro-laplace-int}); see
Section~\ref{sec:background} which introduces the notion of a real
log-canonical threshold (RLCT).  Gaussian latent tree/forest models
have a monomial parametrization and we clarify in
Section~\ref{sec:monred} how the monomial structure allows for
calculation of RLCTs via techniques from polyhedral geometry.  In
Section~\ref{sec:gaussian-latent-tree}, these techniques are applied
to derive the above mentioned Theorems~\ref{th:main2}
and~\ref{th:main3}. In Section~\ref{sec:simulations}, we demonstrate
how our results can be used in model selection with Bayesian
information criteria (BIC).  In a simulation study and an example of
temperature data, we compare a criterion based on RLCTs to the
standard BIC, which is based on model dimension alone.

\section{Background}
\label{sec:background}

Consider an arbitrary parametric statistical model $\mathbf{M}=\{
P_\theta :\theta\in \Theta\}$, with parameter space $\Theta\subseteq
\R^d$.  Let each distribution $P_\theta$ have density $p(x|\theta)$
and, for Bayesian inference, consider a prior distribution with
density $\varphi(\theta)$ on $\Theta$.  Writing $\mathbf
X_n=(X^{(1)},\dots,X^{(n)})$ for a sample of size $n$ from $P_\theta$, the
log-likelihood function of $\mathbf{M}$ is
\[
\ell(\theta| \mathbf{X}_n) = \sum_{i=1}^n \log p(X^{(i)}|\theta).
\]
The key quantity for Bayesian model determination is the integrated or
\emph{marginal likelihood}
\begin{equation}
  \label{eq:marg-lik}
  L(\mathbf{M}| \mathbf{X}_n) = \int_\Theta e^{\ell(\theta|
  \mathbf{X}_n)} \varphi(\theta)\;d\theta.
\end{equation}
As in the derivation of the Bayesian information criterion in
\cite{schwarz1978edm}, our interest is in the large-sample behavior of
the marginal likelihood.  

Let the sample $\mathbf X_n$ be drawn from a true distribution with
density $q$ that can be realized by the model, that is,
$q(x)=p(x|\theta^*)$ for some $\theta^*\in \Theta$.  Then, as we will
make more precise below, the asymptotic properties of the marginal
likelihood $L(\mathbf{M}|\mathbf X_n)$ are tied to those of the
Laplace integral
\begin{equation}\label{eq:ZnKq}
Z_n(K_q;\varphi)=\int_\Theta e^{-n K_q(\theta)}\varphi(\theta)\;d\theta,
\end{equation}
where 
\begin{equation}
  \label{eq:kl-div}
K_q(\theta) =\int \log\frac{q(x)}{p(x|\theta)}q(x)\;dx
\end{equation}
is the Kullback-Leibler divergence between the data-generating
distribution $q$ and distributions in the model
$\mathbf{M}$.  Note that $K_q(\theta)\geq 0$ for all $\theta$, and $K_q(\theta)=0$
precisely when $\theta$ satisfies $p(x|\theta)=p(x|\theta^*)$.  For large $n$ the integrand in~(\ref{eq:ZnKq}) is equal to $\varphi(\theta)$ if $K_{q}(\theta)=0$ and is negligibly small otherwise. Therefore, the main contribution to the integral $Z_n(K_q;\varphi)$ comes from a neighborhood
of the zero set
$$\mathcal{V}_\Theta(K_q)=\{\theta\in \Theta:\,\,K_q(\theta)=0\},$$
which we also call the \emph{$q$-fiber}.

Suppose now that $\Theta\subseteq\R^d$ is a semianalytic set and that
$K_q:\, \Theta\to [0,\infty)$ is an analytic function with compact $q$-fiber
$\mathcal{V}_\Theta(K_q)$.  Suppose further that the prior density $\varphi$ is
a smooth and positive function.
% and that the
% $q$-fiber $\mathcal{F}_q$ is compact.  
Then, under additional integrability conditions, the Main Theorem 6.2
in \cite{watanabe_book} shows that the marginal likelihood has the
following asymptotic behavior as the sample size $n$ tends to
infinity:
\begin{equation}
\label{eq:th:watanabe}
\log L(\mathbf M|\mathbf X_n)\;\;=\;\;\ell(\theta^*|
  \mathbf{X}_n)-\frac{\lambda}{2} \log n + (\mathfrak{m}-1)\log \log
n + O_p(1). 
\end{equation}
In~(\ref{eq:th:watanabe}), 
% the symbol $O_p(1)$ stands for a sequence
% of random variables that is bounded in probability, 
$\lambda$ is a
rational number in $[0,d]$, and $\mathfrak{m}$ is an integer in
$\{1,\dots,d\}$.  The number $\lambda$ is known as \emph{learning
coefficient}, \emph{stochastic complexity} or also \emph{real log-canonical
threshold}, and $\mathfrak{m}$ is the associated \emph{multiplicity}.  As
explained in \cite[Chap.~4]{watanabe_book}, the pair
$(\lambda,\mathfrak{m})$ also satisfies
\begin{equation}
  \label{eq:Zn-asy}
  \log Z_n(K_q;\varphi) = -\frac{\lambda}{2} \log n + 
  (\mathfrak{m}-1)\log \log
  n + O(1). 
\end{equation}
Moreover, the pair $(\lambda,\mathfrak{m})$ can equivalently be
defined using the concept of a zeta function as illustrated below; compare also
\cite{shaowei_rlct}.

\begin{defn}[The real log-canonical threshold]\label{def:rlct}
  Let $f:\Theta\to[0,\infty)$ be a nonnegative analytic function whose
  zero set $\mathcal{V}_\Theta(f)$ is compact and nonempty.  The
  \textit{zeta function}
  \begin{equation}\label{eq:zeta}
    \zeta(z)\;=\;\int_{\Theta} f(\theta)^{-z/2} \varphi(\theta)\;d \theta,
    \quad \text{Re}(z)\le 0,
  \end{equation}
  can be analytically continued to a meromorphic function on the
  complex plane.  The poles of this continuation are real and
  positive.  Let $\lambda$ be the smallest pole, known as the
  \textit{real log-canonical threshold (rlct)} of $f$, and let $\mathfrak{m}$
  be its multiplicity.  Since we are interested in both the rlct and
  its multiplicity, we use the notation
  ${\rm RLCT}_{\Theta}(f;\varphi):=(\lambda,\mathfrak{m})$.
  % We denote it by ${\rm rlct}_{\Theta}(f;\varphi)$ and write
  % ${\rm mult}_{\Theta}(f;\varphi)$ for its multiplicity.  
  When $\varphi(\theta)\equiv 1$, we simply write
  ${\rm RLCT}_{\Theta}(f)$.  Finally, if $g$ is another analytic
  function with
  ${\rm RLCT}_{\Theta}(g;\varphi)=(\lambda',\mathfrak{m}')$, then we
  write
  ${\rm RLCT}_{\Theta}(f;\varphi)>{\rm RLCT}_{\Theta}(g;\varphi)$ if
  $\lambda>\lambda'$ or if $\lambda=\lambda'$ and
  $\mathfrak{m}<\mathfrak{m}'$.
\end{defn}

\begin{exmp}
  \label{ex:xy}
  Suppose $K_q(\theta)=\theta_1^2\theta_2^2$ and $\Theta=[0,1]^2$.
  Then the $q$-fiber $\mathcal{V}_\Theta(K_q)$ is the union of two
  segments of the coordinate axes.  Taking $\varphi\equiv 1$, we have
  \[
  Z_n(K_q; \varphi) = \int_0^1\int_0^1 e^{-n \theta_1^2\theta_2^2}
  \,d\theta_1\,d\theta_2 . 
  \]
  This example is simple enough that $\RLCT_\Theta(K_q)$ can be
  computed by elementary means.  Let $\Phi(z)$ be the distribution
  function of the standard normal distribution.  Then 
  \[
  Z_n(K_q; \varphi) = \int_0^1\sqrt{\frac{\pi}{n\theta_2^2}}\left[
    \Phi(\sqrt{n}\theta_2) - 
  \Phi(0)\right]\,d\theta_2  
  =\sqrt{\frac{\pi}{n}}\int_0^{\sqrt{n}} \frac{\Phi(v) -
  \frac{1}{2}}{v} \,d v .
  \]
  Integration by parts yields
  \begin{align*}
    Z_n(K_q; \varphi) 
    &=\sqrt{\frac{\pi}{n}} \cdot \left[
      \log(v)\left(\Phi(v) -
      \frac{1}{2} \right) \right]_0^{\sqrt{n}}
      -\frac{1}{\sqrt{n}}\int_0^{\sqrt{n}} \log(v) e^{-v^2} \,d v \\
    &=\sqrt{\frac{\pi}{n}}
      \log\left(\sqrt{n}\right)\left(\Phi(\sqrt{n}) - 
      \frac{1}{2} \right) + O\left(n^{-1/2}\right).\\
    &=\frac{\sqrt{\pi}}{4}\cdot 
      \frac{\log(n)}{\sqrt{n}}\left(1+o(1)\right).
  \end{align*}
  Taking logarithms, we see that~(\ref{eq:Zn-asy}) holds with
  $\lambda=1$ and $\mathfrak{m}=2$.  It follows that
  $\RLCT_\Theta(K_q)=(1,2)$.  Concerning Definition~\ref{def:rlct}, we
  have that
  \[
  \zeta(z)=\int_0^1\int_0^1 \left(\theta_1^2\theta_2^2\right)^{-z/2}
  \,d\theta_1\,d\theta_2 = \frac{1}{(1-z)^2}
  \]
  for all $z\in\mathbb{C}$ with $\text{Re}(z)\le 0$.  In fact, this
  holds as long as $\text{Re}(z)<1$.  The meromorphic continuation of
  $\zeta(z)$ given by $1/(1-z)^2$ has one pole at $\lambda=1$.  The
  pole has multiplicity $\mathfrak{m}=2$ confirming that
  $\RLCT_\Theta(K_q)=(1,2)$.
\end{exmp}

\smallskip

In this paper we are concerned with Gaussian models for which we may
assume, without loss of generality, that all distributions are
centered.  So let the data-generating distribution $q$ be the
multivariate normal distribution $\mathcal{N}(0,\Sigma^*)$, with
positive definite $k\times k$ covariance matrix
$\Sigma^*=(\sigma^*_{ij})$.  Further, let $p(\cdot|\theta)$ be the
density of the distribution $\mathcal{N}(0,\Sigma(\theta))$ with
positive definite $k\times k$ covariance matrix
$\Sigma(\theta)=(\sigma_{ij}(\theta))$.
Then
\[
K_q(\theta) \;=\; \frac{1}{2}\left(
  \mathrm{tr}\left(\Sigma(\theta)^{-1}\Sigma^*\right)-k-\log\left(\frac{\det
    \Sigma^*}{\det\Sigma(\theta)}\right)\right).
\]
For fixed positive definite $\Sigma^*$, the function
\[
\Phi\mapsto \frac{1}{2}\left(
  \mathrm{tr}\left(\Phi^{-1}\Sigma^*\right)-k-\log\left(\frac{\det
    \Sigma^*}{\det\Phi}\right)\right)
\]
has a full rank Hessian at $\Phi=\Sigma^*$.  Hence, in a neighborhood
of $\Sigma^*$, we can both lower-bound and upper-bound $K_q$ by
positive multiples of the function
\[
\tilde K_q(\theta) \;=\; \sum_{i\le  j} \left(\sigma_{ij}(\theta)-\sigma^*_{ij}\right)^2.
\]
It follows that ${\rm RLCT}_{\Theta}(K_q;\varphi)={\rm
  RLCT}_{\Theta}(\tilde K_q;\varphi)$; compare \cite[Remark
7.2]{watanabe_book}.
%% \cite[Example 1.10, Theorem 1.11]{shaowei_thesis}
For our study of Gaussian latent tree (and forest) models, it is
convenient to change coordinates to correlations and consider the
function
\begin{equation}
  \label{eq:ftheta}
  H_q(\theta)\;=\;\sum_{i=1}^k
  \left(\sigma_{ii}(\theta)-\sigma_{ii}^*\right)^2+
  \sum_{i<j} \left(\rho_{ij}(\theta)-\rho_{ij}^*\right)^2,
\end{equation}
where $\rho_{ij}^*$ and $\rho_{ij}(\theta)$ are the correlations
obtained from $\Sigma^*$ or $\Sigma(\theta)$; so, e.g.,
$\rho^*_{ij}=\sigma^*_{ij}/\sqrt{\sigma^*_{ii}\sigma^*_{jj}}$.  Since
\begin{equation}
  \label{eq:regmodels}
  \RLCT_\Theta( K_q(\theta);\varphi) \;=\; \RLCT_\Theta( H_q(\theta);\varphi),
\end{equation}
%% for example from \cite[Example 1.10, Theorem 1.11]{shaowei_thesis} 
our discussion of latent tree models may thus start from
the following fact.

\begin{prop}
  \label{prop:simple} 
  Let $T=(U,E)$ be a tree with set of leaves $V\subset U$.  Let
  $\Omega=(0,\infty)^V\times [-1,1]^E$ be the parameter space for the
  Gaussian latent tree model $\mathbf{M}(T)$, the parameters being the
  variances $\omega_v$, $v\in V$, and the correlation parameters
  $\omega_e$, $e\in E$.  Suppose the (data-generating) distribution
  $q$ is in $\mathbf{M}(T)$ and has variances $\sigma_{vv}^*>0$ and a
  positive definite correlation matrix with entries $\rho^*_{vw}$.
  Then
  ${\rm RLCT}_\Omega(K_q;\varphi)={\rm RLCT}_\Omega(H_q;\varphi)$,
  where
  \begin{equation}\label{eq:monredG}
H_q(\omega)\;=\;\sum_{v\in V} (\omega_{v}-\sigma_{vv}^*)^2+\sum_{v,w\in
  V\atop v\not= w} \bigg(\prod_{e\in \overline{vw}}
  \omega_{e}-\rho_{vw}^*\bigg)^2. 
\end{equation}
\end{prop}

\section{Monomial parametrizations}
\label{sec:monred}

According to Proposition \ref{prop:simple}, the asymptotic behavior
of the marginal likelihood of a Gaussian latent tree model is
determined by the real log-canonical threshold of the function $H_q$
in~(\ref{eq:monredG}).  This function is a sum of squared differences
between monomials formed from the parameter vector $\omega$ and
constants determined by the data-generating distribution $q$.  In this
section, we formulate general results on the real log-canonical
thresholds for such monomial parametrizations, which also arise in
other contexts \cite{rusakov2006ams,pwz-2010-bic}.  

Specifically,
we treat functions of the form
\begin{equation}\label{eq:sos}
  H(\omega)\;=\; \sum_{i=1}^k (\omega^{u_i}-c_i^*)^2, \quad \omega\in\Omega,
\end{equation}
with domain $\Omega\subseteq\R^d$.  Here,
$c_1^*,\dots,c_k^*\in\mathbb{R}$ are constants and each monomial
$\omega^{u_i}:=\omega_1^{u_{i1}}\cdots\omega_d^{u_{id}}$ is given by a
vector of nonnegative integers $u_i=(u_{i1},\dots,u_{id})$.  Special
cases of this setup are the \emph{regular} case with
$H(\omega)=\omega_1^2+\cdots+\omega_d^2$, and the \emph{quasi-regular}
case of \cite{yamada2012statistical}, in which the vectors $u_i$ have
pairwise disjoint supports and all $c_i^*=0$.

Let $r$ be the number of summands on the right-hand side
of~(\ref{eq:sos}) that have $c_i^*\not=0$.  Without loss of
generality, assume that $c_1^*,\ldots,c_r^*\not=0$ and
$c^*_{r+1}=\cdots=c^*_k=0$.  Furthermore, suppose that
$\omega_1,\ldots,\omega_s$ are the parameters appearing in the
monomials $\omega^{u_1},\ldots,\omega^{u_r}$, that is,
$\cup_{i=1}^r\{j:u_{ij}>0\} =\{1,\dots,s\}$.  If $H(\omega)=0$ then
$\omega_i\neq 0$ for all $i=1,\ldots,s$.  Moreover, if the zero set
$\mathcal{V}_\Omega(H)=\{\omega\in\Omega : H(\omega)=0\}$ is compact,
then each one of the parameters $\omega_1,\dots,\omega_s$ is bounded
away from zero on $\mathcal{V}_\Omega(H)$.  (Clearly, the zero set of the
function $H_q$ from Proposition \ref{prop:simple} is compact.)

Now define the \emph{nonzero part} $H^1$ of $H$ as
\begin{equation}\label{eq:f1}
H^1(\omega_1,\ldots,\omega_s) \;:=\; \sum_{i=1}^r (\omega^{u_i}-c^*_i)^2
\end{equation}
and the \emph{zero part} $H^0$ of $H$ as 
\begin{equation}\label{eq:f0}
H^0(\omega_{s+1},\ldots,\omega_d)\;:=\;\sum_{i=r+1}^k \prod_{j=s+1}^d
\omega_j^{2u_{ij}}. 
\end{equation}

\begin{defn}\label{def:NP}
  The \emph{Newton polytope} $\Gamma(H^0)$ of the zero part $H^0$ is
  the convex hull of the points $(u_{ij}:s+1\le j\le d)\in \R^{d-s}$
  for $i=r+1,\ldots, k$.  The \emph{Newton polyhedron} of $H^0$ is the
  polyhedron
  $$
  \Gamma_+(H^0)\;:=\;\{x+y\in \R^{d-s}:\, x\in \Gamma(H^0),\, 
  y\in [0,\infty)^{d-s}\}. 
  $$
  Let $\mathbf 1=(1,\dots,1)\in \R^{d-s}$ be the vector of all ones.
  Then the \emph{$\mathbf 1$-distance} of $\Gamma_+(H^0)$ is the
  smallest $t\in \R$ such that $t\mathbf 1\in \Gamma_+(H^0)$.  The
  associated \textit{multiplicity} is the codimension of the
  (inclusion-minimal) face of $\Gamma_+(H^0)$ containing
  $t\mathbf 1$.
\end{defn}

We say that $A\subseteq\R^d$ is a \emph{product of intervals} if
$A = [a_1, b_1] \times [a_2, b_2] \times \cdots \times [a_d, b_d]$
with $a_i< b_i \in \R\cup \{-\infty,\infty\}$.  The following is the
main result of this section.  It is proved in
Appendix~\ref{sec:proof-th:main1}.

\begin{thm}\label{th:main1}
  Suppose that $\Omega$ is a product of intervals, and let $\Omega_1$
  and $\Omega_0$ be the projections of $\Omega$ onto the first $s$ and
  the last $d-s$ coordinates, respectively.  Let $H$ be the sum of
  squares from~\eqref{eq:sos} and assume that 
  %% $\mathcal{F}_0$, the zero set of $H$ on $\Omega$, 
  the zero set $\{\omega\in\Omega:H(\omega)=0\}$ is non-empty and
  compact.  Let $\varphi:\Omega\to (0,\infty)$ be a smooth positive
  function that is bounded above on $\Omega$.
    Then
  \[
  \RLCT_{\Omega}(H;\varphi) = (\lambda_0+\lambda_1,\mathfrak{m}),
  \]
  where $\lambda_1$ is the codimension of
  $\mathcal{V}_{\Omega_1}(H^1)=\{\omega\in \Omega_1:\,
  H^1(\omega)=0\}$
  in $\R^s$, and $1/\lambda_0$ is the $\mathbf 1$-distance of the
  Newton polyhedron $\Gamma_+(H^0)$ with associated multiplicity
  $\mathfrak{m}$.  Here, $\lambda_0=0$ and $\mathfrak{m}=1$ if $H$ has
  no zero part, i.e., $s=d$.
\end{thm}

\begin{rem}
  In order to compute the codimension of
  $\mathcal{V}_{\Omega_1}(H^1)$, one may consider one orthant at a
  time and take logarithms (accounting for signs).  This turns the
  equations $H^1(\omega)=0$ into linear equations in
  $\log \omega_1,\dots,\log \omega_s$.  
\end{rem}

\begin{exmp}
  \label{ex:reg-case}
  If $H(\omega)=\omega_1^2+\dots +\omega_d^2$ and $\Omega=\R^d$,
  then~(\ref{eq:ZnKq}) is a Gaussian integral and it is clear (c.f.~(\ref{eq:Zn-asy})) that
  $\RLCT_{\Omega}(H)=(d,1)$.  The Newton polytope for $H^0=H$ is the convex
  hull of the canonical basis vectors of $\R^d$.  The Newton
  polyhedron of $H$ has $\mathbf 1$-distance $1/d$ with
  multiplicity 1.  The same is true whenever
  \begin{equation}\label{eq:regcond}
    H(\omega_{1},\ldots,\omega_d)\;=\;\omega_{1}^2+\cdots+\omega_d^2+\mbox{``higher even order terms''}.
   \end{equation}
\end{exmp}

\begin{exmp}
  \label{ex:xy-newton}
  Earlier, we have shown that on $\Omega=[0,1]^2$ the function
  $H(\omega)=\omega_1^2\omega_2^2$ has $\RLCT_\Omega(H)=(1,2)$;
  recall Example~\ref{ex:xy}.  The function has no nonzero part.  Its
  Newton polytope consists of a single point, namely, $(1,1)$.  The
  Newton polyhedron is $[1,\infty)^2$.  Clearly, the
  $\mathbf{1}$-distance of the Newton polyhedron is 1.  Since the ray
  spanned by $\mathbf{1}$ meets the Newton polyhedron in the vertex
  $(1,1)$, the multiplicity is 2, as it had to be according to
  our earlier calculation.
\end{exmp}

\begin{exmp}
  Consider the function
  $$
  H(\omega)=(\omega_1\omega_2-1)^2+\omega_1^2\omega_3^2+
  \omega_2^2\omega_3^2+
  \omega_3^2\omega_4^2
  $$
  on $\Omega=[-2,2]^4$.  The nonzero part is
  $H^1(\omega_1,\omega_2)=(\omega_1\omega_2-1)^2$ and the zero part is
  $H^0(\omega_3,\omega_4)=2\omega_3^2+\omega_3^2\omega_4^2$.  With
  $\Omega_1=[-2,2]^2$, the codimension of
  $\mathcal{V}_{\Omega_1}(H^1)$ is $\lambda_1=1$.  The Newton polytope
  of $H^0$ is the convex hull of $(1,0)$ and $(1,1)$.  The Newton
  polyhedron of $H^0$ is $[1,\infty)\times[0,\infty)$. Hence,
  $\lambda_0=1$ and $\mathfrak{m}=1$.  Note that while the point
  $(1,1)$ is a vertex of the Newton polytope, it lies on a
  one-dimensional face of the Newton polyhedron.  In conclusion,
  $\RLCT_{\Omega}(H)=(2,1)$.
\end{exmp}

\section{Gaussian latent tree and forest models}
\label{sec:gaussian-latent-tree}

Let $T=(U,E)$ be a tree with set of leaves $V$.  By Proposition
\ref{prop:simple}, our study of the marginal likelihood of the
Gaussian latent tree model $\mathbf{M}(T)$ turns into the study of the
function
\begin{equation}
  \label{eq:tree:Hq}
H_q(\omega)=\sum_{v\in V} (\omega_{v}-\sigma_{vv}^*)^2+\sum_{v,w\in
  V\atop v\not= w} \bigg(\prod_{e\in \overline{vw}}
\omega_{e}-\rho_{vw}^*\bigg)^2. 
\end{equation}
Since $\sigma_{vv}^*>0$ for all $v\in V$, the split of $H_q$ into its
zero and nonzero part depends solely on the zero pattern among the
correlations $\rho_{vw}^*$ of the data-generating distribution $q$.  Furthermore,
from the form of the parametrization in (\ref{eq:rhoijpath}), it is
clear that zero correlations can arise only if one sets $\omega_e=0$
for one or more edges $e$ in the edge set $E$.  For a fixed set
$E_0\subseteq E$, the set of parameter vectors $\omega\in \Omega$ with
$\omega_e=0$ for all $e\in E_0$ parametrizes the forest model
$\mathbf{M}(F_0)$, where $F_0$ is the forest obtained from $T$ by
removing the edges in $E_0$.  In this submodel, $\rho_{vw}\equiv0$ if and
only if $v$ and $w$ lie in two different connected components of
$F_0$.
 
It is possible that two different subforests induce the same pattern
of zeros among the correlations of the data-generating distribution
$q$.  However, there is always a unique minimal forest $F^*(q)=(U^*,E^*)$
inducing this zero pattern, and we term $F^{*}(q)$ the {\em $q$-forest}.
Put differently, the $q$-forest $F^*(q)$ is obtained from $T$ by first
removing \emph{all} edges $e\in\overline{uv}$ for all pairs of nodes $u,v\in U$ that can have zero
correlation under $q$  and then removing all inner nodes of
$T$ that have become isolated.  Isolated leaf nodes are retained so
that $V\subseteq U^*$.  In the remainder of this section, we take
$E_{0}=E\setminus E^*$ to be the set of edges whose removal defines
$F^{*}(q)$.  We write $v\sim w$ if $v$ and $w$ are two leaves in $V$ that
are joined by a path in the $q$-forest $F^*(q)$.
% The equivalence classes
% $B_1,\dots,B_r$ of this equivalence relation give a partition of the
% set of leaves $V$ that we refer to as the \emph{$q$-partition}.

\begin{exmp}\label{ex:quartet1}
  Let $T$ be the quartet tree in Figure \ref{fig:quartet}(a).  Let $q$
  have $\rho_{12}^*\neq 0$ but $\rho_{vw}^*=0$ for all other
  $\{v,w\}\subseteq V= \{1,2,3,4\}$.  The $q$-forest $F^*(q)$ is obtained
  by removing the edges in $E_0=\{\{a,b\},\{b,3\},\{b,4\}\}$.  Inner
  node $b$ becomes isolated and is removed as well.  The forest $F^*(q)$
  thus has the five nodes in the set $U^*=\{1,2,3,4,a\}$, and the two
  edges in the set $E^*=\{\{1,a\},\{2,a\}\}$; see Figure
  \ref{fig:quartet}(b).
  % The
  % $q$-partition is $B_1=\{1,2\}$, $B_2=\{3\}$, $B_3=\{4\}$.
  % % Any point $\omega$ in the $q$-fiber
  % % has $\omega_{1a},\omega_{2a}\neq 0$ and at least two of the three
  % % parameters $\omega_{ab}$, $\omega_{3b}$, $\omega_{4b}$.  
\end{exmp}

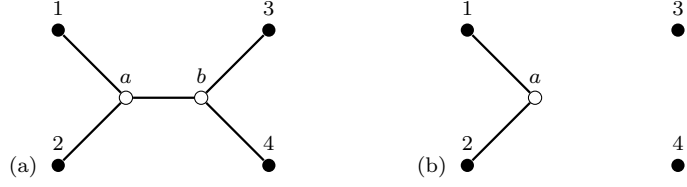
\begin{figure}[t!]
\centering
\tikzstyle{vertex}=[circle,fill=black,minimum size=5pt,inner sep=0pt]
\tikzstyle{hidden}=[circle,draw,minimum size=5pt,inner sep=0pt]
  (a) 
  \begin{tikzpicture}
  \node[vertex] (1) at (-.9,.9)  [label=above:$1$] {};
    \node[vertex] (2) at (-.9,-.9)  [label=above:$2$] {};
    \node[vertex] (3) at (1.9,.9) [label=above:$3$]{};
    \node[vertex] (4) at (1.9,-.9) [label=above:$4$]{};
    \node[hidden] (a) at (0,0) [label=above:$a$] {};
    \node[hidden] (b) at (1,0) [label=above:$b$] {};
          \draw[line width=.3mm] (a) to (b);
    \draw[line width=.3mm] (a) to (1);
    \draw[line width=.3mm] (a) to (2);
    \draw[line width=.3mm] (b) to (3);
    \draw[line width=.3mm] (b) to (4);
  \end{tikzpicture}\qquad\qquad\qquad
  (b)
  \begin{tikzpicture}
  \node[vertex] (1) at (-.9,.9)  [label=above:$1$] {};
    \node[vertex] (2) at (-.9,-.9)  [label=above:$2$] {};
    \node[vertex] (3) at (1.9,.9) [label=above:$3$]{};
    \node[vertex] (4) at (1.9,-.9) [label=above:$4$]{};
    \node[hidden] (a) at (0,0) [label=above:$a$] {};
    \draw[line width=.3mm] (a) to (1);
    \draw[line width=.3mm] (a) to (2);
  \end{tikzpicture}

  \caption{(a) A quartet tree $T$; (b) the $q$-forest from Example
    \ref{ex:quartet1}.}\label{fig:quartet}
\end{figure}

Moving on to the decomposition of the function
from~(\ref{eq:tree:Hq}), recall that we divide the parameter vector
$\omega$ into coordinates $(\omega_1,\ldots, \omega_s)$ that never
vanish on the $q$-fiber $\mathcal{V}_\Omega(H_q)$ and
the remaining part $(\omega_{s+1},\ldots,\omega_d)$. In our case,
$(\omega_1,\ldots,\omega_s)$ consists of all $\omega_{v}$ for $v\in V$
and $\omega_{e}$ for $e\in E^*$ and $(\omega_{s+1},\ldots,\omega_d)$
consists of $\omega_{e}$ for $e\in E_0=E\setminus E^*$. Moreover, 
\begin{equation}\label{eq:f1tree}
H_{q}^1(\omega_1,\ldots,\omega_s)\;=\; 
\sum_{v\in V} (\omega_{v}-\sigma_{vv}^*)^2 + 
\sum_{v,w\in
  V\atop v\neq w, \; v\sim w}\bigg(\prod_{e\in
  \overline{vw}}\omega_{e}-\rho_{vw}^*\bigg)^2
\end{equation}
and
\begin{equation}\label{eq:f0tree}
H_{q}^0(\omega_{s+1},\ldots,\omega_d)\;=\;\sum_{v\not\sim
  w}\prod_{e\in \overline{vw}\cap E_0}\omega_{e}^2. 
\end{equation}

The Gaussian latent tree model $\mathbf{M}(T)$ given by a tree $T$
with set of leaves $V$ and edge set $E$ has dimension
$$\dim \mathbf{M}(T)=|V|+|E|-l_2,$$
where $l_2$ denotes the number of degree two nodes in $T$.  Similarly,
the model given by a forest $F$ with set of leaves $V$ and edge set
$E$ has dimension
$$
\dim \mathbf{M}(F)\;=\; \sum_{i=1}^r \dim \mathbf{M}(T_i)\;=\; |V|+|E|-l_2,
$$
where $T_1,\dots,T_r$ are the trees defined by the connected
components of $F$ and $l_2$ is again the number of degrees two nodes.

\begin{exmp}\label{ex:dim5}
  The $q$-forest $F^*$ from Example \ref{ex:quartet1} has
  $\dim \mathbf{M}(F^*)=4+2-1=5$.  The dimensions for the trees in the
  forest $F^*$ are $\dim \mathbf{M}(T^*_1)=3$,
  $\dim \mathbf{M}(T^*_2)=1$, and $\dim \mathbf{M}(T^*_3)=1$; the
  trees $T^*_2$ and $T^*_3$
  each contain only a single node.
\end{exmp}

The following theorem provides the real log-canonical thresholds of
Gaussian latent tree models.  The proof of theorem is given in
Appendix~\ref{sec:app:tree-proofs}.

\begin{thm}\label{th:main2}
  Let $T=(U,E)$ be a tree with set of leaves $V\subset U$, and let $q$
  be a distribution in the Gaussian latent tree model $\mathbf{M}(T)$.
  Write $\Omega=(0,\infty)^V\times [-1,1]^E$ for the parameter space
  of $\mathbf{M}(T)$, and let $F^*(q)=(U^*,E^*)$ be the $q$-forest.  If
  $\varphi:\Omega\to(0,\infty)$ is a smooth positive function that is
  bounded above on $\Omega$, then the function $H_q$
  from~(\ref{eq:tree:Hq}) has
  $$
  {\rm RLCT}_\Omega(H_q;\varphi)
  \;=\;\left(\dim \mathbf{M}(F^*(q))+\frac{\sum_{e\in E\setminus E^*}w(e)}{2},\;
    1+l_2'\right),
  $$
  where $w(e)=|e\cap U^*|\in \{0,1,2\}$ is the number of nodes that
  $e$ shares with $F^*(q)$, and $l_2'$ is the number of nodes in $T$ that
  have degree two and are not in $U^{*}$.
\end{thm}
Theorem \ref{th:main2} implies in particular that the pair $(\lambda_{q}^{T},\mathfrak{m}_{q}^{T})$ depends on $q$ only through the forest $F^{*}(q)$ and we write
$$
\lambda_{F^{*}(q),T}\;:=\;\lambda_{q}^{T},\qquad \mathfrak{m}_{F^{*}(q),T}\;:=\;\mathfrak{m}_{q}^{T}.
$$

\begin{exmp}
  In Example \ref{ex:quartet1}, $\dim\mathbf{M}(F^*)=5$ (c.f. Example~\ref{ex:dim5}) and
  $\sum_{e\in E_0}w(e)=3$.  Hence, the real log-canonical threshold $\lambda_{F^{*}(q),T}$ is
  13/2, which translates into a coefficient of 13/4 for the $\log n$
  term in the asymptotic expansion of the log-marginal likelihood.
  Note that the threshold 13/2 is smaller than $\dim\mathbf{M}(T)=9$,
  making the latent tree model behave like a lower-dimensional model.
  % The standard BIC would thus
  % overpenalize the log-likelihood.
\end{exmp}

\begin{exmp}
  Suppose $T$ has two leaves, labelled 1 and 2, and one inner node
  $a$, which then necessarily has degree two.  If $q$ is a
  distribution under which the random variables at the two leaves are
  uncorrelated, then we have
  \[
  H_q(\omega)=(\omega_1-\sigma_{11}^*)^2+(\omega_2-\sigma_{22}^*)^2
  +(\omega_{1a}\omega_{2a})^2.
  \]
  Using the calculation from Example~\ref{ex:xy} or
  Example~\ref{ex:xy-newton}, we see that $\RLCT_\Omega(H_q)=(3,2)$.
  When applying Theorem~\ref{th:main2}, the $q$-forest $F^*$ has the
  leaves 1 and 2 isolated and $\dim\mathbf{M}(F^*)=2$.  Since $l_2'=1$
  and each one of the two removed edges satisfies $w(e)=1$, the
  formula from Theorem~\ref{th:main2} yields
  $\RLCT_\Omega(H_q)=(3,2)$, as it should.
\end{exmp}

\begin{rem}\label{rem:deg2contract}
  Note that if $T$ has an (inner) node of degree two, then we can
  contract one of the edges the node is adjacent to obtain a tree
  $\tilde T$ with $\mathbf{M}(\tilde T)=\mathbf{M}(T)$.  Repeating
  such edge contraction it is always possible to find a tree with all
  inner nodes of degree at least three that defines the same model as
  the original tree $T$.  Moreover, in applications such
  as phylogenetics, the trees of interesting do not have nodes of degree
  two, in which case the multiplicity in RLCT is always equal to
  one.
\end{rem}

In the model selection problems that motivate this work, we wish to
choose between different forests.  We thus state an explicit result
for forests in the below Theorem~\ref{th:main3}.  For a forest $F$, we
define $q$-forests in analogy to the definition we made for trees.  In
other words, we apply the previous definitions to each tree appearing
in the connected components of $F$ and then form the union of the
results.  Similarly, the proof of Theorem~\ref{th:main3} is obtained
by simply applying Theorem~\ref{th:main2} to each connected component
of the given forest $F$.

\begin{thm}\label{th:main3}
  Let $F=(U,E)$ be a forest with the set of leaves $V\subset U$, and
  let $q$ be a distribution in the Gaussian latent forest model
  $\mathbf{M}(F)$.  Write $\Omega=(0,\infty)^V\times [-1,1]^E$ for the
  parameter space of $\mathbf{M}(F)$, and let $F^*(q)=(U^*,E^*)$ be the
  $q$-forest.  If $\varphi:\Omega\to(0,\infty)$ is a smooth positive
  function that is bounded above on $\Omega$, then the function $H_q$
  from~(\ref{eq:tree:Hq}) has
  $$
  {\rm RLCT}_\Omega(H_q;\varphi) :=(\lambda_{q}^{F},\mathfrak{m}_{q}^{F})=
  \left( 
    \dim\mathbf{M}(F^*(q))+\frac{\sum_{e\in E\setminus E^*} w(e)}{2},1+l_2'\right),
  $$
  where $w(e)=|e\cap U^*|\in \{0,1,2\}$ is the number of nodes that
  $e$ shares with $F^*(q)$, and $l_2'$ is the number of nodes in $F$ that
  have degree two and are not in $U^{*}$.
\end{thm}
As in Theorem \ref{th:main2}, the pair $(\lambda_{q}^{F},\mathfrak{m}_{q}^{F})$ depends on $q$ only through the forest $F^{*}(q)$ and we write
$$
\lambda_{F^{*}(q),F}\;:=\;\lambda_{q}^{F},\qquad \mathfrak{m}_{F^{*}(q),F}\;:=\;\mathfrak{m}_{q}^{F}.
$$

\begin{rem}
  Fix a forest $F=(U,E)$ with leaves $V\subset U$, and let
  $F^{*}=(U^{*},E^{*})$ be any subforest of $F$ with the same leaves
  (any $F^{*}(q)$ is of this form). Let $d_{F}$ and $d_{F^*}$ be such that
  $d_F(u)$ is the degree of $u$ in $F$ for all $u\in U$ and similarly
  for $d_{F^*}$. Note that
\begin{align*}
	\sum_{e\in E\setminus E^*} w(e) = \sum_{u\in U^*} (d_{F}(u) - d_{F^*}(u)).
\end{align*}
From this and our prior formula for $\dim \mathbf{M}(F^*)$ we have that
$$
	\lambda_{F^*,F} = |U^*| + |E^*| - l_2 + \frac{1}{2}\sum_{u\in
          U^*} (d_{F}(u) - d_{F^*}(u)). 
$$
where $l_2$ is the number of degree 2 nodes in $F^*$. Computing
$\lambda_{F^*,F}$ can now easily be done in linear time in the size of
$F$, i.e.\ in $O(|U|+|E|) = O(|U|)$ time, under the assumption that we
have stored $F$ and $F^*$ as adjacency lists and there is a map, with
$O(1)$ access time, associating vertices in $F^*$ with those in
$F$. In computational practice we found that the prior two conditions
are trivial to guarantee. In particular, note that if $F$ and $F^*$
are stored as adjacency lists we may simply loop over these lists,
taking $O(|U|+|E| + |U^*| + |E^*|) = O(|U|)$ time, and precompute
$d_{F}$, $d_{F^*}$, $l_2$, $|U^*|$, and $|E^*|$. Computing
$\lambda_{F^*,F}$ is then simply a matter of summing over $u\in U^*$
and using the precomputed values of $d_{F}(u)$ and $d_{F^*}(u)$,
taking $O(U^*)$ time.  Similarly, noting that $l'_2 = \sum_{u\in
  U\setminus U^*} 1_{[d_{F}(u)=2]}$, we have that
$\mathfrak{m}_{F^*,F}$ can also be computed in linear time in the size
of $F$.
\end{rem}

\section{Singular BIC for latent Gaussian tree models}\label{sec:simulations}

In this section, we consider the model selection problem of inferring
the forest $F$ underlying a Gaussian latent forest model
$\mathbf{M}(F)$ based on a sample of independent and identically
distributed observations $\mathbf X_n=(X^{(1)},\dots,X^{(n)})$.  To
this end, we consider Bayesian information criteria that are inspired
by the developed large-sample theory for the marginal
likelihood $L(\mathbf M(F)|\mathbf X_n)$. Note that for all the following
simulations the space of models we consider implicitly include only forests and 
trees without degenerate degree 2 nodes; as described in Remark \ref{rem:deg2contract},
this results in an RLCT whose multiplicity is always 1.

As stated in~(\ref{eq:th:watanabe:intro}) and~(\ref{eq:th:watanabe}),
the RLCTs found in Section~\ref{sec:gaussian-latent-tree} give the
coefficients for logarithmic terms that capture the main differences
between the log-marginal likelihood and the log-likelihood of the true
data-generating distribution $q$.  Let $\hat q_F$ be the maximum
likelihood estimator of $q$ in the Gaussian latent forest model
$\mathbf{M}(F)$.  By the results of \cite{drton2007_LRandSING}, if
$q\in\mathbf{M}(F)$ and $n\to\infty$, then
\[
\sum_{i=1}^n \left[ \log \hat q_F(X^{(i)}) -  \log q(X^{(i)}) \right]
\;=\; O_p(1)
\]
and thus, by~(\ref{eq:th:watanabe}), we also have
\begin{multline}
  \label{eq:th:wata:with:mle}
  \log L(\mathbf M(F)|\mathbf X_n) \;=\; \\ \sum_{i=1}^n \log \hat
  q_F(X^{(i)})  -\frac{\lambda_{F^{*}(q),F}}{2} \log n +
  (\mathfrak{m}_{F^{*}(q),F}-1)\log \log n + O_p(1).
\end{multline}
The pair $(\lambda_{F^{*}(q),F},\mathfrak{m}_{F^{*}(q),F})$ on the right hand side still
depends on the unknown data-generating distribution $q$ through the
forest $F^{*}(q)$.  However, the pair is a
discontinuous function of $q$ and plugging in the MLE $\hat q_F$ has
little appeal.  Instead, we will consider a criterion proposed by
\cite{drton:2013:sbic}, in which one averages over the possible values
of $(\lambda_{F',F},\mathfrak{m}_{F',F})$ for all subforests $F'$ of $F$.  As in \cite{drton:2013:sbic}, we
refer to the resulting model selection score as the `singular Bayesian
information criterion', or sBIC for short.

We briefly describe how sBIC is computed.  Let $\mathcal{F}$ be the
set of forests in the model selection problem, which we assume to
contain the empty forest $F_\emptyset=(V,\emptyset)$.  Note that every
forest $F\in\mathcal{F}$ has set of leaves $V$.  For forest
$F\in\mathcal{F}$ with subforest $F'\in\mathcal{F}$, let
$(\lambda_{F',F},\mathfrak{m}_{F',F})$ be the pair
from~(\ref{eq:th:wata:with:mle}) when the distribution $q$ has $F'$
as $q$-forest, that is $F^{*}(q)=F'$.  Theorem~\ref{th:main3} gives the value of this RLCT
pair.  Define
\begin{equation}
  \label{eq:LFF*}
  L_{F'F}' \;=\;   n^{-\lambda_{F'F}/2} (\log
  n)^{\mathfrak{m}_{F'F}-1}\prod_{i=1}^n \hat q_F(X^{(i)}),
\end{equation}
which is a proxy for the marginal likelihood
$L(\mathbf M(F)|\mathbf X_n)$ obtained by exponentiating the right
hand side of~(\ref{eq:th:wata:with:mle}) and omitting the
$O_p(1)$ remainder.  For each $F\in\mathcal{F}$, the sBIC of model
$\mathbf{M}(F)$ is defined as $\log x_F$, where
$(x_F:F\in\mathcal{F})$ is the unique positive solution to the
equation system
\begin{equation}
  \label{eq:sbic-eqns}
  \sum_{F'\subseteq F} \left( x_F-L_{F'F}' \right) x_{F'}
  \;=\;0,\quad F\in\mathcal{F}.
\end{equation}
The system~(\ref{eq:sbic-eqns}) is triangular and can be solved by
recursively solving univariate quadratic equations.  The starting
point is the case when $F$ is the empty forest $F_\emptyset$, for
which $F'=F_\emptyset$ is the only possible $q$-forest and
(\ref{eq:sbic-eqns}) gives
$x_{F_\emptyset}(x_{F_\emptyset}-L_{F_\emptyset F_\emptyset}')=0$.  The sBIC of the model
$\mathbf{M}(F_\emptyset)$ is thus $\log L_{F_\emptyset F_\emptyset}'$,
which coincides with 
the usual BIC as the relevant RLCT is given by
$\lambda_{F_{\emptyset}F_{\emptyset}}=\dim \mathbf{M}(F_\emptyset)=|V|$ and $\mathfrak{m}_{F_{\emptyset}F_{\emptyset}}=1$.
When the forest $F$ is nonempty, the sBIC and the BIC of
$\mathbf{M}(F)$ differ.

In \cite{drton:2013:sbic}, sBIC is motivated by considering weighted
averages of the approximations $L_{F'F}'$, with the weights being
data-dependent.  Furthermore, it is shown that the sBIC of
$\mathbf{M}(F)$ differs from $\log L(\mathbf M(F)|\mathbf X_n)$ by an
$O_p(1)$ remainder whenever data are generated from a distribution
$q\in\mathbf{M}(F)$, even if $q$ lies in a strict submodel
$\mathbf{M}(F^*)\subset\mathbf{M}(F)$.  The same is true for BIC only
if $q\in\mathbf{M}(F)$ does not belong to any strict submodel (i.e.,
all edge and path correlations are nonzero and $F$ equals the
$q$-forest $F^*$).  In what follows, we explore the differences
between the RLCT-based sBIC and the dimension-based BIC in two
simulation studies and on a temperature data set.

\subsection{Simulation Studies} \label{sec:sim_study}

The first task we consider is selection a subforest of a given 
tree $T$, where each subforest as well the tree $T$ share a set of
leaves $V$, or in other words, each subforest is a $q$-forest for some $q\in \mathbf{M}(T)$.  When ordering edge sets by inclusion, the
set of all subforests of $T$ becomes a poset that we denote by
$\mathcal P_T$.  The poset is a lattice with the empty graph (with
$|V|$ isolated nodes) as minimal element and the tree $T$
as maximal element.  To select a forest, we optimize BIC and sBIC,
respectively, over the set $\mathcal P_T$.  Maximum likelihood
estimates are computed with an EM algorithm, in which we repeatedly
maximize the conditional expectation of the complete-data
log-likelihood function of forest models $\mathbf{N}(F)$ for a random
vector $Z$ comprising both the observed variables at the leaves in $V$
and the latent variables at the inner nodes of $F$; recall the
notation from the introduction.

As a concrete example, we choose $T$ to be the tree in
Figure~\ref{fig:5leaves}(a).  We generate data from a distribution $q$
that lies in $\mathbf M(T)$ but under which the third leaf is
independent from all other leaves.  The corresponding
$q$-forest $F^*$ is depicted on Figure \ref{fig:5leaves}(b).  We
choose $q$ to have covariance matrix
\begin{equation}\label{eq:true5}
\Sigma^*\quad=\left[\begin{array}{rrrrr}
  1 & 0.13 & 0 & 0.22 & 0.36  \\ 
  0.13 & 1 & 0 & 0.22 & 0.13  \\ 
  0 & 0 & 1 & 0 & 0 \\ 
  0.22 & 0.22 & 0 & 1 & 0.22  \\ 
  0.36 & 0.13 & 0 & 0.22 & 1 \\ 
  \end{array}\right]
\end{equation}
which is obtained by taking all edge correlations equal to 0.6.  We
then generate a random sample of size $n$ from
$N(\mathbf 0, \Sigma^*)$ and pick the best model with respect to the
BIC and the best model with respect to sBIC.  For each considered
choice of a sample size $n$, this procedure is  repeated $100$ times.

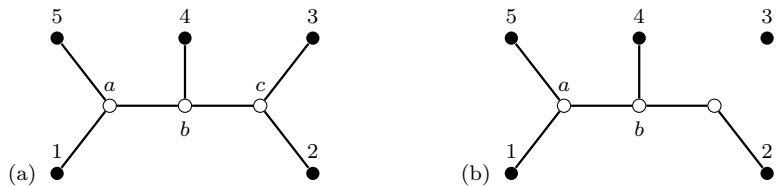
\begin{figure}[t]
\centering
\tikzstyle{vertex}=[circle,fill=black,minimum size=5pt,inner sep=0pt]
\tikzstyle{hidden}=[circle,draw,minimum size=5pt,inner sep=0pt]
  (a)
  \begin{tikzpicture}
  \node[vertex] (5) at (-.9,.9)  [label=above:$5$] {};
    \node[vertex] (1) at (-.9,-.9)  [label=above:$1$] {};
    \node[vertex] (4) at (.8,.9) [label=above:$4$]{};
    \node[vertex] (3) at (2.5,.9) [label=above:$3$]{};
    \node[vertex] (2) at (2.5,-.9) [label=above:$2$]{};
    \node[hidden] (a) at (-.2,0) [label=above:$a$] {};
    \node[hidden] (b) at (.8,0) [label=below:$b$] {};
    \node[hidden] (c) at (1.8,0) [label=above:$c$] {};
          \draw[line width=.3mm] (a) to (b);
    \draw[line width=.3mm] (a) to (5);
    \draw[line width=.3mm] (a) to (1);
    \draw[line width=.3mm] (b) to (4);
    \draw[line width=.3mm] (b) to (c);
    \draw[line width=.3mm] (3) to (c);
    \draw[line width=.3mm] (2) to (c);
  \end{tikzpicture}\qquad\qquad\qquad
  (b)
  \begin{tikzpicture}
  \node[vertex] (5) at (-.9,.9)  [label=above:$5$] {};
    \node[vertex] (1) at (-.9,-.9)  [label=above:$1$] {};
    \node[vertex] (4) at (.8,.9) [label=above:$4$]{};
    \node[vertex] (3) at (2.5,.9) [label=above:$3$]{};
    \node[vertex] (2) at (2.5,-.9) [label=above:$2$]{};
    \node[hidden] (a) at (-.2,0) [label=above:$a$] {};
    \node[hidden] (b) at (.8,0) [label=below:$b$] {};
    \node[hidden] (c) at (1.8,0) {};
          \draw[line width=.3mm] (a) to (b);
    \draw[line width=.3mm] (a) to (5);
    \draw[line width=.3mm] (a) to (1);
    \draw[line width=.3mm] (b) to (4);
    \draw[line width=.3mm] (b) to (c);
    \draw[line width=.3mm] (2) to (c);
  \end{tikzpicture}
  \caption{(a) A tree with five leaves; (b) one of its subforests.}\label{fig:5leaves}
\end{figure}

The poset $\mathcal P_T$ comprises $34$ possible forests/models. In
Figures~\ref{fig:lattice-25}-\ref{fig:lattice-125},
%% , \ref{fig:lattice-75}, \ref{fig:lattice-125}
we display the lattice structure of $\mathcal P_T$ overlaid with a
heat map of how frequently the models were chosen at the particular
sample size.  The subforest/submodels are labeled from 1 to 34 with
$1$ corresponding to the complete independence model and $34$
corresponding to $\mathbf M(T)$, where $T$ is the tree in Figure
\ref{fig:5leaves}(a).  If we order the edges as $\{a,1\}$, $\{a,5\}$,
$\{a,b\}$, $\{b,4\}$, $\{b,c\}$, $\{c,2\}$, $\{c,3\}$ and use
$\{0,1\}$-vectors 
to indicate the presence of edges then the submodels are:
\begin{equation*}
\begin{array}{r}
\mathbf{ 1}:\quad 0 0 0 0 0 0 0\\
\mathbf{ 2}:\quad 1 1 0 0 0 0 0\\
\mathbf{ 3}:\quad 1 0 1 1 0 0 0\\
\mathbf{ 4}:\quad 0 1 1 1 0 0 0\\
\mathbf{ 5}:\quad 1 1 1 1 0 0 0\\
\mathbf{ 6}:\quad 1 0 1 0 1 1 0\\
\mathbf{ 7}:\quad 0 1 1 0 1 1 0\\
\mathbf{ 8}:\quad 1 1 1 0 1 1 0\\
\mathbf{ 9}:\quad 0 0 0 1 1 1 0\\
\end{array}\qquad
\begin{array}{r}
\mathbf{ 10}:\quad 1 1 0 1 1 1 0\\
\mathbf{ 11}:\quad 1 0 1 1 1 1 0\\
\mathbf{ 12}:\quad 0 1 1 1 1 1 0\\
\mathbf{ 13}:\quad 1 1 1 1 1 1 0\\
\mathbf{ 14}:\quad 1 0 1 0 1 0 1\\
\mathbf{ 15}:\quad 0 1 1 0 1 0 1\\
\mathbf{ 16}:\quad 1 1 1 0 1 0 1\\
\mathbf{ 17}:\quad 0 0 0 1 1 0 1\\
\mathbf{ 18}:\quad 1 1 0 1 1 0 1\\
\end{array}\qquad
\begin{array}{r}
\mathbf{ 19}:\quad 1 0 1 1 1 0 1\\
\mathbf{ 20}:\quad 0 1 1 1 1 0 1\\
\mathbf{ 21}:\quad 1 1 1 1 1 0 1\\
\mathbf{ 22}:\quad 0 0 0 0 0 1 1\\
\mathbf{ 23}:\quad 1 1 0 0 0 1 1\\
\mathbf{ 24}:\quad 1 0 1 1 0 1 1\\
\mathbf{ 25}:\quad 0 1 1 1 0 1 1\\
\mathbf{ 26}:\quad 1 1 1 1 0 1 1\\
\mathbf{ 27}:\quad 1 0 1 0 1 1 1\\
\end{array}\qquad
\begin{array}{r}
\mathbf{ 28}:\quad 0 1 1 0 1 1 1\\
\mathbf{ 29}:\quad 1 1 1 0 1 1 1\\
\mathbf{ 30}:\quad 0 0 0 1 1 1 1\\
\mathbf{ 31}:\quad 1 1 0 1 1 1 1\\
\mathbf{ 32}:\quad 1 0 1 1 1 1 1\\
\mathbf{ 33}:\quad 0 1 1 1 1 1 1\\
\mathbf{ 34}:\quad 1 1 1 1 1 1 1
\end{array}
\end{equation*}
In particular, the smallest true model is model $13$.

Figures~\ref{fig:lattice-25}-\ref{fig:lattice-125} show that the
standard dimension-based BIC tends to select too small models that do
not contain the data-generating distribution $q$.  In particular, BIC never
selects the full tree model 34.  The RLCT-based sBIC, on the other hand,
invokes a milder penalty, occasionally selects the tree model 34, and
more frequently selects the smallest true model 13.  Indeed, already
for $n=75$, sBIC selects the true model more often than any other
model. On the other hand, the regular BIC procedure selects too simple
a model also when the sample size is increased to $n=125$.

\input{./lattice_node_colors}

\begin{figure}[t!]
  \makebox[\textwidth]{
    \centerline{
      \scalebox{0.7}{\input{./bics_lattice_25}}
      \scalebox{0.7}{\input{./bic_lattice_25}}
    }}
  \caption{Results from $100$ simulations with true covariance
    matrix given by (\ref{eq:true5}) for $n=25$ (sBIC left, BIC right).
    Darker color corresponds to higher selection frequency. The square
    node $13$ is the smallest true model and includes the selection frequency. 
    Models never chosen are without border.}
  \label{fig:lattice-25}
\end{figure}
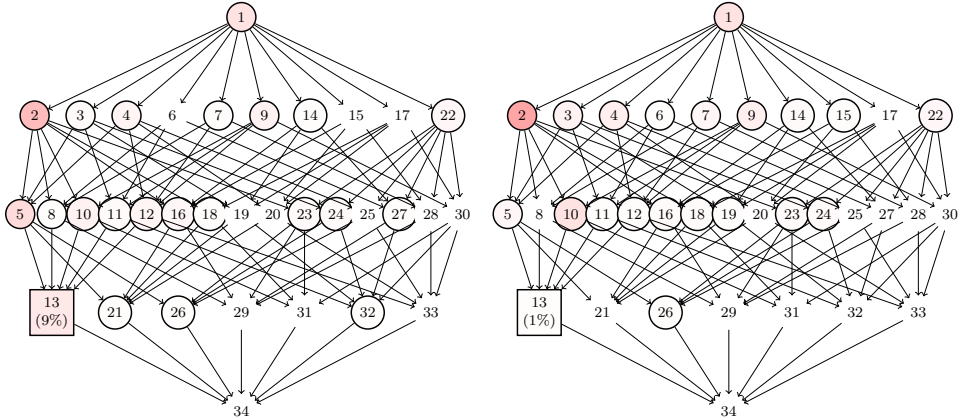
\begin{figure}[ht!p]
  \makebox[\textwidth]{
    \centerline{
      \scalebox{0.7}{\input{./bics_lattice_75}}
      \scalebox{0.7}{\input{./bic_lattice_75}}
    }}
  % \caption{Results from $100$ simulations with the true covariance
  %   matrix given by (\ref{eq:true5}) for $n=75$ (sBIC left, BIC right).
  %   Darker color corresponds to higher selection frequency. The square
  %   node $13$ is the smallest true model and includes the selection frequency.
  %   Models never chosen are without border.}
  \caption{Results from $100$ simulations as for
    Fig.~\ref{fig:lattice-25} but with sample size 
    $n=75$.}
  \label{fig:lattice-75}
\end{figure}
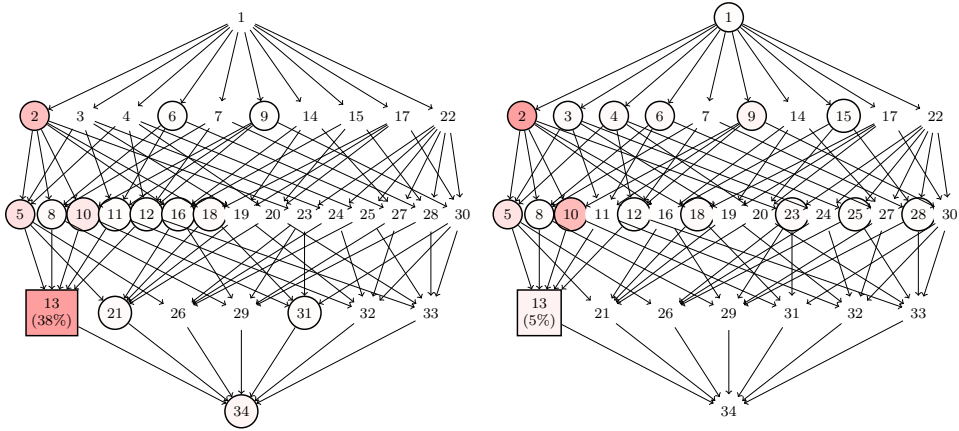
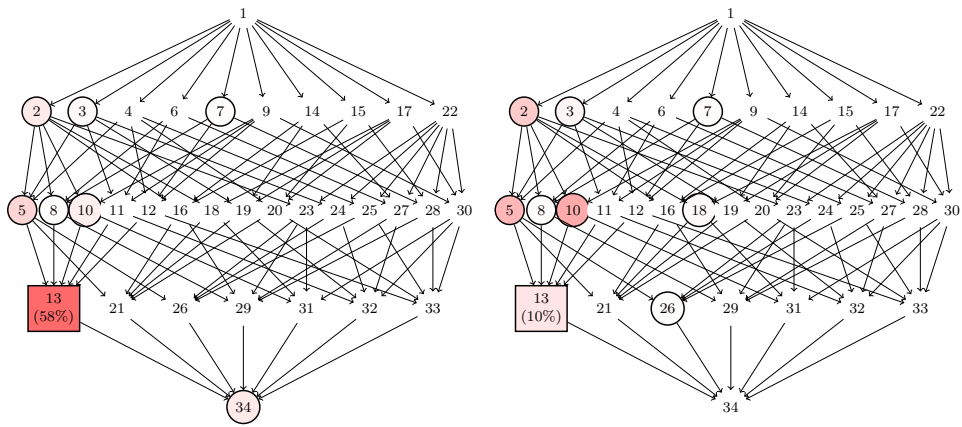
\begin{figure}[ht!p]
  \makebox[\textwidth]{
    \centerline{
      \scalebox{0.7}{\input{./bics_lattice_125}}
      \scalebox{0.7}{\input{./bic_lattice_125}}
    }}
  % \caption{Results from $100$ simulations with the true covariance
  %   matrix given by (\ref{eq:true5}) for $n=125$ (sBIC left, BIC right).
  %   Darker color corresponds to higher selection frequency. The square
  %   node $13$ is the smallest true model and includes the selection frequency.
  %   Models never chosen are without border.}
  \caption{Results from $100$ simulations as for
    Fig.~\ref{fig:lattice-25} but with sample size $n=125$.}
  \label{fig:lattice-125}
\end{figure}

Next, we consider examples with $10$ and $11$ leaves, in which case
the number of considered models is still tractable.  Writing $m:=|V|$
for the number of leaves, the lattice $\mathcal{P}_T$ has depth $m-1$
with the complete independence model having depth 0 and the
maximal element $\mathbf{M}(T)$ having depth $m-1$. Since the penalty in BIC is
always at least the penalty in sBIC, it holds trivially that BIC will
select the smallest true model more often than sBIC when the smallest
true model is at depth 0; the converse is true if the smallest true
model is at depth $m-1$.  We thus focus on the middle depth and randomly choose 50 trees $T_1,...,T_{50}$ with corresponding randomly chosen subforests $F_1,...,F_{50}$ each at depth
$\lfloor\frac{m-1}{2}\rfloor$. From each subforest which we pick $q_i\in\mathbf{M}(F_i)$ by setting
all edge correlations to 0.6 and all leaf variances to 1; note that
$F_i$ equals the $q$-forest $F^*(q_{i})$.  From each $q_i$, we generate a dataset of
a fixed size $n$ and compare the proportion of times that sBIC and BIC
correctly identify the smallest true model $\mathbf{M}(F_i)$ for $1\leq i\leq 50$.  The
results of these simulations are summarized in Figure
\ref{fig:big_selection_comparison}.  We see that sBIC outperforms BIC
for smaller sample sizes with BIC marginally overtaking sBIC in very
large samples.

\begin{figure}[t!]
              \begin{subfigure}[b]{0.49\textwidth}
                \includegraphics[width=\textwidth]{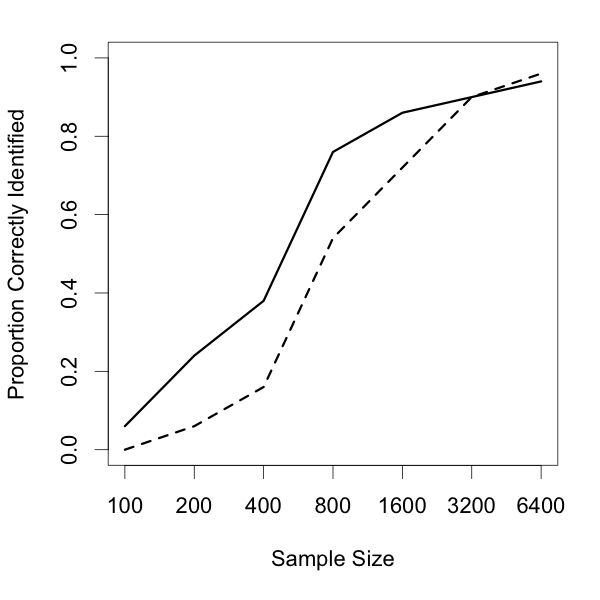}
                \caption{$m=10$ leaves}
                \label{fig:10-leaves-comparison}
        \end{subfigure}
      %%}
        \hfill
        \begin{subfigure}[b]{0.49\textwidth}
                \includegraphics[width=\textwidth]{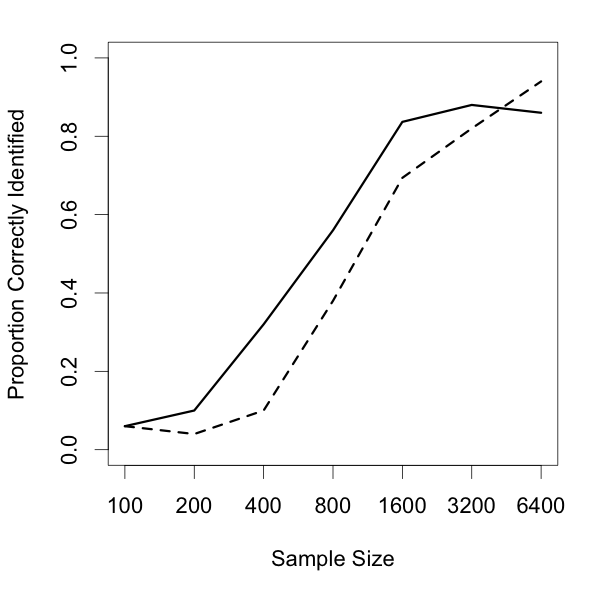}
                \caption{$m=11$ leaves}
                \label{fig:11-leaves-comparison}
        \end{subfigure}
       \centering
       \caption{Proportion of times, out of 50 simulations, sBIC
         (solid line) and BIC (dashed line) select the smallest true model for
         different samples sizes (displayed on a logarithmic
         scale). }
        \label{fig:big_selection_comparison} 
\end{figure}

\begin{rem}
  In the simulations, we evaluated the quality of the forests found by
  BIC and sBIC through the proportion of times the chosen forest
  matched the truth exactly.   An exact match is a very strong
  requirement and one may instead wish to compute the average
  distance, based on some metric, between selected forest and the
  truth. Unfortunately, the most natural metrics in our setting are
  NP-hard to compute and can only be approximated in general
  \cite{hein1996,hickey2008spr}. 
  % Hence, our focus on the strong
  % requirement of exact matching is motivated by computational
  % feasibility and interpretability of the results.
\end{rem}

\subsection{Temperature Data}

We consider a dataset consisting of average daily temperature values
on 310 days from 37 cities across North America, South America,
Africa, and Europe. The data was sourced from the National Climatic
Data Center and compiled in a readily available format by the average
daily temperature archive of the University of Dayton
\cite{dayton-temp}. In order to decorrelate and localize the data we
first perform a seasonality adjustment where we regress each observed
time series of temperature values on a sinusoid corresponding to the
seasons and retain only the residuals. We then consider only the
differences of average temperatures on consecutive days reducing the
number of data points to $n=309$.

In the simulations of Section \ref{sec:sim_study} we performed an
exhaustive search over the lattice of all considered forests, a
strategy which quickly becomes infeasible when increasing the number
of observed variables beyond the low teens. Thus, in order to do model
selection with the 37 observed nodes described above, we need to
formulate an approximate sBIC. There are a plurality of possible
heuristic strategies for producing this approximation involving
combinations of greedy search, truncation of the considered model
space, and simulated annealing. An in-depth exploration of these
strategies and their relative performance is beyond the scope of this
paper, instead we will show the results of using one such method as a
proof of concept.

Our selection strategy, which we call a pruned chain search, has the following form:

\begin{itemize}
	\item[(1)] Generate an approximate maximum likelihood
          trivalent tree structure $T$. 
	\item[(2)] Prune the model space of considered forests to only
          consider a single decreasing path in the poset $\cP_T$ starting at $T$ and ending with the empty forest.
	\item[(3)] Compute the sBIC (or BIC) for models in the pruned
          space and select the highest scoring model. 
\end{itemize}

Note that after (2) the number of considered models will equal to the number of observed variables making computation tractable for many observed nodes. We accomplish (1) using a version of the structural EM algorithm proposed by \cite{friedman2002structural}. To produce the decreasing path of models in (2) we start with $T$ and iteratively select subforests in a greedy fashion:

\begin{itemize} 
	\item[(a)] Suppose that after the $m$th iteration we have constructed the decreasing chain $\mathcal{C}_m$ of forests $T=F_0\supset F_1 \supset F_2 \supset ... \supset F_m$.
	\item[(b)] If $F_m$ is the empty forest then we are done.
	\item[(c)] Otherwise, we extend $\mathcal{C}_m$ to
          $\mathcal{C}_{m+1}$ by adding to it the forest with largest
          BIC-penalized log-likelihood (with log-likelihood maximized
          using the EM algorithm described in Section
          \ref{sec:sim_study}) among all maximal subforests of $F_m$.  
\end{itemize}

We present the results of applying above selection procedure to the
temperature data in Figure \ref{fig:temperature}. Note that the models
selected by the sBIC and BIC are quite similar with the majority of
the connections following our physical intuition that geographically
adjacent cities should have similar temperature fluctuations while
further separated cities should be essentially uncorrelated. For
instance, all three cities in Washington, USA are connected to each
other but to no other cities. The one difference between the the model
selected by sBIC and that selected by the BIC is the connection of
Barbados to the component containing the Bahamas in the sBIC
graph. The distance between these nodes is just far enough to place
this connection on the border between spurious and reasonable. As in
the simulation experiments, we observe sBIC's ability to select
larger models.

\begin{figure}[t!]
              \begin{subfigure}[b]{.85\textwidth}
              \vspace{-5mm}
                \includegraphics[width=\textwidth]{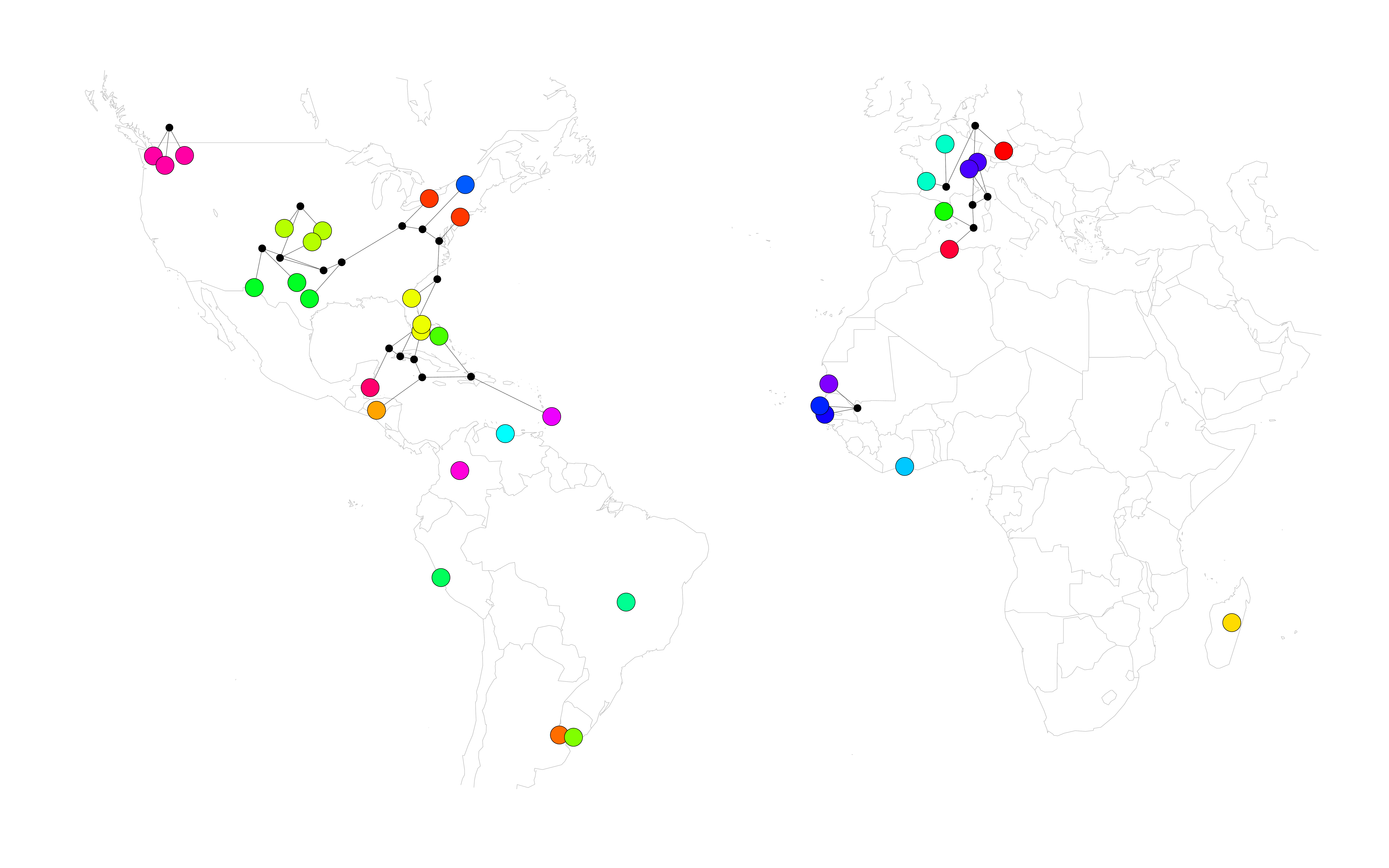}
                \vspace{-8mm}
                \caption{Model chosen by sBIC}
                \label{fig:sbic-map}
        \end{subfigure}
      %%}
        \begin{subfigure}[b]{.85\textwidth}
        \vspace{-3mm}
                \includegraphics[width=\textwidth]{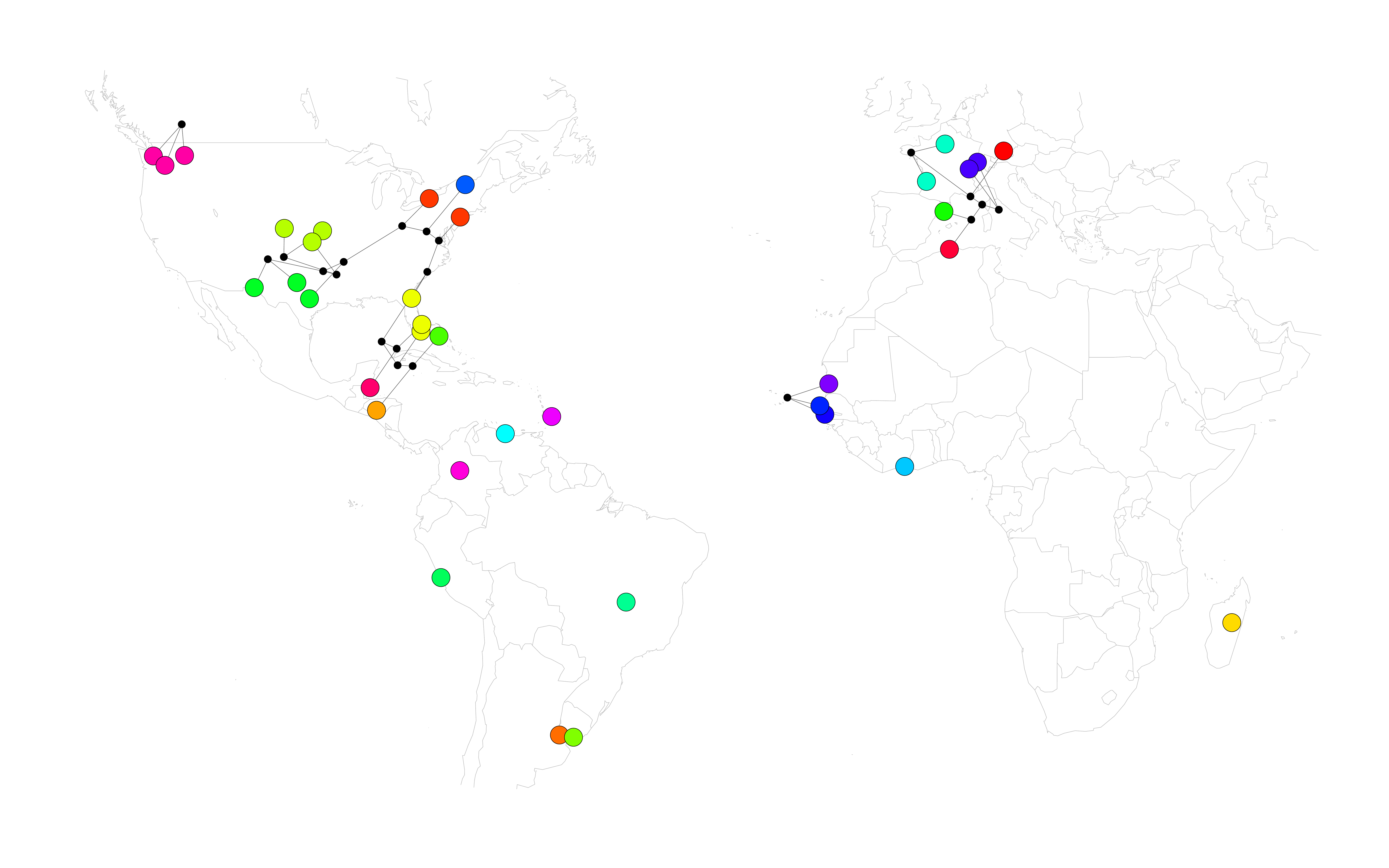}
                 \vspace{-8mm}
                \caption{Model chosen by BIC}
                \label{fig:bic-map}
        \end{subfigure}
       \centering
       
       \caption{The models selected by sBIC and BIC pruned chain search. Each colored node represents an observed node (nodes with the same color are from the same country or US state) and the black nodes correspond to latent variables. The position of colored nodes corresponds to the city from where the data was collected.}
        \label{fig:temperature} 
\end{figure}

\section{Conclusion}
\label{sec:conclusion}

Real log-canonical thresholds and associated multiplicities quantify
the large-sample properties of the marginal likelihood in Bayesian
approaches to model selection.  In this paper, we computed these RLCTs
for Gaussian latent tree and forest models; the main results being
Theorems~\ref{th:main2} and~\ref{th:main3}.  Our computations relied
on the fact that the considered tree and forest models have a monomial
parametrization, which allows one to apply methods from polyhedral
geometry that we presented in Theorem~\ref{th:main1}.

Knowing RLCTs makes it possible to apply a `singular Bayesian
information criterion' (sBIC) that was recently proposed by
\cite{drton:2013:sbic}.  RLCTs provide refined information about the
marginal likelihood and our simulations show that, at least in smaller
problems, the sBIC outperforms the usual BIC of
\cite{schwarz1978edm} that is defined using model dimension alone.
As an exhaustive search over all models becomes quickly infeasible as
the number of observed variables increases, we demonstrated, by example
of a temperature dataset, how the sBIC might be approximated 
and applied to larger problems. In particular, we combined the structural EM 
of \cite{friedman2002structural} with a greedy search methodology to reduce
the number of considered models to a small collection for which the sBIC
can be readily computed.

\appendix

\section{Proof of Theorem \ref{th:main1}}
\label{sec:proof-th:main1}

Let $H$ be the function from~(\ref{eq:sos}).  By assumption, the
`prior' $\varphi:\Omega\to(0,\infty)$ is bounded above and
$\mathcal{V}_\Omega(H)=\{\omega\in\Omega:H(\omega)=0\}$ is compact.
Since $\varphi$ is smooth and positive, $\varphi$ is bounded away from
zero on $\mathcal{V}_\Omega(H)$ and any compact neighborhood of this
zero set.  The poles of the zeta function in (\ref{eq:zeta}) can
be shown to be the same for all such choices of $\varphi$, and we have
${\rm RLCT}_\Omega(H;\varphi)={\rm RLCT}_\Omega(H)$.

Our proof of Theorem~\ref{th:main1} now proceeds in three steps:
\begin{itemize}
\item[Step 1.] Show that ${\rm RLCT}_\Omega(H)={\rm
    RLCT}_{\Omega}(H^0+H^1)$, where $H^0$, $H^1$ are the zero and
  nonzero parts of $H$ that are defined in (\ref{eq:f1}) and
  (\ref{eq:f0}).
\item[Step 2.] Show that  ${\rm RLCT}_{\Omega_1}(H^1)=(\lambda_1,1)$, where $\lambda_1=\codim \mathcal{V}_{\Omega_1}(H^1)$.
\item[Step 3.] Show that
  ${\rm RLCT}_{\Omega_0}(H^0)=(\lambda_0,\mathfrak{m})$, where
  $\lambda_0$ is the $\mathbf{1}$-distance of the Newton polyhedron
  $\Gamma_+(H^0)$ and $\mathfrak{m}$ is the multiplicity (recall
  Definition \ref{def:NP}).
\end{itemize}
Since $H^0$ and $H^1$ are functions of disjoint sets of coordinates
and $\Omega=\Omega_0\times \Omega_1$ is a Cartesian product, it
follows from Remark 7.2(3) in \cite{watanabe_book} and the above Steps
1-3 that
\[
{\rm
  RLCT}_{\Omega}(H^0+H^1)\;=\;(\lambda_0+\lambda_1,(\mathfrak{m}+1)-1)
\;=\;(\lambda_0+\lambda_1,\mathfrak{m}),
\]
which is the claim of Theorem~\ref{th:main1}.

Before moving on to Step 1 we make a definition.  Let
$f,g:\,\Omega\to [0,\infty)$ be two nonnegative functions with common
zero set $\mathcal{V}_\Omega(g)=\mathcal{V}_\Omega(f)$.  Then $f$ and
$g$ are \emph{asymptotically equivalent}, we write $f\sim g$, if there
exist two constants $c,C>0$ and a neighborhood $W$ of
$\mathcal{V}_\Omega(g)=\mathcal{V}_\Omega(f)$ such that
\begin{equation}\label{eq:equivalenceAux}
c f(\omega)\;\leq \; g(\omega)\;\leq\; C f(\omega)
\end{equation}
for all $\omega\in W\cap\Omega$.  Note that $\sim$ is indeed an
equivalence relation.  According to Remark 7.2(1) in
\cite{watanabe_book}, $f\sim g$ implies
${\rm RLCT}_\Omega(f)={\rm RLCT}_\Omega(g)$.

\subsection{Step 1}
First, note that ${\rm RLCT}_\Omega(H)={\rm RLCT}_{W\cap \Omega}(H)$
for any neighborhood $W$ of the compact zero set
$\mathcal{V}_\Omega(H)$.  Choose $W$ sufficiently small such that
$\omega_1,\ldots, \omega_s$ are bounded away from zero on
$W\cap\Omega$.  Next, by definition of the index $r$ in
Section~\ref{sec:monred}, we have that $H=H^1+H^{01}$, where
\[
H^{01}=\sum_{i=r+1}^k \omega^{2u_i}.
\]
When viewed as functions restricted to $W\cap \Omega$, we have
$H^0\sim H^{01}$ because
$$
\bigg(\min_i \prod_{j=1}^s \omega_j^{2u_{ij}}\bigg)\, H^0
\;\leq\; H^{01}\;\leq\; \bigg(\max_i \prod_{j=1}^s
\omega_j^{2u_{ij}}\bigg)\, H^0 
$$
and $\omega_1,\ldots, \omega_s$ are bounded above and bounded away
from zero on the compactum $W\cap\Omega$.  It follows that
${\rm RLCT}_\Omega(H)={\rm RLCT}_{\Omega}(H^0+H^1)$ because
$H^{01}\sim H^0$ implies that $H^1+H^0\sim H^1+H^{01}=H$.

\subsection{Step 2}
To complete Step 2 we will prove the following result.

\begin{prop}
  \label{prop:cidim}
  Suppose that $H$ satisfies (\ref{eq:sos}) with all $c_i^*\neq 0$,
  i.e., $H$ is equal to its nonzero part. Let
  $\mathcal{V}_{\Omega}(H)$ be the zero set of $H$ on $\Omega$. Then
  $$
  \RLCT_{\Omega}(H) = (\codim \mathcal{V}_\Omega(H), 1).
  $$
\end{prop}

Before turning to the proof, we exemplify the application of
Proposition~\ref{prop:cidim}.

\begin{exmp}
  Let $\Omega = [0,1] \times [0,1]$ be the unit square in $\R^2$, and
  consider two functions $g_1(\omega)=(\omega_1-\omega_2)^2$ and
  $g_2(\omega)=(\omega_1+\omega_2)^2$.  The zero set of either
  function is a line in $\R^2$.  When restricting to $\Omega$, the
  zero set $\mathcal{V}_\Omega(g_1)$ is a line segment and of
  codimension one.  The zero set $\mathcal{V}_\Omega(g_2)$, on the
  other hand, consists only of the origin and is of codimension two.
  We have $\RLCT_\Omega(g_1) = (1,1)$ but $\RLCT_\Omega(g_2) = (2,1)$.
\end{exmp}

To prove Proposition~\ref{prop:cidim}, note first that when $H$ is
equal to its nonzero part and $\mathcal{V}_\Omega(H)$ is compact,
$\RLCT_\Omega(H)$ is equal to the RLCT of $H$ over a compact set on
which all coordinates of the argument $\omega$ are bounded away from
zero.  Partition this compactum into the intersections with each one
of the $2^d$ orthants in $\R^d$.  Then $\RLCT_\Omega(H)$ is
the minimum RLCT in any orthant.  Similarly, the
codimension of $\mathcal{V}_\Omega(H)$ is the minimum of any
codimension obtained from intersection with an orthant.  We may thus
consider one orthant at a time.  Changing signs as needed to make all
coordinates positive, the following lemma becomes applicable.

\begin{lem}
  \label{lem:tologs}
  Let $W=[a_1,b_1]\times\cdots\times [a_s,b_s]$ with
  $0<a_i<b_i<\infty$. Let
  $\log W=[\log a_1,\log b_1]\times\cdots\times [\log a_s,\log b_s]$.
  If $H$ satisfies (\ref{eq:sos}) with all
  $c_{i}^*>0$ and $\mathcal{V}_W(H)$ is nonempty, then
  $$
  {\rm RLCT}_W(H)= {\rm
    RLCT}_{\log W}\left(\sum_{i=1}^r\left(u_{i}^T
  \omega-\log c_i^*\right)^2\right).
  $$
\end{lem}
The result follows from a change of coordinates and an argument about
asymptotic equivalence that has been used in other contexts.
% ; see
% e.g.~\cite[Proposition 3.4]{shaowei-useful}.  
We include the proof of
the lemma for sake of completeness.

\begin{proof}
  Change coordinates via the substitution $\tilde \omega=\log(\omega)$,
  where the logarithm is applied entry-wise.  Since the Jacobian of
  this transformation is bounded above and bounded away from zero on
  $W$, it may be ignored in the computation of the RLCT and thus
  \[
  {\rm RLCT}_W(H)= {\rm
    RLCT}_{\log W}\left(\sum_{i=1}^r\left(e^{u_{i}^T \omega}-
  e^{\log c_i^*}\right)^2\right).
  \]
  Since $W$, and thus also $\log W$, is compact, each of the $r$
  linear combinations $u_i^T \omega$ takes its values in a compact
  set.  Restricted to this compact set, the function
  \[
  h_1(x) \;=\;\sum_{i=1}^r \left(e^{x_i}-
  e^{\log c_i^*}\right)^2
  \]
  is asymptotically equivalent to the sum of squares
  \[
  h_2(x) \;=\;\sum_{i=1}^r \left(x_i-
  \log c_i^*\right)^2,
  \]
  as can been seen by a quadratic Taylor approximation to $h_1$ around
  the point $(\log c_1^*,\dots,\log c_r^*)$.  Since asymptotically
  equivalent functions have the same RLCT, the claim is proven.
\end{proof}

By an application of Lemma~\ref{lem:tologs}, the proof of Proposition
\ref{prop:cidim} reduces to an analysis of sums of squares of linear
forms, that is, functions of the form
\begin{equation}
  \label{eq:sos-linear-forms}
  g(\omega)\;=\;\sum_{i=1}^{r}
  (u_{i}^{T}\omega-C^*_i)^2
\end{equation}
with $C_i^*\in \R$ and $u_{i}\in \R^{d}$.   Proposition
\ref{prop:cidim} thus follows from
Proposition~\ref{thm:LinearInterval} below.  Note that 
$\mathcal{V}_\Omega(g)$ is a polyhedron, which we assume to
be nonempty.

\begin{prop} \label{thm:LinearInterval} If $g:\Omega\to[0,\infty)$ is
  a sum of squares of linear forms as in~(\ref{eq:sos-linear-forms})
  and $\Omega$ is a product of intervals, then
  $\RLCT_\Omega(g) = (\codim \mathcal{V}_\Omega(g), 1)$.
\end{prop}
\begin{proof}
  By \cite[Prop.~2.5, Prop.~3.2]{shaowei_rlct}, or also
  \cite[Remark 2.14]{watanabe_book}, $\RLCT_\Omega(g)$ is the minimum
  of local thresholds $\RLCT_{\Omega(x)}(g)$ over
  $x \in \mathcal{V}_\Omega(g)$.  Here, each set
  $\Omega(x)=W(x)\cap\Omega$, where $W(x)$ is a sufficiently small
  neighbhorhood of $x$.  We will show that
  $\RLCT_{\Omega(x)}(g)=(\codim \mathcal{V}_\Omega(g), 1)$ for $x \in
  \mathcal{V}_\Omega(g)$, which implies our claim.

  Consider any point $x \in \mathcal{V}_\Omega(g)$.  By translation,
  we may assume without loss of generality that $x=0$ and
  $g(\omega)=\sum_{i} (u_{i}^{T}\omega)^{2}$.  We may then take the
  neighborhood $\Omega(0)$ to be equal to
  $\{\omega \in \Omega: \max_i |\omega_i| \leq \varepsilon\}$ for
  sufficiently small $\varepsilon > 0$.

  When partitioning $\Omega(0)$ into orthants, the codimension of
  $\mathcal{V}_\Omega(g)$ is the minimum of the codimensions of the
  intersection between $\mathcal{V}_\Omega(g)$ and each one of the
  orthants.  Furthermore, $\RLCT_{\Omega(0)}(g)$ is equal to the
  smallest RLCT of $g$ over any of these orthants.  Therefore,
  changing the signs of the coordinates $\omega_i$ as needed, we are
  left with checking that $\RLCT_{\Omega_+}(g)$ is given by the
  codimension of $\mathcal{V}_{\Omega_+}(g)$ for
  $\Omega_+ = \{\omega \in \Omega : 0 \leq \omega_i \leq \varepsilon
  \text{ for all }i=1,\dots, d\}$
  and $g(\omega)=\sum_{i} (u_{i}^{T}\omega)^{2}$.

  {\em Case 1.} If $\cV_{{\Omega_{+}}}(g)$ intersects the interior of
  $\Omega_{+}$, then we may pick any point $x_{+}$ in this
  intersection and consider $\Omega_{+}$ as a neighborhood of $x_{+}$.
  After a change of coordinates, we have
  $g(\omega)=\omega_1^2+\dots+\omega_s^2$, where $s$ is the
  codimension of $\mathcal{V}_{\Omega_+}(g)$.  By
  Example~\ref{ex:reg-case}, $\RLCT_{\Omega_{+}}(g)=(s,1)$, which was
  to be shown.

  {\em Case 2.} Suppose now that $\cV_{{\Omega_{+}}}(g)$ is contained
  in the boundary of ${\Omega_{+}}$. Since the zero set of $g$ on all
  of $\R^d$ is a linear space, $\cV_{{\Omega_{+}}}(g)$ is in fact a
  face of ${{\Omega_{+}}}$, and each $u_{i}^{T}\omega$ is a supporting
  hyperplane of $\Omega_{+}$.  In particular, after appropriate sign
  changes, we may assume that $u_{i}^{T}\omega\geq 0$ on $\Omega_{+}$.
  The codimension of $\cV_{{\Omega_{+}}}(g)$ is equal to the number,
  say $s$, of facets of $\Omega_+$ containing it.  Without loss of
  generality, we may assume that these facets are given by
  $\omega_{1}=0$, $\omega_{2}=0$, \ldots, $\omega_{s}=0$.  This
  implies that all $u_{i}$ have nonzero entries only in the first $s$
  coordinates.  We now show that when restricted to $\Omega_+$,
  the functions $g(\omega)$ and
  $f(\omega)=\omega_1^2+\dots+\omega_s^2$ are asymptotically
  equivalent; recall~(\ref{eq:equivalenceAux}).

  To show that on $\Omega_+$, the function $g$ can be bounded from
  below by a positive multiple of $f$, note that the fact that
  $u_{i}^{T}\omega\geq 0$ on $\Omega_+$ implies that all $u_{i}$ have
  nonnegative entries.  Hence,
  $$
  \sum_{i=1}^{r} (u_{i}^{T}\omega)^{2}\;=\;
  \sum_{i=1}^{r}\Bigg(\sum_{j=1}^{s}u_{ij}\omega_{j}\Bigg)^{2}
  \;\geq\; \sum_{j=1}^{s} \left(\sum_{i=1}^{r}
    u_{ij}^{2}\right)\omega_{j}^{2}, 
  $$
  where the inequality is obtained by expanding squares and
  dropping the mixed terms, which are nonnegative.  If
  $\sum_{i=1}^{r} u_{ij}^{2}=0$ for some index $j$ then $u_{ij}=0$ for
  all $i$, which contradicts the fact that $\omega_{j}=0$ for all
  $\omega\in\cV_{\Omega_+}(g)$.  Thus,
  \[
  c \;=\; \min \left\{ \sum_{i=1}^{r} u_{ij}^{2}: 1\le j\le s\right\} \;>\; 0,
  \]
  and $g(\omega)\ge c f(\omega)$ for
  all $\omega\in\Omega_+$.

  To prove that $g$ can be bounded above by a multiple of $f$, note
  that all $u_{i}^{T}\omega$ are nonnegative on $\Omega_+$ and thus
  $$
  \sum_{i=1}^{r} (u_{i}^{T}\omega)^{2} \;\leq\; \left(\sum_{i=1}^{r}
    u_{i}^{T}\omega\right)^{2}\;=\;\Bigg(\sum_{i=1}^{r}
  \sum_{j=1}^{s} u_{ij}\omega_{j}\Bigg)^{2}.  
  $$
  Let $u_{+j}=\sum_{i}u_{ij}$ and $u_{++}=\sum_{j}u_{+j}$.  Then,
  since all $u_{i}$ have nonnegative entries, Jensen's inequality
  implies that
  $$
  \Bigg(\sum_{i=1}^{r}
  \sum_{j=1}^{s} u_{ij}\omega_{j}\Bigg)^{2} \;=\;
  u_{++}^{2}\Bigg(\sum_{j=1}^{s}\frac{u_{+j}}{u_{++}}\omega_{j}\Bigg)^{2}\;\leq\;
  u_{++}\,\max\{u_{+j}: 1\le j\le s\}\sum_{i=1}^{s}\omega_{i}^{2}.
  $$
  
  Since $g$ is asymptotically equivalent to
  $f(\omega)=\omega_1^2+\dots+\omega_s^2$,
  we have $\RLCT_{\Omega_+}(g)=\RLCT_{\Omega_+}(f)$.  Let
  $\Omega_+'=[-\epsilon,\epsilon]^s\times [0,\epsilon]^{d-s}$.  Then
  $$
  \int_{\Omega_+} (\omega_{1}^{2}+\ldots+\omega_{s}^{2})^{-z/2}\,
    d\omega \;=\; 2^{-s}\int_{\Omega_+'}
  (\omega_{1}^{2}+\ldots+\omega_{s}^{2})^{-z/2}\,d\omega.
  $$
  Hence, $\RLCT_{\Omega_+}(f)=\RLCT_{\Omega_+'}(f)$.  From Case 1, we
  know that $\RLCT_{\Omega_+'}(f)\;=\;(s,1)$.  Putting it all
  together, we have shown that $\RLCT_{\Omega_+}(g)=(s,1)$.
\end{proof}

\subsection{Step 3}

The remaining step amounts to proving the following result, which
concerns the case where the considered function $H$ is equal to its
zero part.

\begin{prop}
  \label{prop:sing}
  Let $\Omega$ be a compact product of intervals containing the
  origin, and let $\Gamma_+(H)$ be the Newton polyhedron of the
  function $H(\omega)=\sum_i \omega^{2u_i}$.  Then
  $$
  \RLCT_\Omega(H)=(\lambda,\mathfrak{m}),
  $$
  where $1/\lambda$ is the $\mathbf 1$-distance of $\Gamma_+(H)$ and
  $\mathfrak{m}$ is its multiplicity.
\end{prop}
\begin{proof}
  Note that $H$ is invariant under sign changes.  Hence,
  $\RLCT_{\Omega}(H)=\RLCT_{\Omega'}(H)=\RLCT_{\Omega\cup\Omega'}(H)$
  when $\Omega'$ is obtained from $\Omega$ by changing the signs of
  any subset of the coordinates $\omega_1,\dots,\omega_d$.  Forming
  the unions of $\Omega$ and its reflected versions shows that in
  order to prove Proposition~\ref{prop:sing}, we may assume that the
  origin is an interior point of $\Omega$.  The claim now follows from
  Theorem 8.6 in \cite{arnold1985singularities}, see also
  \cite[Section 4]{shaowei_rlct}, and by Remark~\ref{rem:deepest} below.
\end{proof}

\begin{rem}\label{rem:deepest}
  When the origin is in the interior of $\Omega$, the function
  $H(\omega)=\sum_i \omega^{2u_i}$ has
  $\RLCT_\Omega(H)=\RLCT_{\Omega(0)}(H)$ for any small neighborhood
  $\Omega(0)$ of the origin.  Indeed, as mentioned in the proof of
  Proposition~\ref{thm:LinearInterval}, $\RLCT_\Omega(H)$ is the
  minimum of local RLCTs of $H$ in small neighborhoods $\Omega(x)$ of
  points $x \in \Omega$.  If $x\not=0$, then some of the variables,
  say $\omega_1, \ldots, \omega_s$, are bounded away from zero on a
  sufficiently small neighborhood $\Omega(x)$.  Substituting these
  variables by $\omega_1 - x_1, \ldots, \omega_s - x_s$, respectively,
  in $H$, we get a new function $H_x$ for which
  $\RLCT_{\Omega(0)}(H)=\RLCT_{\Omega(x)}(H_x)$.  Now, $0 \le H_x \le H$
  near $x$.  Consequently,
  $\RLCT_{\Omega(x)}(H_x) \leq \RLCT_{\Omega(x)}(H)$.  We conclude
  that $\RLCT_{\Omega(0)}(H)\le \RLCT_{\Omega(x)}(H)$.
\end{rem}

\section{Proof of Theorem \ref{th:main2}}
\label{sec:app:tree-proofs}

Let $T=(U,E)$ be a tree with set of leaves $V$, and   let $q$ be a
distribution in the latent tree model $\mathbf{M}(T)$, which
has parameter space $\Omega=(0,\infty)^V\times [-1,1]^{E}$.  We are to
compute $\RLCT_\Omega(H_q)$ for the function $H_q$
from~(\ref{eq:tree:Hq}), where $\omega_v^*$ and $\rho_{vw}^*$ are the
variances and correlations of the distribution $q$. The basic idea of this proof follows \cite{pwz-2010-bic}.

First, observe that
Theorem~\ref{th:main1} is applicable to this problem.  Indeed, $H_q$
has the form from~(\ref{eq:sos}) and the $q$-fiber
$\mathcal{V}_\Omega(H_q)$ is compact.  Compactness holds because
$H_q(\omega)=0$ implies that $\omega_v=\omega_v^*$ for all $v\in V$,
and all edge correlations $\omega_e$, $e\in E$, are in the compact
interval $[-1,1]$.

Now, let $F^*:= F^{*}(q)=(U^*,E^*)$ be the $q$-forest, and let
$H_{q}^1(\omega_1,\ldots,\omega_s)$ be the nonzero part of $H_q$ given
in (\ref{eq:f1tree}).  The set $\mathcal{V}_{\Omega_1}(H_{q}^1)$ is
equal to the $q$-fiber under the model $\mathbf{M}(F^*)$; recall that
$\Omega_1$ is the projection of $\Omega$ onto the first $s$
coordinates.  We deduce that
$\codim \mathcal{V}_{\Omega_1}(H_{q}^1)=\dim \mathbf{M}(F^*)$, which
gives the value of $\lambda_1$ in Theorem \ref{th:main1}.  It
remains to show that the zero part
$H_{q}^0(\omega_{s+1},\ldots,\omega_d)$ defined in (\ref{eq:f0tree}) satisfies
\begin{equation}
  \label{eq:tree-zero-part}
(\lambda_0,\mathfrak{m})\;=\;{\rm
  RLCT}_{\Omega_0}(H_q^0)\;=\;\left(\frac{1}{2}\sum_{e\in E_0}
  w(e),1+l_{2}'\right), 
\end{equation}
where $\Omega_0$ is the projection of
$\Omega$ onto the last $d-s$ coordinates, $E_0=E\setminus E^*$ is
the set of edges that appear in $T$ but not in $F^*$, and $l_{2}'$ is the number of degree two nodes of $T$ that are not in  $U^{*}$.  

The zero part of $H_q$ is the sum of squares of the monomials
\begin{equation}\label{eq:gensI0}
  \prod_{e\in \overline{vw}\cap E_0} \omega_{e}, \qquad v,w\in V, \; v\not\sim w;
\end{equation}
recall that $v\not\sim w$ if there is no path between $v$ and $w$ in
the $q$-forest $F^*=(U^*,E^*)$.  The edge set $E_0$ can be partitioned
into sets $E_{01},\dots,E_{0t}$ such that each $E_{0i}$ defines a tree
$S_i=(U_i,E_{0i})$ that has the set of nodes $L_i:=U_i\cap U^*$ as
leaves.  In other words, the set of leaves $L_i$ of tree $S_i$
comprises precisely those nodes that belong to both $S_i$ and the
$q$-forest $F^*$.    For
example, in Figure \ref{fig:quartet}, we have $t=1$ and $S_1$ is the
tree with one inner node $b$ and three leaves $a,3,4$.  As a further
example, consider the tree and $q$-forest in
Figure~\ref{fig:star}(a) and (b), for which we form two subtrees $S_1$ and
$S_2$ with edge sets $E_{01}=\{\{a,3\}\}$ and $E_{02}=\{\{a,4\}\}$, as
shown in Figure~\ref{fig:star}(c) and (d).  In this second example, the
sets of leaves are $L_1=\{a,3\}$ and $L_2=\{a,4\}$, illustrating that
the sets $L_1,\dots,L_t$ need not be disjoint.

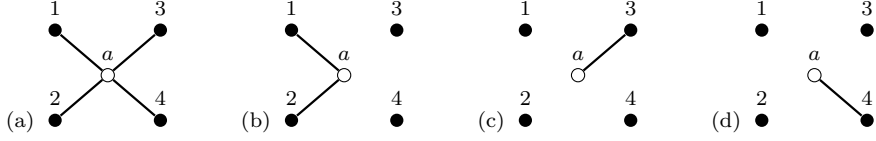
\begin{figure}[t!]
\centering
\tikzstyle{vertex}=[circle,fill=black,minimum size=5pt,inner sep=0pt]
\tikzstyle{hidden}=[circle,draw,minimum size=5pt,inner sep=0pt]
  (a)
  \begin{tikzpicture}
  \node[vertex] (1) at (-.7,.6)  [label=above:$1$] {};
    \node[vertex] (2) at (-.7,-.6)  [label=above:$2$] {};
    \node[vertex] (3) at (.7,.6) [label=above:$3$]{};
    \node[vertex] (4) at (.7,-.6) [label=above:$4$]{};
    \node[hidden] (a) at (0,0) [label=above:$a$] {};
    \draw[line width=.3mm] (a) to (1);
    \draw[line width=.3mm] (a) to (2);
    \draw[line width=.3mm] (a) to (3);
    \draw[line width=.3mm] (a) to (4);
  \end{tikzpicture}\qquad\quad
  (b)
  \begin{tikzpicture}
  \node[vertex] (1) at (-.7,.6)  [label=above:$1$] {};
    \node[vertex] (2) at (-.7,-.6)  [label=above:$2$] {};
    \node[vertex] (3) at (.7,.6) [label=above:$3$]{};
    \node[vertex] (4) at (.7,-.6) [label=above:$4$]{};
    \node[hidden] (a) at (0,0) [label=above:$a$] {};
    \draw[line width=.3mm] (a) to (1);
    \draw[line width=.3mm] (a) to (2);
  \end{tikzpicture}\qquad\quad
  (c)
  \begin{tikzpicture}
  \node[vertex] (1) at (-.7,.6)  [label=above:$1$] {};
    \node[vertex] (2) at (-.7,-.6)  [label=above:$2$] {};
    \node[vertex] (3) at (.7,.6) [label=above:$3$]{};
    \node[vertex] (4) at (.7,-.6) [label=above:$4$]{};
    \node[hidden] (a) at (0,0) [label=above:$a$] {};
    \draw[line width=.3mm] (a) to (3);
  \end{tikzpicture}\qquad\quad
  (d)
  \begin{tikzpicture}
  \node[vertex] (1) at (-.7,.6)  [label=above:$1$] {};
    \node[vertex] (2) at (-.7,-.6)  [label=above:$2$] {};
    \node[vertex] (3) at (.7,.6) [label=above:$3$]{};
    \node[vertex] (4) at (.7,-.6) [label=above:$4$]{};
    \node[hidden] (a) at (0,0) [label=above:$a$] {};
    \draw[line width=.3mm] (a) to (4);
  \end{tikzpicture}

  \caption{(a) Star tree; (b)  $q$-forest when
    $\rho_{12}^*$ is the only nonzero correlation; (c),(d)
    subtrees formed from the removed edges.}\label{fig:star}
\end{figure}

Consider now the function $\tilde H_{q}^0$ given by the sum of squares of
the monomials
\begin{equation}\label{eq:mingensI0}
  \prod_{e\in \overline{uu'}} \omega_{e}, \qquad i\in [t],\; u,u'\in L_i,\, u\not=
  u',
\end{equation}
where $[t]=\{1,\dots,t\}$ and $\overline{uu'}$ refers to the unique
path between $u$ and $u'$ in tree $S_i$.  Each monomial
listed in (\ref{eq:mingensI0}) is also listed in (\ref{eq:gensI0}).
To see this, observe that two distinct nodes $u,u'\in L_i$ belong to
distinct connected components in $F^*$.  If we take $v\in V$ from one
of the two connected components and $w\in V$ from the other, then the
monomial they define in (\ref{eq:gensI0}) is equal to the monomial
that $u$ and $u'$ define in (\ref{eq:mingensI0}).  Moreover, by the
definition of the trees $S_i$, every monomial listed in
(\ref{eq:gensI0}) is the product of monomials from
(\ref{eq:mingensI0}).  It follows that the Newton polyhedra
$\Gamma_+(H_{q}^0)$ and $\Gamma_+(\tilde H_{q}^0)$ are equal and hence
${\rm RLCT}_{\Omega_0}(H_{q}^0)={\rm RLCT}_{\Omega_0}(\tilde
H_{q}^0)$ (c.f.~Proposition~\ref{prop:sing}).

Let $f_i$ be the sum of squares of the monomials in
(\ref{eq:mingensI0}) that are associated with pairs of distinct nodes
$u$ and $u'$ in the set of leaves $L_i$ of the tree $S_i$.  No two
trees $S_i$ and $S_j$ for $i\not=j$ share an edge.  Hence, the two
sums of squares $f_i$ and $f_j$ depend on different subvectors of
$\omega$.  Since $\tilde H_{q}^0= f_{1}+\cdots+ f_{t}$, it follows
from~(\ref{eq:Zn-asy}) that
\begin{equation}
  \label{eq:rlct-sum-of-fi}
\RLCT_{\Omega_0}(H_{q}^0)=\sum_{i=1}^t \RLCT_{\Omega_0}( f_{i})-(0,t-1);
\end{equation}
see also Remark 7.2(3) in \cite{watanabe_book}.  If $T$ has no nodes
of degree two, i.e., $l_2=l_{2}'=0$, then the same is true for the each tree
$S_i$.  Lemma~\ref{lem:single-tree-S} below then implies that
\begin{equation}
  \label{eq:rlct-single-tree}
  \RLCT_{\Omega_0}( f_{i}) =\left( \frac{|L_i|}{2}, 1\right).
\end{equation}
Since the nodes in $L_i$ lie in $F^*$, we have
\begin{equation}
  \label{eq:adding-leaves}
\sum_{i=1}^t |L_i|=\sum_{e\in E_0} w(e),
\end{equation}
where $w(e)\in \{0,1,2\}$ is the number of nodes of $e$ that lie in
the $q$-forest $F^*$.
Combining~(\ref{eq:rlct-sum-of-fi})-(\ref{eq:adding-leaves}), we
obtain~(\ref{eq:tree-zero-part}) and have thus proven Theorem
\ref{th:main2} in the case of $l_2=0$ nodes of degree two.  The case
with nodes of degree two follows the same way applying
Lemma~\ref{lem:single-tree-S-deg2} instead of
Lemma~\ref{lem:single-tree-S}.

\begin{lem}
  \label{lem:single-tree-S}
  Let $S=(V,E)$ be a tree with set of leaves $L$ and all inner nodes
  of degree at least three. Let $f$ be the sum of squares of the
  monomials
  \begin{equation}\label{eq:simpleSS}
    \prod_{e\in \overline{vw}} \omega_{e}, \qquad v,w\in L,\; v\not= w.
  \end{equation}
  If $\Omega$ is a neighborhood of the origin, then 
  \[
  \RLCT_\Omega(f)=\left(\frac{|L|}{2},1\right). 
  \]
\end{lem}
\begin{proof}
  If $|L|=2$, then $S$ has a single edge and no inner nodes.  In this
  case, $f$ is the square of a single variable and it is clear
  $\RLCT_\Omega(f)=(1,1)=(|L|/2,1)$.  In the remainder of this proof,
  we assume that $|L|\ge 3$.

  By Proposition \ref{prop:sing}, it suffices to compute the
  $\mathbf 1$-distance and its multiplicity for the Newton polyhedron
  $\Gamma_+(f)\subset \R^{E}$.  By Definition~\ref{def:NP}, the
  polyhedron $\Gamma_+(f)$ is determined by the exponent vectors of
  the monomials in (\ref{eq:simpleSS}).  Each exponent vector is the
  incidence vector for a path between a pair of leaves.  In other
  words, each pair of two distinct leaves $v$ and $w$ defines a vector
  $u\in\R^E$ with $u_e=1$ if $e\in \overline{vw}$ and $u_e=0$
  otherwise.  Write $\mathcal{U}$ for the set of all these
  $\binom{|L|}{2}$ vectors.

  Let $E_L$ be the set of terminal edges of $S$, i.e., the $|L|$ edges
  that are incident to a leaf.  We claim that every point $x$ in the
  Newton polyhedron $\Gamma_+(f)$ satisfies
  \begin{equation}
    \label{eq:facet}
    \sum_{e\in E_L}x_e\geq 2
  \end{equation}
  and that the inequality defines a facet of
  $\Gamma_+(f)$.  Indeed, if $x\in\mathcal{U}$ then
  $\sum_{e\in E_L}x_e= 2$ because every path between two leaves in $L$
  includes precisely two edges in $E_L$.  It is then clear that
  (\ref{eq:facet}) holds for all points $x\in \Gamma_+(f)$.  Moreover,
  by \cite[Lemma 1]{mihaescu2008combinatorics}, the span of
  $\mathcal{U}$ is all of $\R^E$.  Hence, the affine hull of
  $\mathcal{U}$ is the hyperplane given by $\sum_{e\in E_L}x_e= 2$,
  and we conclude that (\ref{eq:facet}) defines a facet of
  $\Gamma_+(f)$.

  Since $|E_L|=|L|$, inequality~(\ref{eq:facet}) implies that the
  $\mathbf{1}$-distance of  $\Gamma_+(f)$ is at
  least $2/|L|$.  We claim that it is equal to $2/|L|$.  In fact, we
  will show that the vector $\frac{2}{|L|}\mathbf{1}$ not only lies in
  the Newton polyhedron but also in the Newton polytope $\Gamma(f)$,
  that is, the vector is a convex combination of the incidence vectors
  in $\mathcal{U}$.  To prove this, we construct a set of paths
  $\mathcal{P}$ in the tree $S$ such that (i) each element of
  $\mathcal{P}$ is a path between leaves of $S$, (ii) $\mathcal{P}$
  contains precisely $|L|$ paths, and (iii) every edge of $S$ is
  covered by exactly two paths of $\mathcal{P}$.  The construction
  implies our claim because the average of the incidence vectors of
  the paths in $\mathcal{P}$ is equal to $\frac{2}{|L|}\mathbf{1}$.

  Let $S^*$ be any trivalent tree that has the same set of leaves $L$ as
  $S$ and that can be obtained from $S^*$ by edge
  contraction. Here, a tree is trivalent if each inner node has degree
  three.  We will use induction on the number of leaves to
  show that a set of paths $\mathcal{P}$ with the desired properties
  (i)-(iii) exists.  Figure \ref{fig:tree network} shows an example.

  If $S^*$ has three $|L|=3$ leaves, then there is a single inner node
  and each path between two leaves has two edges.  We may simply take
  $\mathcal{P}$ to be the set of all the three paths that exist
  between pairs of leaves.  This provides the induction base.

  In the induction step, pick two leaves $v$ and $w$ of the tree $S^*$
  that are joined by a path with two edges $\{v,a\}$ and $\{a,w\}$.
  The node $a$ is an inner node of $S^*$.  Remove the two edges and
  the two leaves to form a subtree $S^{**}$, in which $a$ becomes a
  leaf.  Then $S^{**}$ has $|L|-1$ leaves and, by the induction
  hypothesis, there is a set of paths $\mathcal{P}^{**}$ that satisfies
  properties (i)-(iii) with respect to $S^{**}$.  In particular,
  $|\mathcal{P}^{**}|=|L|-1$.  Now, precisely two paths in
  $\mathcal{P}^{**}$ have the node $a$ as an endpoint.  Extend one of
  them by adding the edge $\{a,v\}$ and extend the other by adding
  $\{a,w\}$.  This gives two paths between leaves of $S^*$.  All other
  paths in $\mathcal{P}^{**}$ are already paths between leaves of
  $S^*$.  Add one further path, namely, $(v,a,w)$, and denote the
  resulting collection of $|L|$ paths by $\mathcal{P}^*$.  Clearly,
  the set $\mathcal{P}^*$ satisfies properties (i)-(iii) with respect
  to $S^*$.  Contracting each path in $\mathcal{P}^*$ by applying the
  edge contractions that transform $S^*$ into $S$, we obtain a system
  of paths $\mathcal{P}$ that satisfies properties (i)-(iii) with respect
  to $S$.

  \begin{figure}
    \includegraphics[scale=1]{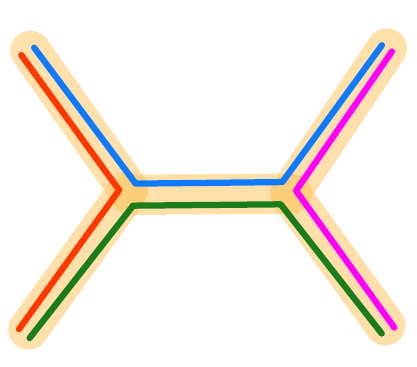}
    \caption{An example of a system of paths such that each edge of a
      trivalent tree is covered by exactly two paths.}\label{fig:tree
      network} 
  \end{figure}

  Finally, note that in the construction we just gave we can ensure
  that $\mathcal{P}$ includes a given path between two leaves in $L$.
  Hence, the vector $\frac{2}{|L|}\mathbf 1$ can be written as a
  convex combination of vertices of $\Gamma(f)$ such that a given
  vertex $x$ get positive weight.  It follows that
  $\frac{2}{|L|}\mathbf 1$ lies in the interior of the Newton polytope
  and thus the multiplicity $\mathfrak{m}$ is  $1$.
\end{proof}

The next result generalizes the previous lemma to the case of trees
with nodes of degree 2.  We remark Example~\ref{ex:xy} is a special
case of this generalization. It matches the case where the
tree $S$ has two leaves and one inner node, which is then necessarily
of degree two.

\begin{lem}
  \label{lem:single-tree-S-deg2}
  Let $S=(V,E)$ be a tree with set of leaves $L$, and let $f$ be the
  sum of squares of the monomials
  \begin{equation}\label{eq:simpleSS-deg2}
    \prod_{e\in \overline{vw}} \omega_{e}, \qquad v,w\in L,\; v\not= w.
  \end{equation}
  If $\Omega$ is a neighborhood of the origin, then 
  \[
  \RLCT_\Omega(f)=\left(\frac{|L|}{2},1+l_2\right)
  \]
  where $l_2$ is the number of (inner) nodes of $S$ that have degree
  two. 
\end{lem}
\begin{proof}
  Suppose $a$ is an inner node of degree two, and that $a$ is incident
  to the two edges $e=\{a,b\}$ and $f=\{a,c\}$.  Then any path connecting
  to leaves in $L$ either uses both $e$ and $f$ or neither $e$ nor
  $f$.  Hence, if $x$ is the incidence vector of a path between two
  leaves in $L$, then $x_e=x_f$.  It follows that the affine hull of
  Newton polytope generated by the path incidence vectors is no longer
  a hyperplane but an affine space of dimension $|E|-1-l_2$.

  Proceeding exactly as in the proof of Lemma~\ref{lem:single-tree-S},
  we see that it still holds that the $\mathbf{1}$-distance of the
  Newton polyhedron $\Gamma_+(f)$ is $2/|L|$.  Similarly, the ray
  spanned by $\mathbf{1}$ still meets $\Gamma_+(f)$ in the relative
  interior of the Newton polytope $\Gamma(f)$.  However, since the
  codimension of the Newton polytope is now $1+l_2$, we have
  $\RLCT_\Omega(f)=(|L|/2,1+l_2)$.
\end{proof}

\section*{Acknowledgments}

This work was partially supported by the European Union 7th Framework
Programme (PIOF-GA-2011-300975), the U.S.~National Science Foundation
(DMS-1305154), the U.S.~National Security Agency (H98230-14-1-0119),
and the University of Washington's Royalty Research Fund.  The United
States Government is authorized to reproduce and distribute reprints. We are thankful to the referee for constructive remarks. 

\bibliographystyle{amsalpha} 
\bibliography{algebraic_statistics}

\end{document}

%% file: lattice_node_colors.tex
\definecolor{1color25_1bics}{RGB}{255,227,227}
\definecolor{2color25_1bics}{RGB}{255,188,188}
\definecolor{3color25_1bics}{RGB}{255,250,250}
\definecolor{4color25_1bics}{RGB}{255,245,245}
\definecolor{5color25_1bics}{RGB}{255,219,219}
\definecolor{6color25_1bics}{RGB}{255,255,255}
\definecolor{7color25_1bics}{RGB}{255,252,252}
\definecolor{8color25_1bics}{RGB}{255,252,252}
\definecolor{9color25_1bics}{RGB}{255,242,242}
\definecolor{10color25_1bics}{RGB}{255,242,242}
\definecolor{11color25_1bics}{RGB}{255,250,250}
\definecolor{12color25_1bics}{RGB}{255,245,245}
\definecolor{13color25_1bics}{RGB}{255,232,232}
\definecolor{14color25_1bics}{RGB}{255,252,252}
\definecolor{15color25_1bics}{RGB}{255,255,255}
\definecolor{16color25_1bics}{RGB}{255,247,247}
\definecolor{17color25_1bics}{RGB}{255,255,255}
\definecolor{18color25_1bics}{RGB}{255,252,252}
\definecolor{19color25_1bics}{RGB}{255,255,255}
\definecolor{20color25_1bics}{RGB}{255,255,255}
\definecolor{21color25_1bics}{RGB}{255,252,252}
\definecolor{22color25_1bics}{RGB}{255,247,247}
\definecolor{23color25_1bics}{RGB}{255,250,250}
\definecolor{24color25_1bics}{RGB}{255,250,250}
\definecolor{25color25_1bics}{RGB}{255,255,255}
\definecolor{26color25_1bics}{RGB}{255,252,252}
\definecolor{27color25_1bics}{RGB}{255,252,252}
\definecolor{28color25_1bics}{RGB}{255,255,255}
\definecolor{29color25_1bics}{RGB}{255,255,255}
\definecolor{30color25_1bics}{RGB}{255,255,255}
\definecolor{31color25_1bics}{RGB}{255,255,255}
\definecolor{32color25_1bics}{RGB}{255,252,252}
\definecolor{33color25_1bics}{RGB}{255,255,255}
\definecolor{34color25_1bics}{RGB}{255,255,255}
\tikzset{
	n1_25_1bics/.style={circle, inner sep=1mm, minimum size=0.55cm, draw, thick, black, fill=1color25_1bics, text=black},
	n2_25_1bics/.style={circle, inner sep=1mm, minimum size=0.55cm, draw, thick, black, fill=2color25_1bics, text=black},
	n3_25_1bics/.style={circle, inner sep=1mm, minimum size=0.55cm, draw, thick, black, fill=3color25_1bics, text=black},
	n4_25_1bics/.style={circle, inner sep=1mm, minimum size=0.55cm, draw, thick, black, fill=4color25_1bics, text=black},
	n5_25_1bics/.style={circle, inner sep=1mm, minimum size=0.55cm, draw, thick, black, fill=5color25_1bics, text=black},
	n6_25_1bics/.style={circle, inner sep=1mm, minimum size=0.55cm, draw, thick, white, fill=6color25_1bics, text=black},
	n7_25_1bics/.style={circle, inner sep=1mm, minimum size=0.55cm, draw, thick, black, fill=7color25_1bics, text=black},
	n8_25_1bics/.style={circle, inner sep=1mm, minimum size=0.55cm, draw, thick, black, fill=8color25_1bics, text=black},
	n9_25_1bics/.style={circle, inner sep=1mm, minimum size=0.55cm, draw, thick, black, fill=9color25_1bics, text=black},
	n10_25_1bics/.style={circle, inner sep=1mm, minimum size=0.55cm, draw, thick, black, fill=10color25_1bics, text=black},
	n11_25_1bics/.style={circle, inner sep=1mm, minimum size=0.55cm, draw, thick, black, fill=11color25_1bics, text=black},
	n12_25_1bics/.style={circle, inner sep=1mm, minimum size=0.55cm, draw, thick, black, fill=12color25_1bics, text=black},
	n13_25_1bics/.style={rectangle, inner sep=1mm, minimum size=0.55cm, draw, thick, black, fill=13color25_1bics, text=black},
	n14_25_1bics/.style={circle, inner sep=1mm, minimum size=0.55cm, draw, thick, black, fill=14color25_1bics, text=black},
	n15_25_1bics/.style={circle, inner sep=1mm, minimum size=0.55cm, draw, thick, white, fill=15color25_1bics, text=black},
	n16_25_1bics/.style={circle, inner sep=1mm, minimum size=0.55cm, draw, thick, black, fill=16color25_1bics, text=black},
	n17_25_1bics/.style={circle, inner sep=1mm, minimum size=0.55cm, draw, thick, white, fill=17color25_1bics, text=black},
	n18_25_1bics/.style={circle, inner sep=1mm, minimum size=0.55cm, draw, thick, black, fill=18color25_1bics, text=black},
	n19_25_1bics/.style={circle, inner sep=1mm, minimum size=0.55cm, draw, thick, white, fill=19color25_1bics, text=black},
	n20_25_1bics/.style={circle, inner sep=1mm, minimum size=0.55cm, draw, thick, white, fill=20color25_1bics, text=black},
	n21_25_1bics/.style={circle, inner sep=1mm, minimum size=0.55cm, draw, thick, black, fill=21color25_1bics, text=black},
	n22_25_1bics/.style={circle, inner sep=1mm, minimum size=0.55cm, draw, thick, black, fill=22color25_1bics, text=black},
	n23_25_1bics/.style={circle, inner sep=1mm, minimum size=0.55cm, draw, thick, black, fill=23color25_1bics, text=black},
	n24_25_1bics/.style={circle, inner sep=1mm, minimum size=0.55cm, draw, thick, black, fill=24color25_1bics, text=black},
	n25_25_1bics/.style={circle, inner sep=1mm, minimum size=0.55cm, draw, thick, white, fill=25color25_1bics, text=black},
	n26_25_1bics/.style={circle, inner sep=1mm, minimum size=0.55cm, draw, thick, black, fill=26color25_1bics, text=black},
	n27_25_1bics/.style={circle, inner sep=1mm, minimum size=0.55cm, draw, thick, black, fill=27color25_1bics, text=black},
	n28_25_1bics/.style={circle, inner sep=1mm, minimum size=0.55cm, draw, thick, white, fill=28color25_1bics, text=black},
	n29_25_1bics/.style={circle, inner sep=1mm, minimum size=0.55cm, draw, thick, white, fill=29color25_1bics, text=black},
	n30_25_1bics/.style={circle, inner sep=1mm, minimum size=0.55cm, draw, thick, white, fill=30color25_1bics, text=black},
	n31_25_1bics/.style={circle, inner sep=1mm, minimum size=0.55cm, draw, thick, white, fill=31color25_1bics, text=black},
	n32_25_1bics/.style={circle, inner sep=1mm, minimum size=0.55cm, draw, thick, black, fill=32color25_1bics, text=black},
	n33_25_1bics/.style={circle, inner sep=1mm, minimum size=0.55cm, draw, thick, white, fill=33color25_1bics, text=black},
	n34_25_1bics/.style={circle, inner sep=1mm, minimum size=0.55cm, draw, thick, white, fill=34color25_1bics, text=black},
}

\definecolor{1color25_1bic}{RGB}{255,227,227}
\definecolor{2color25_1bic}{RGB}{255,165,165}
\definecolor{3color25_1bic}{RGB}{255,242,242}
\definecolor{4color25_1bic}{RGB}{255,239,239}
\definecolor{5color25_1bic}{RGB}{255,247,247}
\definecolor{6color25_1bic}{RGB}{255,252,252}
\definecolor{7color25_1bic}{RGB}{255,247,247}
\definecolor{8color25_1bic}{RGB}{255,255,255}
\definecolor{9color25_1bic}{RGB}{255,239,239}
\definecolor{10color25_1bic}{RGB}{255,227,227}
\definecolor{11color25_1bic}{RGB}{255,252,252}
\definecolor{12color25_1bic}{RGB}{255,252,252}
\definecolor{13color25_1bic}{RGB}{255,252,252}
\definecolor{14color25_1bic}{RGB}{255,250,250}
\definecolor{15color25_1bic}{RGB}{255,252,252}
\definecolor{16color25_1bic}{RGB}{255,250,250}
\definecolor{17color25_1bic}{RGB}{255,255,255}
\definecolor{18color25_1bic}{RGB}{255,250,250}
\definecolor{19color25_1bic}{RGB}{255,252,252}
\definecolor{20color25_1bic}{RGB}{255,255,255}
\definecolor{21color25_1bic}{RGB}{255,255,255}
\definecolor{22color25_1bic}{RGB}{255,247,247}
\definecolor{23color25_1bic}{RGB}{255,250,250}
\definecolor{24color25_1bic}{RGB}{255,250,250}
\definecolor{25color25_1bic}{RGB}{255,255,255}
\definecolor{26color25_1bic}{RGB}{255,252,252}
\definecolor{27color25_1bic}{RGB}{255,255,255}
\definecolor{28color25_1bic}{RGB}{255,255,255}
\definecolor{29color25_1bic}{RGB}{255,255,255}
\definecolor{30color25_1bic}{RGB}{255,255,255}
\definecolor{31color25_1bic}{RGB}{255,255,255}
\definecolor{32color25_1bic}{RGB}{255,255,255}
\definecolor{33color25_1bic}{RGB}{255,255,255}
\definecolor{34color25_1bic}{RGB}{255,255,255}
\tikzset{
	n1_25_1bic/.style={circle, inner sep=1mm, minimum size=0.55cm, draw, thick, black, fill=1color25_1bic, text=black},
	n2_25_1bic/.style={circle, inner sep=1mm, minimum size=0.55cm, draw, thick, black, fill=2color25_1bic, text=black},
	n3_25_1bic/.style={circle, inner sep=1mm, minimum size=0.55cm, draw, thick, black, fill=3color25_1bic, text=black},
	n4_25_1bic/.style={circle, inner sep=1mm, minimum size=0.55cm, draw, thick, black, fill=4color25_1bic, text=black},
	n5_25_1bic/.style={circle, inner sep=1mm, minimum size=0.55cm, draw, thick, black, fill=5color25_1bic, text=black},
	n6_25_1bic/.style={circle, inner sep=1mm, minimum size=0.55cm, draw, thick, black, fill=6color25_1bic, text=black},
	n7_25_1bic/.style={circle, inner sep=1mm, minimum size=0.55cm, draw, thick, black, fill=7color25_1bic, text=black},
	n8_25_1bic/.style={circle, inner sep=1mm, minimum size=0.55cm, draw, thick, white, fill=8color25_1bic, text=black},
	n9_25_1bic/.style={circle, inner sep=1mm, minimum size=0.55cm, draw, thick, black, fill=9color25_1bic, text=black},
	n10_25_1bic/.style={circle, inner sep=1mm, minimum size=0.55cm, draw, thick, black, fill=10color25_1bic, text=black},
	n11_25_1bic/.style={circle, inner sep=1mm, minimum size=0.55cm, draw, thick, black, fill=11color25_1bic, text=black},
	n12_25_1bic/.style={circle, inner sep=1mm, minimum size=0.55cm, draw, thick, black, fill=12color25_1bic, text=black},
	n13_25_1bic/.style={rectangle, inner sep=1mm, minimum size=0.55cm, draw, thick, black, fill=13color25_1bic, text=black},
	n14_25_1bic/.style={circle, inner sep=1mm, minimum size=0.55cm, draw, thick, black, fill=14color25_1bic, text=black},
	n15_25_1bic/.style={circle, inner sep=1mm, minimum size=0.55cm, draw, thick, black, fill=15color25_1bic, text=black},
	n16_25_1bic/.style={circle, inner sep=1mm, minimum size=0.55cm, draw, thick, black, fill=16color25_1bic, text=black},
	n17_25_1bic/.style={circle, inner sep=1mm, minimum size=0.55cm, draw, thick, white, fill=17color25_1bic, text=black},
	n18_25_1bic/.style={circle, inner sep=1mm, minimum size=0.55cm, draw, thick, black, fill=18color25_1bic, text=black},
	n19_25_1bic/.style={circle, inner sep=1mm, minimum size=0.55cm, draw, thick, black, fill=19color25_1bic, text=black},
	n20_25_1bic/.style={circle, inner sep=1mm, minimum size=0.55cm, draw, thick, white, fill=20color25_1bic, text=black},
	n21_25_1bic/.style={circle, inner sep=1mm, minimum size=0.55cm, draw, thick, white, fill=21color25_1bic, text=black},
	n22_25_1bic/.style={circle, inner sep=1mm, minimum size=0.55cm, draw, thick, black, fill=22color25_1bic, text=black},
	n23_25_1bic/.style={circle, inner sep=1mm, minimum size=0.55cm, draw, thick, black, fill=23color25_1bic, text=black},
	n24_25_1bic/.style={circle, inner sep=1mm, minimum size=0.55cm, draw, thick, black, fill=24color25_1bic, text=black},
	n25_25_1bic/.style={circle, inner sep=1mm, minimum size=0.55cm, draw, thick, white, fill=25color25_1bic, text=black},
	n26_25_1bic/.style={circle, inner sep=1mm, minimum size=0.55cm, draw, thick, black, fill=26color25_1bic, text=black},
	n27_25_1bic/.style={circle, inner sep=1mm, minimum size=0.55cm, draw, thick, white, fill=27color25_1bic, text=black},
	n28_25_1bic/.style={circle, inner sep=1mm, minimum size=0.55cm, draw, thick, white, fill=28color25_1bic, text=black},
	n29_25_1bic/.style={circle, inner sep=1mm, minimum size=0.55cm, draw, thick, white, fill=29color25_1bic, text=black},
	n30_25_1bic/.style={circle, inner sep=1mm, minimum size=0.55cm, draw, thick, white, fill=30color25_1bic, text=black},
	n31_25_1bic/.style={circle, inner sep=1mm, minimum size=0.55cm, draw, thick, white, fill=31color25_1bic, text=black},
	n32_25_1bic/.style={circle, inner sep=1mm, minimum size=0.55cm, draw, thick, white, fill=32color25_1bic, text=black},
	n33_25_1bic/.style={circle, inner sep=1mm, minimum size=0.55cm, draw, thick, white, fill=33color25_1bic, text=black},
	n34_25_1bic/.style={circle, inner sep=1mm, minimum size=0.55cm, draw, thick, white, fill=34color25_1bic, text=black},
}

\definecolor{1color75_1bics}{RGB}{255,255,255}
\definecolor{2color75_1bics}{RGB}{255,191,191}
\definecolor{3color75_1bics}{RGB}{255,255,255}
\definecolor{4color75_1bics}{RGB}{255,255,255}
\definecolor{5color75_1bics}{RGB}{255,224,224}
\definecolor{6color75_1bics}{RGB}{255,252,252}
\definecolor{7color75_1bics}{RGB}{255,255,255}
\definecolor{8color75_1bics}{RGB}{255,250,250}
\definecolor{9color75_1bics}{RGB}{255,252,252}
\definecolor{10color75_1bics}{RGB}{255,232,232}
\definecolor{11color75_1bics}{RGB}{255,252,252}
\definecolor{12color75_1bics}{RGB}{255,252,252}
\definecolor{13color75_1bics}{RGB}{255,158,158}
\definecolor{14color75_1bics}{RGB}{255,255,255}
\definecolor{15color75_1bics}{RGB}{255,255,255}
\definecolor{16color75_1bics}{RGB}{255,252,252}
\definecolor{17color75_1bics}{RGB}{255,255,255}
\definecolor{18color75_1bics}{RGB}{255,250,250}
\definecolor{19color75_1bics}{RGB}{255,255,255}
\definecolor{20color75_1bics}{RGB}{255,255,255}
\definecolor{21color75_1bics}{RGB}{255,250,250}
\definecolor{22color75_1bics}{RGB}{255,255,255}
\definecolor{23color75_1bics}{RGB}{255,255,255}
\definecolor{24color75_1bics}{RGB}{255,255,255}
\definecolor{25color75_1bics}{RGB}{255,255,255}
\definecolor{26color75_1bics}{RGB}{255,255,255}
\definecolor{27color75_1bics}{RGB}{255,255,255}
\definecolor{28color75_1bics}{RGB}{255,255,255}
\definecolor{29color75_1bics}{RGB}{255,255,255}
\definecolor{30color75_1bics}{RGB}{255,255,255}
\definecolor{31color75_1bics}{RGB}{255,252,252}
\definecolor{32color75_1bics}{RGB}{255,255,255}
\definecolor{33color75_1bics}{RGB}{255,255,255}
\definecolor{34color75_1bics}{RGB}{255,245,245}
\tikzset{
	n1_75_1bics/.style={circle, inner sep=1mm, minimum size=0.55cm, draw, thick, white, fill=1color75_1bics, text=black},
	n2_75_1bics/.style={circle, inner sep=1mm, minimum size=0.55cm, draw, thick, black, fill=2color75_1bics, text=black},
	n3_75_1bics/.style={circle, inner sep=1mm, minimum size=0.55cm, draw, thick, white, fill=3color75_1bics, text=black},
	n4_75_1bics/.style={circle, inner sep=1mm, minimum size=0.55cm, draw, thick, white, fill=4color75_1bics, text=black},
	n5_75_1bics/.style={circle, inner sep=1mm, minimum size=0.55cm, draw, thick, black, fill=5color75_1bics, text=black},
	n6_75_1bics/.style={circle, inner sep=1mm, minimum size=0.55cm, draw, thick, black, fill=6color75_1bics, text=black},
	n7_75_1bics/.style={circle, inner sep=1mm, minimum size=0.55cm, draw, thick, white, fill=7color75_1bics, text=black},
	n8_75_1bics/.style={circle, inner sep=1mm, minimum size=0.55cm, draw, thick, black, fill=8color75_1bics, text=black},
	n9_75_1bics/.style={circle, inner sep=1mm, minimum size=0.55cm, draw, thick, black, fill=9color75_1bics, text=black},
	n10_75_1bics/.style={circle, inner sep=1mm, minimum size=0.55cm, draw, thick, black, fill=10color75_1bics, text=black},
	n11_75_1bics/.style={circle, inner sep=1mm, minimum size=0.55cm, draw, thick, black, fill=11color75_1bics, text=black},
	n12_75_1bics/.style={circle, inner sep=1mm, minimum size=0.55cm, draw, thick, black, fill=12color75_1bics, text=black},
	n13_75_1bics/.style={rectangle, inner sep=1mm, minimum size=0.55cm, draw, thick, black, fill=13color75_1bics, text=black},
	n14_75_1bics/.style={circle, inner sep=1mm, minimum size=0.55cm, draw, thick, white, fill=14color75_1bics, text=black},
	n15_75_1bics/.style={circle, inner sep=1mm, minimum size=0.55cm, draw, thick, white, fill=15color75_1bics, text=black},
	n16_75_1bics/.style={circle, inner sep=1mm, minimum size=0.55cm, draw, thick, black, fill=16color75_1bics, text=black},
	n17_75_1bics/.style={circle, inner sep=1mm, minimum size=0.55cm, draw, thick, white, fill=17color75_1bics, text=black},
	n18_75_1bics/.style={circle, inner sep=1mm, minimum size=0.55cm, draw, thick, black, fill=18color75_1bics, text=black},
	n19_75_1bics/.style={circle, inner sep=1mm, minimum size=0.55cm, draw, thick, white, fill=19color75_1bics, text=black},
	n20_75_1bics/.style={circle, inner sep=1mm, minimum size=0.55cm, draw, thick, white, fill=20color75_1bics, text=black},
	n21_75_1bics/.style={circle, inner sep=1mm, minimum size=0.55cm, draw, thick, black, fill=21color75_1bics, text=black},
	n22_75_1bics/.style={circle, inner sep=1mm, minimum size=0.55cm, draw, thick, white, fill=22color75_1bics, text=black},
	n23_75_1bics/.style={circle, inner sep=1mm, minimum size=0.55cm, draw, thick, white, fill=23color75_1bics, text=black},
	n24_75_1bics/.style={circle, inner sep=1mm, minimum size=0.55cm, draw, thick, white, fill=24color75_1bics, text=black},
	n25_75_1bics/.style={circle, inner sep=1mm, minimum size=0.55cm, draw, thick, white, fill=25color75_1bics, text=black},
	n26_75_1bics/.style={circle, inner sep=1mm, minimum size=0.55cm, draw, thick, white, fill=26color75_1bics, text=black},
	n27_75_1bics/.style={circle, inner sep=1mm, minimum size=0.55cm, draw, thick, white, fill=27color75_1bics, text=black},
	n28_75_1bics/.style={circle, inner sep=1mm, minimum size=0.55cm, draw, thick, white, fill=28color75_1bics, text=black},
	n29_75_1bics/.style={circle, inner sep=1mm, minimum size=0.55cm, draw, thick, white, fill=29color75_1bics, text=black},
	n30_75_1bics/.style={circle, inner sep=1mm, minimum size=0.55cm, draw, thick, white, fill=30color75_1bics, text=black},
	n31_75_1bics/.style={circle, inner sep=1mm, minimum size=0.55cm, draw, thick, black, fill=31color75_1bics, text=black},
	n32_75_1bics/.style={circle, inner sep=1mm, minimum size=0.55cm, draw, thick, white, fill=32color75_1bics, text=black},
	n33_75_1bics/.style={circle, inner sep=1mm, minimum size=0.55cm, draw, thick, white, fill=33color75_1bics, text=black},
	n34_75_1bics/.style={circle, inner sep=1mm, minimum size=0.55cm, draw, thick, black, fill=34color75_1bics, text=black},
}

\definecolor{1color75_1bic}{RGB}{255,252,252}
\definecolor{2color75_1bic}{RGB}{255,160,160}
\definecolor{3color75_1bic}{RGB}{255,250,250}
\definecolor{4color75_1bic}{RGB}{255,250,250}
\definecolor{5color75_1bic}{RGB}{255,227,227}
\definecolor{6color75_1bic}{RGB}{255,250,250}
\definecolor{7color75_1bic}{RGB}{255,255,255}
\definecolor{8color75_1bic}{RGB}{255,252,252}
\definecolor{9color75_1bic}{RGB}{255,245,245}
\definecolor{10color75_1bic}{RGB}{255,186,186}
\definecolor{11color75_1bic}{RGB}{255,255,255}
\definecolor{12color75_1bic}{RGB}{255,252,252}
\definecolor{13color75_1bic}{RGB}{255,242,242}
\definecolor{14color75_1bic}{RGB}{255,255,255}
\definecolor{15color75_1bic}{RGB}{255,252,252}
\definecolor{16color75_1bic}{RGB}{255,255,255}
\definecolor{17color75_1bic}{RGB}{255,255,255}
\definecolor{18color75_1bic}{RGB}{255,250,250}
\definecolor{19color75_1bic}{RGB}{255,255,255}
\definecolor{20color75_1bic}{RGB}{255,255,255}
\definecolor{21color75_1bic}{RGB}{255,255,255}
\definecolor{22color75_1bic}{RGB}{255,255,255}
\definecolor{23color75_1bic}{RGB}{255,250,250}
\definecolor{24color75_1bic}{RGB}{255,255,255}
\definecolor{25color75_1bic}{RGB}{255,252,252}
\definecolor{26color75_1bic}{RGB}{255,255,255}
\definecolor{27color75_1bic}{RGB}{255,255,255}
\definecolor{28color75_1bic}{RGB}{255,252,252}
\definecolor{29color75_1bic}{RGB}{255,255,255}
\definecolor{30color75_1bic}{RGB}{255,255,255}
\definecolor{31color75_1bic}{RGB}{255,255,255}
\definecolor{32color75_1bic}{RGB}{255,255,255}
\definecolor{33color75_1bic}{RGB}{255,255,255}
\definecolor{34color75_1bic}{RGB}{255,255,255}
\tikzset{
	n1_75_1bic/.style={circle, inner sep=1mm, minimum size=0.55cm, draw, thick, black, fill=1color75_1bic, text=black},
	n2_75_1bic/.style={circle, inner sep=1mm, minimum size=0.55cm, draw, thick, black, fill=2color75_1bic, text=black},
	n3_75_1bic/.style={circle, inner sep=1mm, minimum size=0.55cm, draw, thick, black, fill=3color75_1bic, text=black},
	n4_75_1bic/.style={circle, inner sep=1mm, minimum size=0.55cm, draw, thick, black, fill=4color75_1bic, text=black},
	n5_75_1bic/.style={circle, inner sep=1mm, minimum size=0.55cm, draw, thick, black, fill=5color75_1bic, text=black},
	n6_75_1bic/.style={circle, inner sep=1mm, minimum size=0.55cm, draw, thick, black, fill=6color75_1bic, text=black},
	n7_75_1bic/.style={circle, inner sep=1mm, minimum size=0.55cm, draw, thick, white, fill=7color75_1bic, text=black},
	n8_75_1bic/.style={circle, inner sep=1mm, minimum size=0.55cm, draw, thick, black, fill=8color75_1bic, text=black},
	n9_75_1bic/.style={circle, inner sep=1mm, minimum size=0.55cm, draw, thick, black, fill=9color75_1bic, text=black},
	n10_75_1bic/.style={circle, inner sep=1mm, minimum size=0.55cm, draw, thick, black, fill=10color75_1bic, text=black},
	n11_75_1bic/.style={circle, inner sep=1mm, minimum size=0.55cm, draw, thick, white, fill=11color75_1bic, text=black},
	n12_75_1bic/.style={circle, inner sep=1mm, minimum size=0.55cm, draw, thick, black, fill=12color75_1bic, text=black},
	n13_75_1bic/.style={rectangle, inner sep=1mm, minimum size=0.55cm, draw, thick, black, fill=13color75_1bic, text=black},
	n14_75_1bic/.style={circle, inner sep=1mm, minimum size=0.55cm, draw, thick, white, fill=14color75_1bic, text=black},
	n15_75_1bic/.style={circle, inner sep=1mm, minimum size=0.55cm, draw, thick, black, fill=15color75_1bic, text=black},
	n16_75_1bic/.style={circle, inner sep=1mm, minimum size=0.55cm, draw, thick, white, fill=16color75_1bic, text=black},
	n17_75_1bic/.style={circle, inner sep=1mm, minimum size=0.55cm, draw, thick, white, fill=17color75_1bic, text=black},
	n18_75_1bic/.style={circle, inner sep=1mm, minimum size=0.55cm, draw, thick, black, fill=18color75_1bic, text=black},
	n19_75_1bic/.style={circle, inner sep=1mm, minimum size=0.55cm, draw, thick, white, fill=19color75_1bic, text=black},
	n20_75_1bic/.style={circle, inner sep=1mm, minimum size=0.55cm, draw, thick, white, fill=20color75_1bic, text=black},
	n21_75_1bic/.style={circle, inner sep=1mm, minimum size=0.55cm, draw, thick, white, fill=21color75_1bic, text=black},
	n22_75_1bic/.style={circle, inner sep=1mm, minimum size=0.55cm, draw, thick, white, fill=22color75_1bic, text=black},
	n23_75_1bic/.style={circle, inner sep=1mm, minimum size=0.55cm, draw, thick, black, fill=23color75_1bic, text=black},
	n24_75_1bic/.style={circle, inner sep=1mm, minimum size=0.55cm, draw, thick, white, fill=24color75_1bic, text=black},
	n25_75_1bic/.style={circle, inner sep=1mm, minimum size=0.55cm, draw, thick, black, fill=25color75_1bic, text=black},
	n26_75_1bic/.style={circle, inner sep=1mm, minimum size=0.55cm, draw, thick, white, fill=26color75_1bic, text=black},
	n27_75_1bic/.style={circle, inner sep=1mm, minimum size=0.55cm, draw, thick, white, fill=27color75_1bic, text=black},
	n28_75_1bic/.style={circle, inner sep=1mm, minimum size=0.55cm, draw, thick, black, fill=28color75_1bic, text=black},
	n29_75_1bic/.style={circle, inner sep=1mm, minimum size=0.55cm, draw, thick, white, fill=29color75_1bic, text=black},
	n30_75_1bic/.style={circle, inner sep=1mm, minimum size=0.55cm, draw, thick, white, fill=30color75_1bic, text=black},
	n31_75_1bic/.style={circle, inner sep=1mm, minimum size=0.55cm, draw, thick, white, fill=31color75_1bic, text=black},
	n32_75_1bic/.style={circle, inner sep=1mm, minimum size=0.55cm, draw, thick, white, fill=32color75_1bic, text=black},
	n33_75_1bic/.style={circle, inner sep=1mm, minimum size=0.55cm, draw, thick, white, fill=33color75_1bic, text=black},
	n34_75_1bic/.style={circle, inner sep=1mm, minimum size=0.55cm, draw, thick, white, fill=34color75_1bic, text=black},
}

\definecolor{1color125_1bics}{RGB}{255,255,255}
\definecolor{2color125_1bics}{RGB}{255,234,234}
\definecolor{3color125_1bics}{RGB}{255,250,250}
\definecolor{4color125_1bics}{RGB}{255,255,255}
\definecolor{5color125_1bics}{RGB}{255,214,214}
\definecolor{6color125_1bics}{RGB}{255,255,255}
\definecolor{7color125_1bics}{RGB}{255,252,252}
\definecolor{8color125_1bics}{RGB}{255,252,252}
\definecolor{9color125_1bics}{RGB}{255,255,255}
\definecolor{10color125_1bics}{RGB}{255,239,239}
\definecolor{11color125_1bics}{RGB}{255,255,255}
\definecolor{12color125_1bics}{RGB}{255,255,255}
\definecolor{13color125_1bics}{RGB}{255,107,107}
\definecolor{14color125_1bics}{RGB}{255,255,255}
\definecolor{15color125_1bics}{RGB}{255,255,255}
\definecolor{16color125_1bics}{RGB}{255,255,255}
\definecolor{17color125_1bics}{RGB}{255,255,255}
\definecolor{18color125_1bics}{RGB}{255,255,255}
\definecolor{19color125_1bics}{RGB}{255,255,255}
\definecolor{20color125_1bics}{RGB}{255,255,255}
\definecolor{21color125_1bics}{RGB}{255,255,255}
\definecolor{22color125_1bics}{RGB}{255,255,255}
\definecolor{23color125_1bics}{RGB}{255,255,255}
\definecolor{24color125_1bics}{RGB}{255,255,255}
\definecolor{25color125_1bics}{RGB}{255,255,255}
\definecolor{26color125_1bics}{RGB}{255,255,255}
\definecolor{27color125_1bics}{RGB}{255,255,255}
\definecolor{28color125_1bics}{RGB}{255,255,255}
\definecolor{29color125_1bics}{RGB}{255,255,255}
\definecolor{30color125_1bics}{RGB}{255,255,255}
\definecolor{31color125_1bics}{RGB}{255,255,255}
\definecolor{32color125_1bics}{RGB}{255,255,255}
\definecolor{33color125_1bics}{RGB}{255,255,255}
\definecolor{34color125_1bics}{RGB}{255,234,234}
\tikzset{
	n1_125_1bics/.style={circle, inner sep=1mm, minimum size=0.55cm, draw, thick, white, fill=1color125_1bics, text=black},
	n2_125_1bics/.style={circle, inner sep=1mm, minimum size=0.55cm, draw, thick, black, fill=2color125_1bics, text=black},
	n3_125_1bics/.style={circle, inner sep=1mm, minimum size=0.55cm, draw, thick, black, fill=3color125_1bics, text=black},
	n4_125_1bics/.style={circle, inner sep=1mm, minimum size=0.55cm, draw, thick, white, fill=4color125_1bics, text=black},
	n5_125_1bics/.style={circle, inner sep=1mm, minimum size=0.55cm, draw, thick, black, fill=5color125_1bics, text=black},
	n6_125_1bics/.style={circle, inner sep=1mm, minimum size=0.55cm, draw, thick, white, fill=6color125_1bics, text=black},
	n7_125_1bics/.style={circle, inner sep=1mm, minimum size=0.55cm, draw, thick, black, fill=7color125_1bics, text=black},
	n8_125_1bics/.style={circle, inner sep=1mm, minimum size=0.55cm, draw, thick, black, fill=8color125_1bics, text=black},
	n9_125_1bics/.style={circle, inner sep=1mm, minimum size=0.55cm, draw, thick, white, fill=9color125_1bics, text=black},
	n10_125_1bics/.style={circle, inner sep=1mm, minimum size=0.55cm, draw, thick, black, fill=10color125_1bics, text=black},
	n11_125_1bics/.style={circle, inner sep=1mm, minimum size=0.55cm, draw, thick, white, fill=11color125_1bics, text=black},
	n12_125_1bics/.style={circle, inner sep=1mm, minimum size=0.55cm, draw, thick, white, fill=12color125_1bics, text=black},
	n13_125_1bics/.style={rectangle, inner sep=1mm, minimum size=0.55cm, draw, thick, black, fill=13color125_1bics, text=black},
	n14_125_1bics/.style={circle, inner sep=1mm, minimum size=0.55cm, draw, thick, white, fill=14color125_1bics, text=black},
	n15_125_1bics/.style={circle, inner sep=1mm, minimum size=0.55cm, draw, thick, white, fill=15color125_1bics, text=black},
	n16_125_1bics/.style={circle, inner sep=1mm, minimum size=0.55cm, draw, thick, white, fill=16color125_1bics, text=black},
	n17_125_1bics/.style={circle, inner sep=1mm, minimum size=0.55cm, draw, thick, white, fill=17color125_1bics, text=black},
	n18_125_1bics/.style={circle, inner sep=1mm, minimum size=0.55cm, draw, thick, white, fill=18color125_1bics, text=black},
	n19_125_1bics/.style={circle, inner sep=1mm, minimum size=0.55cm, draw, thick, white, fill=19color125_1bics, text=black},
	n20_125_1bics/.style={circle, inner sep=1mm, minimum size=0.55cm, draw, thick, white, fill=20color125_1bics, text=black},
	n21_125_1bics/.style={circle, inner sep=1mm, minimum size=0.55cm, draw, thick, white, fill=21color125_1bics, text=black},
	n22_125_1bics/.style={circle, inner sep=1mm, minimum size=0.55cm, draw, thick, white, fill=22color125_1bics, text=black},
	n23_125_1bics/.style={circle, inner sep=1mm, minimum size=0.55cm, draw, thick, white, fill=23color125_1bics, text=black},
	n24_125_1bics/.style={circle, inner sep=1mm, minimum size=0.55cm, draw, thick, white, fill=24color125_1bics, text=black},
	n25_125_1bics/.style={circle, inner sep=1mm, minimum size=0.55cm, draw, thick, white, fill=25color125_1bics, text=black},
	n26_125_1bics/.style={circle, inner sep=1mm, minimum size=0.55cm, draw, thick, white, fill=26color125_1bics, text=black},
	n27_125_1bics/.style={circle, inner sep=1mm, minimum size=0.55cm, draw, thick, white, fill=27color125_1bics, text=black},
	n28_125_1bics/.style={circle, inner sep=1mm, minimum size=0.55cm, draw, thick, white, fill=28color125_1bics, text=black},
	n29_125_1bics/.style={circle, inner sep=1mm, minimum size=0.55cm, draw, thick, white, fill=29color125_1bics, text=black},
	n30_125_1bics/.style={circle, inner sep=1mm, minimum size=0.55cm, draw, thick, white, fill=30color125_1bics, text=black},
	n31_125_1bics/.style={circle, inner sep=1mm, minimum size=0.55cm, draw, thick, white, fill=31color125_1bics, text=black},
	n32_125_1bics/.style={circle, inner sep=1mm, minimum size=0.55cm, draw, thick, white, fill=32color125_1bics, text=black},
	n33_125_1bics/.style={circle, inner sep=1mm, minimum size=0.55cm, draw, thick, white, fill=33color125_1bics, text=black},
	n34_125_1bics/.style={circle, inner sep=1mm, minimum size=0.55cm, draw, thick, black, fill=34color125_1bics, text=black},
}

\definecolor{1color125_1bic}{RGB}{255,255,255}
\definecolor{2color125_1bic}{RGB}{255,204,204}
\definecolor{3color125_1bic}{RGB}{255,245,245}
\definecolor{4color125_1bic}{RGB}{255,255,255}
\definecolor{5color125_1bic}{RGB}{255,183,183}
\definecolor{6color125_1bic}{RGB}{255,255,255}
\definecolor{7color125_1bic}{RGB}{255,252,252}
\definecolor{8color125_1bic}{RGB}{255,252,252}
\definecolor{9color125_1bic}{RGB}{255,255,255}
\definecolor{10color125_1bic}{RGB}{255,171,171}
\definecolor{11color125_1bic}{RGB}{255,255,255}
\definecolor{12color125_1bic}{RGB}{255,255,255}
\definecolor{13color125_1bic}{RGB}{255,229,229}
\definecolor{14color125_1bic}{RGB}{255,255,255}
\definecolor{15color125_1bic}{RGB}{255,255,255}
\definecolor{16color125_1bic}{RGB}{255,255,255}
\definecolor{17color125_1bic}{RGB}{255,255,255}
\definecolor{18color125_1bic}{RGB}{255,250,250}
\definecolor{19color125_1bic}{RGB}{255,255,255}
\definecolor{20color125_1bic}{RGB}{255,255,255}
\definecolor{21color125_1bic}{RGB}{255,255,255}
\definecolor{22color125_1bic}{RGB}{255,255,255}
\definecolor{23color125_1bic}{RGB}{255,255,255}
\definecolor{24color125_1bic}{RGB}{255,255,255}
\definecolor{25color125_1bic}{RGB}{255,255,255}
\definecolor{26color125_1bic}{RGB}{255,252,252}
\definecolor{27color125_1bic}{RGB}{255,255,255}
\definecolor{28color125_1bic}{RGB}{255,255,255}
\definecolor{29color125_1bic}{RGB}{255,255,255}
\definecolor{30color125_1bic}{RGB}{255,255,255}
\definecolor{31color125_1bic}{RGB}{255,255,255}
\definecolor{32color125_1bic}{RGB}{255,255,255}
\definecolor{33color125_1bic}{RGB}{255,255,255}
\definecolor{34color125_1bic}{RGB}{255,255,255}
\tikzset{
	n1_125_1bic/.style={circle, inner sep=1mm, minimum size=0.55cm, draw, thick, white, fill=1color125_1bic, text=black},
	n2_125_1bic/.style={circle, inner sep=1mm, minimum size=0.55cm, draw, thick, black, fill=2color125_1bic, text=black},
	n3_125_1bic/.style={circle, inner sep=1mm, minimum size=0.55cm, draw, thick, black, fill=3color125_1bic, text=black},
	n4_125_1bic/.style={circle, inner sep=1mm, minimum size=0.55cm, draw, thick, white, fill=4color125_1bic, text=black},
	n5_125_1bic/.style={circle, inner sep=1mm, minimum size=0.55cm, draw, thick, black, fill=5color125_1bic, text=black},
	n6_125_1bic/.style={circle, inner sep=1mm, minimum size=0.55cm, draw, thick, white, fill=6color125_1bic, text=black},
	n7_125_1bic/.style={circle, inner sep=1mm, minimum size=0.55cm, draw, thick, black, fill=7color125_1bic, text=black},
	n8_125_1bic/.style={circle, inner sep=1mm, minimum size=0.55cm, draw, thick, black, fill=8color125_1bic, text=black},
	n9_125_1bic/.style={circle, inner sep=1mm, minimum size=0.55cm, draw, thick, white, fill=9color125_1bic, text=black},
	n10_125_1bic/.style={circle, inner sep=1mm, minimum size=0.55cm, draw, thick, black, fill=10color125_1bic, text=black},
	n11_125_1bic/.style={circle, inner sep=1mm, minimum size=0.55cm, draw, thick, white, fill=11color125_1bic, text=black},
	n12_125_1bic/.style={circle, inner sep=1mm, minimum size=0.55cm, draw, thick, white, fill=12color125_1bic, text=black},
	n13_125_1bic/.style={rectangle, inner sep=1mm, minimum size=0.55cm, draw, thick, black, fill=13color125_1bic, text=black},
	n14_125_1bic/.style={circle, inner sep=1mm, minimum size=0.55cm, draw, thick, white, fill=14color125_1bic, text=black},
	n15_125_1bic/.style={circle, inner sep=1mm, minimum size=0.55cm, draw, thick, white, fill=15color125_1bic, text=black},
	n16_125_1bic/.style={circle, inner sep=1mm, minimum size=0.55cm, draw, thick, white, fill=16color125_1bic, text=black},
	n17_125_1bic/.style={circle, inner sep=1mm, minimum size=0.55cm, draw, thick, white, fill=17color125_1bic, text=black},
	n18_125_1bic/.style={circle, inner sep=1mm, minimum size=0.55cm, draw, thick, black, fill=18color125_1bic, text=black},
	n19_125_1bic/.style={circle, inner sep=1mm, minimum size=0.55cm, draw, thick, white, fill=19color125_1bic, text=black},
	n20_125_1bic/.style={circle, inner sep=1mm, minimum size=0.55cm, draw, thick, white, fill=20color125_1bic, text=black},
	n21_125_1bic/.style={circle, inner sep=1mm, minimum size=0.55cm, draw, thick, white, fill=21color125_1bic, text=black},
	n22_125_1bic/.style={circle, inner sep=1mm, minimum size=0.55cm, draw, thick, white, fill=22color125_1bic, text=black},
	n23_125_1bic/.style={circle, inner sep=1mm, minimum size=0.55cm, draw, thick, white, fill=23color125_1bic, text=black},
	n24_125_1bic/.style={circle, inner sep=1mm, minimum size=0.55cm, draw, thick, white, fill=24color125_1bic, text=black},
	n25_125_1bic/.style={circle, inner sep=1mm, minimum size=0.55cm, draw, thick, white, fill=25color125_1bic, text=black},
	n26_125_1bic/.style={circle, inner sep=1mm, minimum size=0.55cm, draw, thick, black, fill=26color125_1bic, text=black},
	n27_125_1bic/.style={circle, inner sep=1mm, minimum size=0.55cm, draw, thick, white, fill=27color125_1bic, text=black},
	n28_125_1bic/.style={circle, inner sep=1mm, minimum size=0.55cm, draw, thick, white, fill=28color125_1bic, text=black},
	n29_125_1bic/.style={circle, inner sep=1mm, minimum size=0.55cm, draw, thick, white, fill=29color125_1bic, text=black},
	n30_125_1bic/.style={circle, inner sep=1mm, minimum size=0.55cm, draw, thick, white, fill=30color125_1bic, text=black},
	n31_125_1bic/.style={circle, inner sep=1mm, minimum size=0.55cm, draw, thick, white, fill=31color125_1bic, text=black},
	n32_125_1bic/.style={circle, inner sep=1mm, minimum size=0.55cm, draw, thick, white, fill=32color125_1bic, text=black},
	n33_125_1bic/.style={circle, inner sep=1mm, minimum size=0.55cm, draw, thick, white, fill=33color125_1bic, text=black},
	n34_125_1bic/.style={circle, inner sep=1mm, minimum size=0.55cm, draw, thick, white, fill=34color125_1bic, text=black},
}

%% file: bics_lattice_25.tex
\begin{tikzpicture}[scale=9]
	\node [n1_25_1bics] (v1) at (0.533333, 0.832000)	{1};
	\node [n2_25_1bics] (v2) at (0.096970, 0.624000)	{2};
	\node [n3_25_1bics] (v3) at (0.193939, 0.624000)	{3};
	\node [n4_25_1bics] (v4) at (0.290909, 0.624000)	{4};
	\node [n5_25_1bics] (v5) at (0.066667, 0.416000)	{5};
	\node [n6_25_1bics] (v6) at (0.387879, 0.624000)	{6};
	\node [n7_25_1bics] (v7) at (0.484848, 0.624000)	{7};
	\node [n8_25_1bics] (v8) at (0.133333, 0.416000)	{8};
	\node [n9_25_1bics] (v9) at (0.581818, 0.624000)	{9};
	\node [n10_25_1bics] (v10) at (0.200000, 0.416000)	{10};
	\node [n11_25_1bics] (v11) at (0.266667, 0.416000)	{11};
	\node [n12_25_1bics] (v12) at (0.333333, 0.416000)	{12};
	\node [n13_25_1bics] (v13) at (0.133333, 0.208000)	{$\begin{array}{@{}c@{}} 13 \\ (9\%) \end{array}$};
	\node [n14_25_1bics] (v14) at (0.678788, 0.624000)	{14};
	\node [n15_25_1bics] (v15) at (0.775758, 0.624000)	{15};
	\node [n16_25_1bics] (v16) at (0.400000, 0.416000)	{16};
	\node [n17_25_1bics] (v17) at (0.872727, 0.624000)	{17};
	\node [n18_25_1bics] (v18) at (0.466667, 0.416000)	{18};
	\node [n19_25_1bics] (v19) at (0.533333, 0.416000)	{19};
	\node [n20_25_1bics] (v20) at (0.600000, 0.416000)	{20};
	\node [n21_25_1bics] (v21) at (0.266667, 0.208000)	{21};
	\node [n22_25_1bics] (v22) at (0.969697, 0.624000)	{22};
	\node [n23_25_1bics] (v23) at (0.666667, 0.416000)	{23};
	\node [n24_25_1bics] (v24) at (0.733333, 0.416000)	{24};
	\node [n25_25_1bics] (v25) at (0.800000, 0.416000)	{25};
	\node [n26_25_1bics] (v26) at (0.400000, 0.208000)	{26};
	\node [n27_25_1bics] (v27) at (0.866667, 0.416000)	{27};
	\node [n28_25_1bics] (v28) at (0.933333, 0.416000)	{28};
	\node [n29_25_1bics] (v29) at (0.533333, 0.208000)	{29};
	\node [n30_25_1bics] (v30) at (1.000000, 0.416000)	{30};
	\node [n31_25_1bics] (v31) at (0.666667, 0.208000)	{31};
	\node [n32_25_1bics] (v32) at (0.800000, 0.208000)	{32};
	\node [n33_25_1bics] (v33) at (0.933333, 0.208000)	{33};
	\node [n34_25_1bics] (v34) at (0.533333, 0.000000)	{34};

	\draw [->] (v1) -- (v2);
	\draw [->] (v1) -- (v3);
	\draw [->] (v1) -- (v4);
	\draw [->] (v1) -- (v6);
	\draw [->] (v1) -- (v7);
	\draw [->] (v1) -- (v9);
	\draw [->] (v1) -- (v14);
	\draw [->] (v1) -- (v15);
	\draw [->] (v1) -- (v17);
	\draw [->] (v1) -- (v22);
	\draw [->] (v2) -- (v5);
	\draw [->] (v2) -- (v8);
	\draw [->] (v2) -- (v10);
	\draw [->] (v2) -- (v16);
	\draw [->] (v2) -- (v18);
	\draw [->] (v2) -- (v23);
	\draw [->] (v3) -- (v5);
	\draw [->] (v3) -- (v11);
	\draw [->] (v3) -- (v19);
	\draw [->] (v3) -- (v24);
	\draw [->] (v4) -- (v5);
	\draw [->] (v4) -- (v12);
	\draw [->] (v4) -- (v20);
	\draw [->] (v4) -- (v25);
	\draw [->] (v5) -- (v13);
	\draw [->] (v5) -- (v21);
	\draw [->] (v5) -- (v26);
	\draw [->] (v6) -- (v8);
	\draw [->] (v6) -- (v11);
	\draw [->] (v6) -- (v27);
	\draw [->] (v7) -- (v8);
	\draw [->] (v7) -- (v12);
	\draw [->] (v7) -- (v28);
	\draw [->] (v8) -- (v13);
	\draw [->] (v8) -- (v29);
	\draw [->] (v9) -- (v10);
	\draw [->] (v9) -- (v11);
	\draw [->] (v9) -- (v12);
	\draw [->] (v9) -- (v30);
	\draw [->] (v10) -- (v13);
	\draw [->] (v10) -- (v31);
	\draw [->] (v11) -- (v13);
	\draw [->] (v11) -- (v32);
	\draw [->] (v12) -- (v13);
	\draw [->] (v12) -- (v33);
	\draw [->] (v13) -- (v34);
	\draw [->] (v14) -- (v16);
	\draw [->] (v14) -- (v19);
	\draw [->] (v14) -- (v27);
	\draw [->] (v15) -- (v16);
	\draw [->] (v15) -- (v20);
	\draw [->] (v15) -- (v28);
	\draw [->] (v16) -- (v21);
	\draw [->] (v16) -- (v29);
	\draw [->] (v17) -- (v18);
	\draw [->] (v17) -- (v19);
	\draw [->] (v17) -- (v20);
	\draw [->] (v17) -- (v30);
	\draw [->] (v18) -- (v21);
	\draw [->] (v18) -- (v31);
	\draw [->] (v19) -- (v21);
	\draw [->] (v19) -- (v32);
	\draw [->] (v20) -- (v21);
	\draw [->] (v20) -- (v33);
	\draw [->] (v21) -- (v34);
	\draw [->] (v22) -- (v23);
	\draw [->] (v22) -- (v24);
	\draw [->] (v22) -- (v25);
	\draw [->] (v22) -- (v27);
	\draw [->] (v22) -- (v28);
	\draw [->] (v22) -- (v30);
	\draw [->] (v23) -- (v26);
	\draw [->] (v23) -- (v29);
	\draw [->] (v23) -- (v31);
	\draw [->] (v24) -- (v26);
	\draw [->] (v24) -- (v32);
	\draw [->] (v25) -- (v26);
	\draw [->] (v25) -- (v33);
	\draw [->] (v26) -- (v34);
	\draw [->] (v27) -- (v29);
	\draw [->] (v27) -- (v32);
	\draw [->] (v28) -- (v29);
	\draw [->] (v28) -- (v33);
	\draw [->] (v29) -- (v34);
	\draw [->] (v30) -- (v31);
	\draw [->] (v30) -- (v32);
	\draw [->] (v30) -- (v33);
	\draw [->] (v31) -- (v34);
	\draw [->] (v32) -- (v34);
	\draw [->] (v33) -- (v34);
\end{tikzpicture}

%% file: bic_lattice_25.tex
\begin{tikzpicture}[scale=9]
	\node [n1_25_1bic] (v1) at (0.533333, 0.832000)	{1};
	\node [n2_25_1bic] (v2) at (0.096970, 0.624000)	{2};
	\node [n3_25_1bic] (v3) at (0.193939, 0.624000)	{3};
	\node [n4_25_1bic] (v4) at (0.290909, 0.624000)	{4};
	\node [n5_25_1bic] (v5) at (0.066667, 0.416000)	{5};
	\node [n6_25_1bic] (v6) at (0.387879, 0.624000)	{6};
	\node [n7_25_1bic] (v7) at (0.484848, 0.624000)	{7};
	\node [n8_25_1bic] (v8) at (0.133333, 0.416000)	{8};
	\node [n9_25_1bic] (v9) at (0.581818, 0.624000)	{9};
	\node [n10_25_1bic] (v10) at (0.200000, 0.416000)	{10};
	\node [n11_25_1bic] (v11) at (0.266667, 0.416000)	{11};
	\node [n12_25_1bic] (v12) at (0.333333, 0.416000)	{12};
	\node [n13_25_1bic] (v13) at (0.133333, 0.208000)	{$\begin{array}{@{}c@{}} 13 \\ (1\%) \end{array}$};
	\node [n14_25_1bic] (v14) at (0.678788, 0.624000)	{14};
	\node [n15_25_1bic] (v15) at (0.775758, 0.624000)	{15};
	\node [n16_25_1bic] (v16) at (0.400000, 0.416000)	{16};
	\node [n17_25_1bic] (v17) at (0.872727, 0.624000)	{17};
	\node [n18_25_1bic] (v18) at (0.466667, 0.416000)	{18};
	\node [n19_25_1bic] (v19) at (0.533333, 0.416000)	{19};
	\node [n20_25_1bic] (v20) at (0.600000, 0.416000)	{20};
	\node [n21_25_1bic] (v21) at (0.266667, 0.208000)	{21};
	\node [n22_25_1bic] (v22) at (0.969697, 0.624000)	{22};
	\node [n23_25_1bic] (v23) at (0.666667, 0.416000)	{23};
	\node [n24_25_1bic] (v24) at (0.733333, 0.416000)	{24};
	\node [n25_25_1bic] (v25) at (0.800000, 0.416000)	{25};
	\node [n26_25_1bic] (v26) at (0.400000, 0.208000)	{26};
	\node [n27_25_1bic] (v27) at (0.866667, 0.416000)	{27};
	\node [n28_25_1bic] (v28) at (0.933333, 0.416000)	{28};
	\node [n29_25_1bic] (v29) at (0.533333, 0.208000)	{29};
	\node [n30_25_1bic] (v30) at (1.000000, 0.416000)	{30};
	\node [n31_25_1bic] (v31) at (0.666667, 0.208000)	{31};
	\node [n32_25_1bic] (v32) at (0.800000, 0.208000)	{32};
	\node [n33_25_1bic] (v33) at (0.933333, 0.208000)	{33};
	\node [n34_25_1bic] (v34) at (0.533333, 0.000000)	{34};

	\draw [->] (v1) -- (v2);
	\draw [->] (v1) -- (v3);
	\draw [->] (v1) -- (v4);
	\draw [->] (v1) -- (v6);
	\draw [->] (v1) -- (v7);
	\draw [->] (v1) -- (v9);
	\draw [->] (v1) -- (v14);
	\draw [->] (v1) -- (v15);
	\draw [->] (v1) -- (v17);
	\draw [->] (v1) -- (v22);
	\draw [->] (v2) -- (v5);
	\draw [->] (v2) -- (v8);
	\draw [->] (v2) -- (v10);
	\draw [->] (v2) -- (v16);
	\draw [->] (v2) -- (v18);
	\draw [->] (v2) -- (v23);
	\draw [->] (v3) -- (v5);
	\draw [->] (v3) -- (v11);
	\draw [->] (v3) -- (v19);
	\draw [->] (v3) -- (v24);
	\draw [->] (v4) -- (v5);
	\draw [->] (v4) -- (v12);
	\draw [->] (v4) -- (v20);
	\draw [->] (v4) -- (v25);
	\draw [->] (v5) -- (v13);
	\draw [->] (v5) -- (v21);
	\draw [->] (v5) -- (v26);
	\draw [->] (v6) -- (v8);
	\draw [->] (v6) -- (v11);
	\draw [->] (v6) -- (v27);
	\draw [->] (v7) -- (v8);
	\draw [->] (v7) -- (v12);
	\draw [->] (v7) -- (v28);
	\draw [->] (v8) -- (v13);
	\draw [->] (v8) -- (v29);
	\draw [->] (v9) -- (v10);
	\draw [->] (v9) -- (v11);
	\draw [->] (v9) -- (v12);
	\draw [->] (v9) -- (v30);
	\draw [->] (v10) -- (v13);
	\draw [->] (v10) -- (v31);
	\draw [->] (v11) -- (v13);
	\draw [->] (v11) -- (v32);
	\draw [->] (v12) -- (v13);
	\draw [->] (v12) -- (v33);
	\draw [->] (v13) -- (v34);
	\draw [->] (v14) -- (v16);
	\draw [->] (v14) -- (v19);
	\draw [->] (v14) -- (v27);
	\draw [->] (v15) -- (v16);
	\draw [->] (v15) -- (v20);
	\draw [->] (v15) -- (v28);
	\draw [->] (v16) -- (v21);
	\draw [->] (v16) -- (v29);
	\draw [->] (v17) -- (v18);
	\draw [->] (v17) -- (v19);
	\draw [->] (v17) -- (v20);
	\draw [->] (v17) -- (v30);
	\draw [->] (v18) -- (v21);
	\draw [->] (v18) -- (v31);
	\draw [->] (v19) -- (v21);
	\draw [->] (v19) -- (v32);
	\draw [->] (v20) -- (v21);
	\draw [->] (v20) -- (v33);
	\draw [->] (v21) -- (v34);
	\draw [->] (v22) -- (v23);
	\draw [->] (v22) -- (v24);
	\draw [->] (v22) -- (v25);
	\draw [->] (v22) -- (v27);
	\draw [->] (v22) -- (v28);
	\draw [->] (v22) -- (v30);
	\draw [->] (v23) -- (v26);
	\draw [->] (v23) -- (v29);
	\draw [->] (v23) -- (v31);
	\draw [->] (v24) -- (v26);
	\draw [->] (v24) -- (v32);
	\draw [->] (v25) -- (v26);
	\draw [->] (v25) -- (v33);
	\draw [->] (v26) -- (v34);
	\draw [->] (v27) -- (v29);
	\draw [->] (v27) -- (v32);
	\draw [->] (v28) -- (v29);
	\draw [->] (v28) -- (v33);
	\draw [->] (v29) -- (v34);
	\draw [->] (v30) -- (v31);
	\draw [->] (v30) -- (v32);
	\draw [->] (v30) -- (v33);
	\draw [->] (v31) -- (v34);
	\draw [->] (v32) -- (v34);
	\draw [->] (v33) -- (v34);
\end{tikzpicture}

%% file: bics_lattice_75.tex
\begin{tikzpicture}[scale=9]
	\node [n1_75_1bics] (v1) at (0.533333, 0.832000)	{1};
	\node [n2_75_1bics] (v2) at (0.096970, 0.624000)	{2};
	\node [n3_75_1bics] (v3) at (0.193939, 0.624000)	{3};
	\node [n4_75_1bics] (v4) at (0.290909, 0.624000)	{4};
	\node [n5_75_1bics] (v5) at (0.066667, 0.416000)	{5};
	\node [n6_75_1bics] (v6) at (0.387879, 0.624000)	{6};
	\node [n7_75_1bics] (v7) at (0.484848, 0.624000)	{7};
	\node [n8_75_1bics] (v8) at (0.133333, 0.416000)	{8};
	\node [n9_75_1bics] (v9) at (0.581818, 0.624000)	{9};
	\node [n10_75_1bics] (v10) at (0.200000, 0.416000)	{10};
	\node [n11_75_1bics] (v11) at (0.266667, 0.416000)	{11};
	\node [n12_75_1bics] (v12) at (0.333333, 0.416000)	{12};
	\node [n13_75_1bics] (v13) at (0.133333, 0.208000)	{$\begin{array}{@{}c@{}} 13 \\ (38\%) \end{array}$};
	\node [n14_75_1bics] (v14) at (0.678788, 0.624000)	{14};
	\node [n15_75_1bics] (v15) at (0.775758, 0.624000)	{15};
	\node [n16_75_1bics] (v16) at (0.400000, 0.416000)	{16};
	\node [n17_75_1bics] (v17) at (0.872727, 0.624000)	{17};
	\node [n18_75_1bics] (v18) at (0.466667, 0.416000)	{18};
	\node [n19_75_1bics] (v19) at (0.533333, 0.416000)	{19};
	\node [n20_75_1bics] (v20) at (0.600000, 0.416000)	{20};
	\node [n21_75_1bics] (v21) at (0.266667, 0.208000)	{21};
	\node [n22_75_1bics] (v22) at (0.969697, 0.624000)	{22};
	\node [n23_75_1bics] (v23) at (0.666667, 0.416000)	{23};
	\node [n24_75_1bics] (v24) at (0.733333, 0.416000)	{24};
	\node [n25_75_1bics] (v25) at (0.800000, 0.416000)	{25};
	\node [n26_75_1bics] (v26) at (0.400000, 0.208000)	{26};
	\node [n27_75_1bics] (v27) at (0.866667, 0.416000)	{27};
	\node [n28_75_1bics] (v28) at (0.933333, 0.416000)	{28};
	\node [n29_75_1bics] (v29) at (0.533333, 0.208000)	{29};
	\node [n30_75_1bics] (v30) at (1.000000, 0.416000)	{30};
	\node [n31_75_1bics] (v31) at (0.666667, 0.208000)	{31};
	\node [n32_75_1bics] (v32) at (0.800000, 0.208000)	{32};
	\node [n33_75_1bics] (v33) at (0.933333, 0.208000)	{33};
	\node [n34_75_1bics] (v34) at (0.533333, 0.000000)	{34};

	\draw [->] (v1) -- (v2);
	\draw [->] (v1) -- (v3);
	\draw [->] (v1) -- (v4);
	\draw [->] (v1) -- (v6);
	\draw [->] (v1) -- (v7);
	\draw [->] (v1) -- (v9);
	\draw [->] (v1) -- (v14);
	\draw [->] (v1) -- (v15);
	\draw [->] (v1) -- (v17);
	\draw [->] (v1) -- (v22);
	\draw [->] (v2) -- (v5);
	\draw [->] (v2) -- (v8);
	\draw [->] (v2) -- (v10);
	\draw [->] (v2) -- (v16);
	\draw [->] (v2) -- (v18);
	\draw [->] (v2) -- (v23);
	\draw [->] (v3) -- (v5);
	\draw [->] (v3) -- (v11);
	\draw [->] (v3) -- (v19);
	\draw [->] (v3) -- (v24);
	\draw [->] (v4) -- (v5);
	\draw [->] (v4) -- (v12);
	\draw [->] (v4) -- (v20);
	\draw [->] (v4) -- (v25);
	\draw [->] (v5) -- (v13);
	\draw [->] (v5) -- (v21);
	\draw [->] (v5) -- (v26);
	\draw [->] (v6) -- (v8);
	\draw [->] (v6) -- (v11);
	\draw [->] (v6) -- (v27);
	\draw [->] (v7) -- (v8);
	\draw [->] (v7) -- (v12);
	\draw [->] (v7) -- (v28);
	\draw [->] (v8) -- (v13);
	\draw [->] (v8) -- (v29);
	\draw [->] (v9) -- (v10);
	\draw [->] (v9) -- (v11);
	\draw [->] (v9) -- (v12);
	\draw [->] (v9) -- (v30);
	\draw [->] (v10) -- (v13);
	\draw [->] (v10) -- (v31);
	\draw [->] (v11) -- (v13);
	\draw [->] (v11) -- (v32);
	\draw [->] (v12) -- (v13);
	\draw [->] (v12) -- (v33);
	\draw [->] (v13) -- (v34);
	\draw [->] (v14) -- (v16);
	\draw [->] (v14) -- (v19);
	\draw [->] (v14) -- (v27);
	\draw [->] (v15) -- (v16);
	\draw [->] (v15) -- (v20);
	\draw [->] (v15) -- (v28);
	\draw [->] (v16) -- (v21);
	\draw [->] (v16) -- (v29);
	\draw [->] (v17) -- (v18);
	\draw [->] (v17) -- (v19);
	\draw [->] (v17) -- (v20);
	\draw [->] (v17) -- (v30);
	\draw [->] (v18) -- (v21);
	\draw [->] (v18) -- (v31);
	\draw [->] (v19) -- (v21);
	\draw [->] (v19) -- (v32);
	\draw [->] (v20) -- (v21);
	\draw [->] (v20) -- (v33);
	\draw [->] (v21) -- (v34);
	\draw [->] (v22) -- (v23);
	\draw [->] (v22) -- (v24);
	\draw [->] (v22) -- (v25);
	\draw [->] (v22) -- (v27);
	\draw [->] (v22) -- (v28);
	\draw [->] (v22) -- (v30);
	\draw [->] (v23) -- (v26);
	\draw [->] (v23) -- (v29);
	\draw [->] (v23) -- (v31);
	\draw [->] (v24) -- (v26);
	\draw [->] (v24) -- (v32);
	\draw [->] (v25) -- (v26);
	\draw [->] (v25) -- (v33);
	\draw [->] (v26) -- (v34);
	\draw [->] (v27) -- (v29);
	\draw [->] (v27) -- (v32);
	\draw [->] (v28) -- (v29);
	\draw [->] (v28) -- (v33);
	\draw [->] (v29) -- (v34);
	\draw [->] (v30) -- (v31);
	\draw [->] (v30) -- (v32);
	\draw [->] (v30) -- (v33);
	\draw [->] (v31) -- (v34);
	\draw [->] (v32) -- (v34);
	\draw [->] (v33) -- (v34);
\end{tikzpicture}

%% file: bic_lattice_75.tex
\begin{tikzpicture}[scale=9]
	\node [n1_75_1bic] (v1) at (0.533333, 0.832000)	{1};
	\node [n2_75_1bic] (v2) at (0.096970, 0.624000)	{2};
	\node [n3_75_1bic] (v3) at (0.193939, 0.624000)	{3};
	\node [n4_75_1bic] (v4) at (0.290909, 0.624000)	{4};
	\node [n5_75_1bic] (v5) at (0.066667, 0.416000)	{5};
	\node [n6_75_1bic] (v6) at (0.387879, 0.624000)	{6};
	\node [n7_75_1bic] (v7) at (0.484848, 0.624000)	{7};
	\node [n8_75_1bic] (v8) at (0.133333, 0.416000)	{8};
	\node [n9_75_1bic] (v9) at (0.581818, 0.624000)	{9};
	\node [n10_75_1bic] (v10) at (0.200000, 0.416000)	{10};
	\node [n11_75_1bic] (v11) at (0.266667, 0.416000)	{11};
	\node [n12_75_1bic] (v12) at (0.333333, 0.416000)	{12};
	\node [n13_75_1bic] (v13) at (0.133333, 0.208000)	{$\begin{array}{@{}c@{}} 13 \\ (5\%) \end{array}$};
	\node [n14_75_1bic] (v14) at (0.678788, 0.624000)	{14};
	\node [n15_75_1bic] (v15) at (0.775758, 0.624000)	{15};
	\node [n16_75_1bic] (v16) at (0.400000, 0.416000)	{16};
	\node [n17_75_1bic] (v17) at (0.872727, 0.624000)	{17};
	\node [n18_75_1bic] (v18) at (0.466667, 0.416000)	{18};
	\node [n19_75_1bic] (v19) at (0.533333, 0.416000)	{19};
	\node [n20_75_1bic] (v20) at (0.600000, 0.416000)	{20};
	\node [n21_75_1bic] (v21) at (0.266667, 0.208000)	{21};
	\node [n22_75_1bic] (v22) at (0.969697, 0.624000)	{22};
	\node [n23_75_1bic] (v23) at (0.666667, 0.416000)	{23};
	\node [n24_75_1bic] (v24) at (0.733333, 0.416000)	{24};
	\node [n25_75_1bic] (v25) at (0.800000, 0.416000)	{25};
	\node [n26_75_1bic] (v26) at (0.400000, 0.208000)	{26};
	\node [n27_75_1bic] (v27) at (0.866667, 0.416000)	{27};
	\node [n28_75_1bic] (v28) at (0.933333, 0.416000)	{28};
	\node [n29_75_1bic] (v29) at (0.533333, 0.208000)	{29};
	\node [n30_75_1bic] (v30) at (1.000000, 0.416000)	{30};
	\node [n31_75_1bic] (v31) at (0.666667, 0.208000)	{31};
	\node [n32_75_1bic] (v32) at (0.800000, 0.208000)	{32};
	\node [n33_75_1bic] (v33) at (0.933333, 0.208000)	{33};
	\node [n34_75_1bic] (v34) at (0.533333, 0.000000)	{34};

	\draw [->] (v1) -- (v2);
	\draw [->] (v1) -- (v3);
	\draw [->] (v1) -- (v4);
	\draw [->] (v1) -- (v6);
	\draw [->] (v1) -- (v7);
	\draw [->] (v1) -- (v9);
	\draw [->] (v1) -- (v14);
	\draw [->] (v1) -- (v15);
	\draw [->] (v1) -- (v17);
	\draw [->] (v1) -- (v22);
	\draw [->] (v2) -- (v5);
	\draw [->] (v2) -- (v8);
	\draw [->] (v2) -- (v10);
	\draw [->] (v2) -- (v16);
	\draw [->] (v2) -- (v18);
	\draw [->] (v2) -- (v23);
	\draw [->] (v3) -- (v5);
	\draw [->] (v3) -- (v11);
	\draw [->] (v3) -- (v19);
	\draw [->] (v3) -- (v24);
	\draw [->] (v4) -- (v5);
	\draw [->] (v4) -- (v12);
	\draw [->] (v4) -- (v20);
	\draw [->] (v4) -- (v25);
	\draw [->] (v5) -- (v13);
	\draw [->] (v5) -- (v21);
	\draw [->] (v5) -- (v26);
	\draw [->] (v6) -- (v8);
	\draw [->] (v6) -- (v11);
	\draw [->] (v6) -- (v27);
	\draw [->] (v7) -- (v8);
	\draw [->] (v7) -- (v12);
	\draw [->] (v7) -- (v28);
	\draw [->] (v8) -- (v13);
	\draw [->] (v8) -- (v29);
	\draw [->] (v9) -- (v10);
	\draw [->] (v9) -- (v11);
	\draw [->] (v9) -- (v12);
	\draw [->] (v9) -- (v30);
	\draw [->] (v10) -- (v13);
	\draw [->] (v10) -- (v31);
	\draw [->] (v11) -- (v13);
	\draw [->] (v11) -- (v32);
	\draw [->] (v12) -- (v13);
	\draw [->] (v12) -- (v33);
	\draw [->] (v13) -- (v34);
	\draw [->] (v14) -- (v16);
	\draw [->] (v14) -- (v19);
	\draw [->] (v14) -- (v27);
	\draw [->] (v15) -- (v16);
	\draw [->] (v15) -- (v20);
	\draw [->] (v15) -- (v28);
	\draw [->] (v16) -- (v21);
	\draw [->] (v16) -- (v29);
	\draw [->] (v17) -- (v18);
	\draw [->] (v17) -- (v19);
	\draw [->] (v17) -- (v20);
	\draw [->] (v17) -- (v30);
	\draw [->] (v18) -- (v21);
	\draw [->] (v18) -- (v31);
	\draw [->] (v19) -- (v21);
	\draw [->] (v19) -- (v32);
	\draw [->] (v20) -- (v21);
	\draw [->] (v20) -- (v33);
	\draw [->] (v21) -- (v34);
	\draw [->] (v22) -- (v23);
	\draw [->] (v22) -- (v24);
	\draw [->] (v22) -- (v25);
	\draw [->] (v22) -- (v27);
	\draw [->] (v22) -- (v28);
	\draw [->] (v22) -- (v30);
	\draw [->] (v23) -- (v26);
	\draw [->] (v23) -- (v29);
	\draw [->] (v23) -- (v31);
	\draw [->] (v24) -- (v26);
	\draw [->] (v24) -- (v32);
	\draw [->] (v25) -- (v26);
	\draw [->] (v25) -- (v33);
	\draw [->] (v26) -- (v34);
	\draw [->] (v27) -- (v29);
	\draw [->] (v27) -- (v32);
	\draw [->] (v28) -- (v29);
	\draw [->] (v28) -- (v33);
	\draw [->] (v29) -- (v34);
	\draw [->] (v30) -- (v31);
	\draw [->] (v30) -- (v32);
	\draw [->] (v30) -- (v33);
	\draw [->] (v31) -- (v34);
	\draw [->] (v32) -- (v34);
	\draw [->] (v33) -- (v34);
\end{tikzpicture}

%% file: bics_lattice_125.tex
\begin{tikzpicture}[scale=9]
	\node [n1_125_1bics] (v1) at (0.533333, 0.832000)	{1};
	\node [n2_125_1bics] (v2) at (0.096970, 0.624000)	{2};
	\node [n3_125_1bics] (v3) at (0.193939, 0.624000)	{3};
	\node [n4_125_1bics] (v4) at (0.290909, 0.624000)	{4};
	\node [n5_125_1bics] (v5) at (0.066667, 0.416000)	{5};
	\node [n6_125_1bics] (v6) at (0.387879, 0.624000)	{6};
	\node [n7_125_1bics] (v7) at (0.484848, 0.624000)	{7};
	\node [n8_125_1bics] (v8) at (0.133333, 0.416000)	{8};
	\node [n9_125_1bics] (v9) at (0.581818, 0.624000)	{9};
	\node [n10_125_1bics] (v10) at (0.200000, 0.416000)	{10};
	\node [n11_125_1bics] (v11) at (0.266667, 0.416000)	{11};
	\node [n12_125_1bics] (v12) at (0.333333, 0.416000)	{12};
	\node [n13_125_1bics] (v13) at (0.133333, 0.208000)	{$\begin{array}{@{}c@{}} 13 \\ (58\%) \end{array}$};
	\node [n14_125_1bics] (v14) at (0.678788, 0.624000)	{14};
	\node [n15_125_1bics] (v15) at (0.775758, 0.624000)	{15};
	\node [n16_125_1bics] (v16) at (0.400000, 0.416000)	{16};
	\node [n17_125_1bics] (v17) at (0.872727, 0.624000)	{17};
	\node [n18_125_1bics] (v18) at (0.466667, 0.416000)	{18};
	\node [n19_125_1bics] (v19) at (0.533333, 0.416000)	{19};
	\node [n20_125_1bics] (v20) at (0.600000, 0.416000)	{20};
	\node [n21_125_1bics] (v21) at (0.266667, 0.208000)	{21};
	\node [n22_125_1bics] (v22) at (0.969697, 0.624000)	{22};
	\node [n23_125_1bics] (v23) at (0.666667, 0.416000)	{23};
	\node [n24_125_1bics] (v24) at (0.733333, 0.416000)	{24};
	\node [n25_125_1bics] (v25) at (0.800000, 0.416000)	{25};
	\node [n26_125_1bics] (v26) at (0.400000, 0.208000)	{26};
	\node [n27_125_1bics] (v27) at (0.866667, 0.416000)	{27};
	\node [n28_125_1bics] (v28) at (0.933333, 0.416000)	{28};
	\node [n29_125_1bics] (v29) at (0.533333, 0.208000)	{29};
	\node [n30_125_1bics] (v30) at (1.000000, 0.416000)	{30};
	\node [n31_125_1bics] (v31) at (0.666667, 0.208000)	{31};
	\node [n32_125_1bics] (v32) at (0.800000, 0.208000)	{32};
	\node [n33_125_1bics] (v33) at (0.933333, 0.208000)	{33};
	\node [n34_125_1bics] (v34) at (0.533333, 0.000000)	{34};

	\draw [->] (v1) -- (v2);
	\draw [->] (v1) -- (v3);
	\draw [->] (v1) -- (v4);
	\draw [->] (v1) -- (v6);
	\draw [->] (v1) -- (v7);
	\draw [->] (v1) -- (v9);
	\draw [->] (v1) -- (v14);
	\draw [->] (v1) -- (v15);
	\draw [->] (v1) -- (v17);
	\draw [->] (v1) -- (v22);
	\draw [->] (v2) -- (v5);
	\draw [->] (v2) -- (v8);
	\draw [->] (v2) -- (v10);
	\draw [->] (v2) -- (v16);
	\draw [->] (v2) -- (v18);
	\draw [->] (v2) -- (v23);
	\draw [->] (v3) -- (v5);
	\draw [->] (v3) -- (v11);
	\draw [->] (v3) -- (v19);
	\draw [->] (v3) -- (v24);
	\draw [->] (v4) -- (v5);
	\draw [->] (v4) -- (v12);
	\draw [->] (v4) -- (v20);
	\draw [->] (v4) -- (v25);
	\draw [->] (v5) -- (v13);
	\draw [->] (v5) -- (v21);
	\draw [->] (v5) -- (v26);
	\draw [->] (v6) -- (v8);
	\draw [->] (v6) -- (v11);
	\draw [->] (v6) -- (v27);
	\draw [->] (v7) -- (v8);
	\draw [->] (v7) -- (v12);
	\draw [->] (v7) -- (v28);
	\draw [->] (v8) -- (v13);
	\draw [->] (v8) -- (v29);
	\draw [->] (v9) -- (v10);
	\draw [->] (v9) -- (v11);
	\draw [->] (v9) -- (v12);
	\draw [->] (v9) -- (v30);
	\draw [->] (v10) -- (v13);
	\draw [->] (v10) -- (v31);
	\draw [->] (v11) -- (v13);
	\draw [->] (v11) -- (v32);
	\draw [->] (v12) -- (v13);
	\draw [->] (v12) -- (v33);
	\draw [->] (v13) -- (v34);
	\draw [->] (v14) -- (v16);
	\draw [->] (v14) -- (v19);
	\draw [->] (v14) -- (v27);
	\draw [->] (v15) -- (v16);
	\draw [->] (v15) -- (v20);
	\draw [->] (v15) -- (v28);
	\draw [->] (v16) -- (v21);
	\draw [->] (v16) -- (v29);
	\draw [->] (v17) -- (v18);
	\draw [->] (v17) -- (v19);
	\draw [->] (v17) -- (v20);
	\draw [->] (v17) -- (v30);
	\draw [->] (v18) -- (v21);
	\draw [->] (v18) -- (v31);
	\draw [->] (v19) -- (v21);
	\draw [->] (v19) -- (v32);
	\draw [->] (v20) -- (v21);
	\draw [->] (v20) -- (v33);
	\draw [->] (v21) -- (v34);
	\draw [->] (v22) -- (v23);
	\draw [->] (v22) -- (v24);
	\draw [->] (v22) -- (v25);
	\draw [->] (v22) -- (v27);
	\draw [->] (v22) -- (v28);
	\draw [->] (v22) -- (v30);
	\draw [->] (v23) -- (v26);
	\draw [->] (v23) -- (v29);
	\draw [->] (v23) -- (v31);
	\draw [->] (v24) -- (v26);
	\draw [->] (v24) -- (v32);
	\draw [->] (v25) -- (v26);
	\draw [->] (v25) -- (v33);
	\draw [->] (v26) -- (v34);
	\draw [->] (v27) -- (v29);
	\draw [->] (v27) -- (v32);
	\draw [->] (v28) -- (v29);
	\draw [->] (v28) -- (v33);
	\draw [->] (v29) -- (v34);
	\draw [->] (v30) -- (v31);
	\draw [->] (v30) -- (v32);
	\draw [->] (v30) -- (v33);
	\draw [->] (v31) -- (v34);
	\draw [->] (v32) -- (v34);
	\draw [->] (v33) -- (v34);
\end{tikzpicture}

%% file: bic_lattice_125.tex
\begin{tikzpicture}[scale=9]
	\node [n1_125_1bic] (v1) at (0.533333, 0.832000)	{1};
	\node [n2_125_1bic] (v2) at (0.096970, 0.624000)	{2};
	\node [n3_125_1bic] (v3) at (0.193939, 0.624000)	{3};
	\node [n4_125_1bic] (v4) at (0.290909, 0.624000)	{4};
	\node [n5_125_1bic] (v5) at (0.066667, 0.416000)	{5};
	\node [n6_125_1bic] (v6) at (0.387879, 0.624000)	{6};
	\node [n7_125_1bic] (v7) at (0.484848, 0.624000)	{7};
	\node [n8_125_1bic] (v8) at (0.133333, 0.416000)	{8};
	\node [n9_125_1bic] (v9) at (0.581818, 0.624000)	{9};
	\node [n10_125_1bic] (v10) at (0.200000, 0.416000)	{10};
	\node [n11_125_1bic] (v11) at (0.266667, 0.416000)	{11};
	\node [n12_125_1bic] (v12) at (0.333333, 0.416000)	{12};
	\node [n13_125_1bic] (v13) at (0.133333, 0.208000)	{$\begin{array}{@{}c@{}} 13 \\ (10\%) \end{array}$};
	\node [n14_125_1bic] (v14) at (0.678788, 0.624000)	{14};
	\node [n15_125_1bic] (v15) at (0.775758, 0.624000)	{15};
	\node [n16_125_1bic] (v16) at (0.400000, 0.416000)	{16};
	\node [n17_125_1bic] (v17) at (0.872727, 0.624000)	{17};
	\node [n18_125_1bic] (v18) at (0.466667, 0.416000)	{18};
	\node [n19_125_1bic] (v19) at (0.533333, 0.416000)	{19};
	\node [n20_125_1bic] (v20) at (0.600000, 0.416000)	{20};
	\node [n21_125_1bic] (v21) at (0.266667, 0.208000)	{21};
	\node [n22_125_1bic] (v22) at (0.969697, 0.624000)	{22};
	\node [n23_125_1bic] (v23) at (0.666667, 0.416000)	{23};
	\node [n24_125_1bic] (v24) at (0.733333, 0.416000)	{24};
	\node [n25_125_1bic] (v25) at (0.800000, 0.416000)	{25};
	\node [n26_125_1bic] (v26) at (0.400000, 0.208000)	{26};
	\node [n27_125_1bic] (v27) at (0.866667, 0.416000)	{27};
	\node [n28_125_1bic] (v28) at (0.933333, 0.416000)	{28};
	\node [n29_125_1bic] (v29) at (0.533333, 0.208000)	{29};
	\node [n30_125_1bic] (v30) at (1.000000, 0.416000)	{30};
	\node [n31_125_1bic] (v31) at (0.666667, 0.208000)	{31};
	\node [n32_125_1bic] (v32) at (0.800000, 0.208000)	{32};
	\node [n33_125_1bic] (v33) at (0.933333, 0.208000)	{33};
	\node [n34_125_1bic] (v34) at (0.533333, 0.000000)	{34};

	\draw [->] (v1) -- (v2);
	\draw [->] (v1) -- (v3);
	\draw [->] (v1) -- (v4);
	\draw [->] (v1) -- (v6);
	\draw [->] (v1) -- (v7);
	\draw [->] (v1) -- (v9);
	\draw [->] (v1) -- (v14);
	\draw [->] (v1) -- (v15);
	\draw [->] (v1) -- (v17);
	\draw [->] (v1) -- (v22);
	\draw [->] (v2) -- (v5);
	\draw [->] (v2) -- (v8);
	\draw [->] (v2) -- (v10);
	\draw [->] (v2) -- (v16);
	\draw [->] (v2) -- (v18);
	\draw [->] (v2) -- (v23);
	\draw [->] (v3) -- (v5);
	\draw [->] (v3) -- (v11);
	\draw [->] (v3) -- (v19);
	\draw [->] (v3) -- (v24);
	\draw [->] (v4) -- (v5);
	\draw [->] (v4) -- (v12);
	\draw [->] (v4) -- (v20);
	\draw [->] (v4) -- (v25);
	\draw [->] (v5) -- (v13);
	\draw [->] (v5) -- (v21);
	\draw [->] (v5) -- (v26);
	\draw [->] (v6) -- (v8);
	\draw [->] (v6) -- (v11);
	\draw [->] (v6) -- (v27);
	\draw [->] (v7) -- (v8);
	\draw [->] (v7) -- (v12);
	\draw [->] (v7) -- (v28);
	\draw [->] (v8) -- (v13);
	\draw [->] (v8) -- (v29);
	\draw [->] (v9) -- (v10);
	\draw [->] (v9) -- (v11);
	\draw [->] (v9) -- (v12);
	\draw [->] (v9) -- (v30);
	\draw [->] (v10) -- (v13);
	\draw [->] (v10) -- (v31);
	\draw [->] (v11) -- (v13);
	\draw [->] (v11) -- (v32);
	\draw [->] (v12) -- (v13);
	\draw [->] (v12) -- (v33);
	\draw [->] (v13) -- (v34);
	\draw [->] (v14) -- (v16);
	\draw [->] (v14) -- (v19);
	\draw [->] (v14) -- (v27);
	\draw [->] (v15) -- (v16);
	\draw [->] (v15) -- (v20);
	\draw [->] (v15) -- (v28);
	\draw [->] (v16) -- (v21);
	\draw [->] (v16) -- (v29);
	\draw [->] (v17) -- (v18);
	\draw [->] (v17) -- (v19);
	\draw [->] (v17) -- (v20);
	\draw [->] (v17) -- (v30);
	\draw [->] (v18) -- (v21);
	\draw [->] (v18) -- (v31);
	\draw [->] (v19) -- (v21);
	\draw [->] (v19) -- (v32);
	\draw [->] (v20) -- (v21);
	\draw [->] (v20) -- (v33);
	\draw [->] (v21) -- (v34);
	\draw [->] (v22) -- (v23);
	\draw [->] (v22) -- (v24);
	\draw [->] (v22) -- (v25);
	\draw [->] (v22) -- (v27);
	\draw [->] (v22) -- (v28);
	\draw [->] (v22) -- (v30);
	\draw [->] (v23) -- (v26);
	\draw [->] (v23) -- (v29);
	\draw [->] (v23) -- (v31);
	\draw [->] (v24) -- (v26);
	\draw [->] (v24) -- (v32);
	\draw [->] (v25) -- (v26);
	\draw [->] (v25) -- (v33);
	\draw [->] (v26) -- (v34);
	\draw [->] (v27) -- (v29);
	\draw [->] (v27) -- (v32);
	\draw [->] (v28) -- (v29);
	\draw [->] (v28) -- (v33);
	\draw [->] (v29) -- (v34);
	\draw [->] (v30) -- (v31);
	\draw [->] (v30) -- (v32);
	\draw [->] (v30) -- (v33);
	\draw [->] (v31) -- (v34);
	\draw [->] (v32) -- (v34);
	\draw [->] (v33) -- (v34);
\end{tikzpicture}